%% file: main.tex
\documentclass[11pt]{article}
\usepackage[margin=1in,letterpaper]{geometry}
\usepackage{libertine}
\usepackage[rightcaption]{sidecap}

\usepackage{microtype} 
\usepackage{flushend}
\usepackage[framemethod=TikZ]{mdframed}
\usepackage{graphicx} 
\usepackage{amsthm}
\usepackage{complexity}
\usepackage{xcolor}
\usepackage{tabularx}
\usepackage{fullpage}
\usepackage[T1]{fontenc}
\usepackage{subcaption}
\usepackage{amsfonts}
\usepackage{amsmath}
\usepackage{amssymb}
\usepackage{enumitem}

\usepackage{thmtools}
\usepackage{thm-restate}

\usepackage{mleftright,xparse}
\usepackage{stmaryrd}
\usepackage{mathtools}
\usepackage[hidelinks]{hyperref}
\usepackage{amsfonts}
\usepackage[capitalize,noabbrev]{cleveref}
\usepackage{bm}
\usepackage{svg}

\usepackage{tikz}
\usetikzlibrary{shapes.multipart}
\usepackage{csquotes}
\usepackage{verbatim}

\makeatletter
\def\pgftransformlineattime#1#2#3{%
	\pgfgettransformentries\aa\ab\ba\bb\notimportantx\notimportanty
	\pgf@process{#2}%
	\pgf@xb=\pgf@x
	\pgf@yb=\pgf@y%
	\pgf@process{#3}%
	\pgf@xc=\pgf@x
	\pgf@yc=\pgf@y%
	\pgftransformshift{\pgfpointlineattime{#1}{\pgfqpoint{\pgf@xb}{\pgf@yb}}{\pgfqpoint{\pgf@xc}{\pgf@yc}}}%
	\ifpgfresetnontranslationattime%
	\pgftransformresetnontranslations%
	\fi%
	\ifpgfslopedattime%
	\advance\pgf@xc by-\pgf@xb%
	\advance\pgf@yc by-\pgf@yb%
	{
		\pgfsettransformentries\aa\ab\ba\bb{0pt}{0pt}
		\pgf@pos@transform{\pgf@xc}{\pgf@yc}
		\global\pgf@xc\pgf@xc
		\global\pgf@yc\pgf@yc
	}
	\ifpgfallowupsidedownattime%
	\else%
	\ifdim\pgf@xc<0pt%
	\pgf@xc=-\pgf@xc%
	\pgf@yc=-\pgf@yc%
	\fi%
	\fi%
	\pgf@x=\pgf@xc%
	\pgf@y=\pgf@yc%
	\pgfpointnormalised{}
	\pgf@ya=-\pgf@y%
	\pgftransformcm%
	{\pgf@sys@tonumber{\pgf@x}}{\pgf@sys@tonumber{\pgf@y}}%
	{\pgf@sys@tonumber{\pgf@ya}}{\pgf@sys@tonumber{\pgf@x}}{\pgfpointorigin}%
	\fi%
}
\makeatother

\usetikzlibrary{automata,positioning,arrows.meta,arrows,decorations.pathmorphing,backgrounds,calc}

\tikzstyle{oscillate} = [decorate, decoration={snake, amplitude=.3mm, segment length=3mm, post length=2mm}]
	
\usepackage{xspace}
\newcommand{\eda}{\ensuremath{\textup{EDA}}\xspace}
\newcommand{\ida}[1][]{\ensuremath{\textup{IDA}_{#1}}\xspace}
\newcommand{\unambiguity}{\ensuremath{\textup{\textsc{Unambiguity}}}\xspace}
\newcommand{\unaryunambiguity}{{\textup{\textsc{Unary Unambiguity}}}\xspace}
\newcommand{\ov}{\textup{\textsc{Orthogonal}}\-\textup{\textsc{Vector}}\xspace}
\newcommand{\twoie}{\textup{\textsc{2-IE}}\xspace}
\newcommand{\threeie}{\textup{\textsc{3-IE}}\xspace}
\newcommand{\kie}{\textup{\textsc{$k$-IE}}\xspace}
\newcommand{\twins}{\textup{\textsc{Twins}}\xspace}
\newcommand{\unarytwins}{{\textup{\textsc{Unary Twins}}}\xspace}
\newcommand{\FFA}{{\textup{\textsc{Finite Ambiguity}}}\xspace}
\newcommand{\PFA}{{\textup{\textsc{Polynomial Ambiguity}}}\xspace}
\newcommand{\determinisability}{\textup{\textsc{Unambiguous Determinisability}}\xspace}
\newcommand{\disjointprogressions}{\textup{\textsc{Disjoint Progressions}}\xspace}

\usepackage[textwidth=1.9cm,textsize=small
,disable
]{todonotes}

\newcommand{\karol}[1]{\todo[color=cyan!10]{{\bf Karol:} #1}}

\newcommand{\Dd}{\mathcal{D}}

\newcommand{\Oh}{\mathcal{O}}
\DeclarePairedDelimiter{\iverson}{\llbracket}{\rrbracket}
\DeclarePairedDelimiter{\floor}{\lfloor}{\rfloor}

\newcommand{\set}[1]{\{#1\}}
\newcommand{\N}{\mathbb{N}}
\newcommand{\Z}{\mathbb{Z}}

\renewcommand{\vec}[1]{{\bm #1}}

\newcommand{\runs}{\mathsf{Runs}}

\newcommand{\reach}[1][*]{\xrightarrow{#1}}

\newtheorem{definition}{Definition}[section]
\newtheorem{theorem}[definition]{Theorem}
\newtheorem{lemma}[definition]{Lemma}

\newtheorem{claim}[definition]{Claim}
\newtheorem{remark}[definition]{Remark}

\newtheorem{observation}[definition]{Observation}

\newtheorem{hypothesis}[definition]{Hypothesis}

\newcommand{\rng}[2]{[#1\textsf{..}#2]}

\mdfdefinestyle{MyFrame}{%
    linecolor=black,
    middlelinewidth=1pt,
    outerlinewidth=0pt,
    roundcorner=5pt,
    innertopmargin=0pt,
    innerbottommargin=2pt,
    innerrightmargin=2pt,
    innerleftmargin=2pt,
	leftmargin = 2pt,
	rightmargin = 2pt
}
\newcommand{\weight}{\omega}

\Crefname{lemma}{Lemma}{Lemmas}
\crefname{thm}{theorem}{theorems}
\Crefname{thm}{Theorem}{Theorems}
\crefname{clm}{claim}{claims}
\Crefname{clm}{Claim}{Claims}
\Crefname{claim}{Claim}{Claims}
\Crefname{observation}{Observation}{Observations}

\newcommand{\Ot}{\widetilde{\mathcal O}}
\newcommand{\mc}[1]{\mathcal{#1}}

\newcommand{\divides}{\mid}

\renewcommand{\subset}{\subseteq}
\renewcommand{\le}{\leqslant}
\renewcommand{\leq}{\leqslant}
\renewcommand{\ge}{\geqslant}
\renewcommand{\geq}{\geqslant}

\newcommand{\eps}{\varepsilon}

\definecolor{cb_blue}{RGB}{220,038,127}
\definecolor{cb_green}{RGB}{100,143,255}
\definecolor{cb_red}{RGB}{255,176,0}

\begin{document}

\title{Fine-Grained Complexity of Ambiguity Problems on Automata and Directed Graphs
}

\date{}


\author{
	Karolina Drabik\footnote{University of Warsaw, \textsf{\{k.drabik, f.mazowiecki\}@mimuw.edu.pl}
	}
	\and
	Anita D\"urr\footnote{Saarland University and Max Planck Institute for Informatics,
        Saarbr\"ucken, Germany, \textsf{\{aduerr,~kwegrzyc\}@mpi-inf.mpg.de}}
	\and 
	Fabian Frei\footnote{
		CISPA Helmholtz Center for Information Security, \textsf{fabian.frei@cispa.de}
	}
    \and
    Filip Mazowiecki\footnotemark[1]
    \and
    Karol W\k{e}grzycki\footnotemark[2]
}

\maketitle

\thispagestyle{empty}
\begin{abstract}
In the field of computational logic, two classes of finite automata are considered fundamental: deterministic and
nondeterministic automata (DFAs and NFAs).
In a more fine-grained approach three natural
intermediate classes were introduced, defined by restricting the number of accepting runs of the input NFA. The classes are called: unambiguous, finitely ambiguous, and polynomially ambiguous finite automata.
It was observed that central problems, like equivalence, become tractable when the input NFA is restricted to some of these classes. This naturally brought interest into problems determining whether an input NFA belongs to the intermediate classes.

Our first result is a nearly complete characterization of the fine-grained complexity of these problems. 
We show that the respective quadratic and cubic running times of Allauzen et al.
are optimal under the Orthogonal Vectors hypothesis or the $k$-Cycle hypothesis,
for alphabets with at least two symbols. In contrast, for unary alphabets we
show that all aforementioned variants of ambiguity can be decided in almost linear time.

Finally, we study determinisability of unambiguous weighted automata. We positively resolve a conjecture of Allauzen and Mohri, proving that their quadratic-time algorithm for verifying determinisability of unambiguous weighted automata is optimal, assuming the Orthogonal Vectors hypothesis or the $k$-Cycle hypothesis. We additionally show that for unary alphabets, this can be decided in linear time.
\end{abstract}

\paragraph{Acknowledgements.}
This work is part of the project TIPEA that has received funding from the European Research Council (ERC) under the European Union's Horizon 2020 research and innovation programme (grant agreement No 850979).
This work is part of the project CORDIS that has received funding from the European Research Council (ERC) under the European Union's Horizon 2023 research and innovation programme (grant agreement No 101126229).




\clearpage
\setcounter{page}{1}
\section{Introduction}\label{sec:introduction}
\input{intro}

\section{Technical overview}\label{sec:technical_overview}
\input{technical}


\section{Preliminaries}\label{sec:preliminaries}
\input{preliminaries}
\section{Unary Unambiguity}\label{sec:unambiguity_upper}
\input{UU.tex}





\section{Unary Unambiguous Determinisability}\label{sec:twins_unary}
\input{twins}

\section{Lower bounds}\label{sec:lower}
\input{lower_bound}

\bibliographystyle{plainurl}
\bibliography{lit.bib}


\appendix

\input{appendix_UU}

\input{wildcard}

\end{document}

%% file: intro.tex
Deterministic and nondeterministic finite automata (DFAs and NFAs) are
fundamental models of computation in the field of computational logic. Both recognize regular languages, but each representation has its own advantages. On the one hand, it is well known that NFAs can
be exponentially more succinct, i.e.,\ given an NFA, an equivalent DFA may
require exponentially more states~\cite{MeyerF71}. On the other
hand, many natural operations such as complementation or minimisation can be
performed efficiently for DFAs,  which enables us to decide problems like containment or equivalence in polynomial time for DFAs, whereas they are PSPACE-complete for NFAs~\cite{StockmeyerM73}.

These exponential gaps fuelled research aimed at understanding intermediate computation models: more succinct than DFAs, but not as intractable as NFAs. The concept of \emph{ambiguity}, which measures the number of accepting runs for input words~\cite{BookEGO71}, allows us to define several such intermediate automata classes~\cite{RavikumarI89}.
The most prominent and natural one is the class of \emph{unambiguous finite automata} (UFA), where every input word has at most one accepting run. It is straightforward to see that they can be exponentially more succinct than DFAs~\cite{MeyerF71}: languages like $(0+1)^*1(0+1)^n$ can be recognised by UFAs with linearly many states in $n$ (see \cref{fig:nfa-examples-a}) but require exponentially many states for DFAs. Interestingly, the equivalence and containment problems for UFAs remain polynomial-time  solvable~\cite{StearnsH85}. 

Other notable ambiguity-based classes include finitely ambiguous and polynomially ambiguous finite automata (FFAs and PFAs), where the number of accepting runs on a word is bounded by, respectively, a constant and a polynomial in the word length; see \cref{fig:ambiguity} for examples.
\begin{figure*}[b]
	\begin{subfigure}{.30\textwidth}
		\resizebox{\textwidth}{!}{%
		\begin{tikzpicture}[state/.append style={minimum size=2.3em},thick,>={Stealth[length=2.4mm, width=1.2mm]}]
			\node[state,initial,initial where = left,initial text=$ $] (p1) {$p_1$};
			\node[state,below = 1cm of p1] (p2) {$p_2$};
			\node[state,right = .8cm of p2] (p3) {$p_3$};
			\node[state,right = .8cm of p3,accepting] (p4) {$p_4$};
			
			\path
			(p1) edge[loop right] node[right] {$0,1$} (p1)
			(p1) edge[->] node[left] {$1$} (p2)
			(p2) edge[->] node[above] {$0,1$} (p3)
			(p3) edge[->] node[above] {$0,1$} (p4)
			;
		\end{tikzpicture}
		}
        \caption{\small{Unambiguous NFA $\mc A_1$.}}
		\label{fig:nfa-examples-a}
	\end{subfigure}
	\hfill
	\begin{subfigure}{.30\textwidth}
		\centering
		\resizebox{0.95\textwidth}{!}{%
	\begin{tikzpicture}[state/.append style={minimum size=2.3em},thick,>={Stealth[length=2.4mm, width=1.2mm]}]
		\node[state,accepting,right = 1cm of p4] (q2) {$q_2$};
		\node[state,initial,initial where = above,initial text=$ $,right = 1cm of q2] (q1) {$q_1$};
		\node[state,accepting,right = 1cm of q1] (q3) {$q_3$};
		
		\path
		(q1) edge[->,bend left] node[above] {$0$} (q2)
		(q1) edge[->,bend right] node[above] {$0$} (q3)
		(q2) edge[loop above] node[above] {$0$} (q2)
		(q3) edge[loop above] node[above] {$1$} (q3)
		(q2) edge[->,dashed,bend left] node[above] {$\#$} (q1)
		(q3) edge[->,dashed,bend right] node[above] {$\#$} (q1)
		;
	\end{tikzpicture}
		}
	\caption{\small{Exponentially ambiguous NFA $\mc A_2$ with and finitely
    ambiguous NFA $\mc A_3$ without the dashed transitions.}}
	\label{fig:nfa-examples-b}
\end{subfigure}
\hfill
	\begin{subfigure}{.30\textwidth}
	\centering
	\resizebox{\textwidth}{!}{%
	\begin{tikzpicture}[state/.append style={minimum size=2.3em},thick,>={Stealth[length=2.4mm, width=1.2mm]}]
		\node[state,initial,initial where = left,initial text= $ $,right = 1cm of q3] (r1) {$r_1$};
		\node[state,above = 1cm of r1] (r2) {$r_2$};
		\node[state,accepting,right = 1cm of r1] (r3) {$r_3$};
		\node[state,above right = 0.5cm and 1cm of r3] (r4) {$r_4$};
		\node[state,above = 1cm of r3] (r5) {$r_5$};
		
		\path
		(r1) edge[->,bend right] node[right] {$1$} (r2)
		(r2) edge[->,bend right] node[left] {$1$} (r1)
		(r1) edge[->] node[above] {$1$} (r3)
		(r3) edge[->,bend right] node[below right] {$1$} (r4)
		(r4) edge[->,bend right] node[above right] {$1$} (r5)
		(r5) edge[->,bend right] node[left] {$1$} (r3)
		;
	\end{tikzpicture}
	}
    \caption{\small{Polynomially ambiguous NFA~$\mc A_4$.}}
	\label{fig:nfa-examples-c}
\end{subfigure}
	\caption{
		\small{The unambiguous NFA ${\mc A}_1$ 
		recognises $(0+1)^*1(0+1)^2$. 
		The exponentially ambiguous NFA ${\mc A}_2$ 
		accepts every word $(0\#)^{n-1}0$ with $2^n$ runs. 
		The finitely ambiguous NFA ${\mc A}_3$ 
		recognises $01^* + 00^*$: with two accepting runs on $0$, and at most one otherwise. 
		The unary polynomially ambiguous NFA ${\mc A}_4$ 
		accepts $1^n$ with less than $n$ runs and, if $n \equiv_6 1$, more than $n/6$ runs.
    }
        }\label{fig:ambiguity}
		\label{fig:nfa-examples}
\end{figure*}
%
%
To emphasize that an NFA is not polynomially ambiguous, we say that it is exponentially ambiguous. 
Thus, we obtain the following hierarchy of five natural classes from DFAs up to NFAs using ambiguity:
\begin{center}
DFA $\subseteq$ UFA $\subseteq$ FFA $\subseteq$ PFA $\subseteq$ NFA
\end{center}
%
%
This induces the following natural problems: 
\begin{description}
    \item[\unambiguity] Given an NFA $\mc A$, decide whether every word has at most one accepting run.
    \item[\FFA] Given an NFA $\mc A$, decide whether every word has at most $c$ accepting runs, where $c$ is a constant depending on $\mc A$. 
    \item[\PFA] Given an NFA $\mc A$, decide whether every word $w$ has at most $p(|w|)$ accepting runs, where $p$ is a polynomial depending on $\mc A$. 
\end{description}
It is known that \unambiguity, \FFA and \PFA can all be decided in polynomial time~\cite{WeberS91}. 
Building on this result, all three problems have been studied in more detail. Allauzen et al.~\cite{AllauzenMR11} showed that \FFA can be verified in cubic time, while \unambiguity and \PFA can be verified in quadratic time. The authors conjectured this to be optimal: ``\emph{We presented simple and eﬃcient algorithms for testing the ﬁnite, polynomial, or exponential ambiguity of ﬁnite automata~[...]. We conjecture that the time complexity of our algorithms is optimal.}''

Our first contribution is to show that this is indeed optimal when assuming popular fine-grained complexity hypotheses.
\begin{theorem}\label{thm:NFA_lowerbounds}
	Assuming the Orthogonal Vectors hypothesis or the $k$-Cycle hypothesis:
	\begin{enumerate}[label=(\alph*)]
		\item There is no algorithm algorithm deciding \unambiguity in $\Oh(|\mc A|^{2- \eps})$ time for any $\eps >0$; 
		\item There is no algorithm algorithm deciding \PFA in $\Oh(|\mc A|^{2- \eps})$ time for any $\eps >0$. 
	\end{enumerate}
	Assuming the 3-Orthogonal Vectors hypothesis:
	\begin{enumerate}[label=(\alph*)]
		\item[(c)] There is no algorithm deciding \FFA in $\Oh(|\mc A|^{3- \eps})$ time for any $\eps >0$.
	\end{enumerate}
\end{theorem}

The ambiguity notion for NFA naturally extends to \emph{weighted automata}. These are finite automata whose transitions are additionally labelled with weights. 
Weighted automata assign weights to words; and one can think that finite automata are weighted automata that assign $0$ and $1$ to words (interpreted as reject and accept).
While in general the weights can come from any fixed semiring, for the purpose of this paper it is sufficient to think of integer weights. The weight of a run is the sum of weights on its transitions; to obtain the weight of the word the weights of runs are aggregated with $\min$. Readers familiar with the model will recognise the tropical semiring $\Z(\min,+)$. For more details we refer to the standard handbook~\cite{droste2009handbook}.
%
Weighted automata, in contrast to finite automata, become more expressive when nondeterminism is allowed. A central
question is the \emph{determinisation problem}: given a weighted automaton, decide whether there is an equivalent deterministic weighted automaton~\cite{LombardyS06}. For the tropical semiring the problem remained open for at least~30 years~\cite{Mohri94,Mohri97}, and was proven decidable only recently in 2025~\cite{abs-2503-23826}. However, much earlier, decidability was established successively for: unambiguous~\cite{Mohri97}, finitely ambiguous~\cite{KlimannLMP04}, and polynomially ambiguous~\cite{KirstenL09} weighted automata.
The particular case of unambiguous weighted automata was further shown to be polynomial-time decidable~\cite{buchsbaum2000determinization} and finally quadratic-time decidable~\cite{Allauzen2003EfficientAF}.
The authors of this upper bound conjecture that it is optimal: ``\emph{We present a new and substantially more efficient algorithm for testing [determinisability], which we conjecture to be optimal.}''

In this paper, we consider the following problem.
\begin{description}
	\item[\determinisability] Given an unambiguous weighted automaton $\mc W$, decide whether there is an equivalent deterministic weighted automaton.
\end{description}
We answer the conjecture in~\cite{Allauzen2003EfficientAF} positively, assuming popular fine-grained hypotheses.
\begin{theorem}\label{thm:WFA_lowerbound}
	Assuming the Orthogonal Vectors hypothesis or the $k$-Cycle hypothesis, there is no algorithm deciding \determinisability in $\Oh(|\mc W|^{2-\eps})$ time for any $\eps >0$.
\end{theorem}
%
\cref{thm:NFA_lowerbounds,thm:WFA_lowerbound} show that the optimality conjectures hold in the case of all alphabets with at least two letters. For the special case of unary alphabets (i.e.,\ the restriction to $1$-letter alphabets), we show that the conjectured quadratic and cubic lower bounds do not hold, and (almost) linear time algorithms exist.
%
\begin{theorem}\label{thm:NFA_unary}
	Given a unary NFA $\mc A$:
	\begin{enumerate}[label=(\alph*)]
		\item \unambiguity can be decided in $|\mc A|^{1+ o(1)}$ time. 
		\item \FFA and \PFA can be decided in $\Oh(|\mc A|)$ time.
	\end{enumerate}
\end{theorem}
\begin{theorem}\label{thm:WFA_unary}
	Given a unary unambiguous weighted automaton $\mc W$, \determinisability can be decided in $\Oh(|\mc W|)$ time. 
\end{theorem}
The study of unary automata is motivated by their simplicity, which allows to capture interesting phenomena in automata theory. Two important examples are: (1) the celebrated Parikh's theorem, which in its simplified form states that regular and context-free languages coincide for unary languages~\cite{parikh1961language}; (2) Chrobak’s influential paper on efficient representations of unary NFAs~\cite{Chrobak86}, later corrected in~\cite{To09}. The importance of the former result does not need to be discussed, the latter found many applications, e.g.\ in database theory~\cite{BarceloLLW12}.
Furthermore, unary (weighted) automata can be seen as (edge-weighted) directed graphs and the studied problems boil down to simple graph algorithmic questions. For example \unambiguity asks \emph{are there two different source-to-target walks of the same length?}, while \determinisability asks \emph{are there two cycles of different average weight?}

The proof of \cref{thm:NFA_unary}\;(a) is nontrivial and the technically most involved contribution of this paper. In fact, we conjecture that the algorithm actually works in $\Ot(|\mc A|)$ time (i.e.\ linear up to polylogarithmic factors). However, we leave this tighter analysis as an open problem. Interestingly, the complexity depends only on the number of states when we assume $\mc A$ to be trim, i.e.\ every state lies on an accepting run.
This relies on the observation that in a graph with $n$ nodes and $\Omega(n)$ edges there always exist two different walks of the same length.
Note that this observation only applies to unary automata, as for larger alphabets there are trim unambiguous automata with quadratically (in number of states) many transitions (see, e.g.~\cite{KieferW19}).


Curiously, the running time of our algorithm reduces to a number theoretical question involving \emph{greatest common divisors}
(GCDs): given natural numbers $0 < x_1 < x_2 < \ldots < x_k$, how big is
$\sum_{i,j} \gcd(x_i,x_j)$
in terms of $n = \sum_{i}x_i$?
Equivalently, this can be seen as bounding the entry-wise one-norm of a GCD matrix -- the $k\times k$ matrix whose $(i,j)$-th entry is $\gcd(x_i,x_j)$. (Such GCD matrices have been studied since the 19th century; see e.g.~\cite{bourque1992gcd}.) 
We provide a short proof that $\sum_{i,j} \gcd(x_i,x_j) \le n^{1 + o(1)}$. We remark that this already follows from more general results
(e.g.,~\cite{bondarenko2015gcd,walker2020gcd}), but our arguments are elementary.

\subsection{Related work}
The classes UFA, FFA and PFA have been studied extensively since the 80s (see e.g.\ a survey~\cite{HanSS17}). For example, it took two decades to show that for any two of these classes there is a language witnessing a superpolynomial increase in state complexity~\cite{RavikumarI89,Leung98,HromkovicS11}. For unary automata (i.e., those with a $1$-letter alphabet) a complete picture was known already in the 80s~\cite{RavikumarI89}: the state complexity increases only by a polynomial factor between the classes FFA, PFA and NFA, whereas the factor between UFA and FFA can be exponential.

Further research has been devoted to unambiguous automata for models beyond finite automata (see e.g.\ the survey~\cite{Colcombet15}). More recently, unambiguous and finitely ambiguous restrictions have been studied in the context of vector addition systems and register automata~\cite{CzerwinskiFH20,CzerwinskiMQ21,CzerwinskiH22}. It should also be noted that some operations like complementation or union have been recently shown to be inefficient already for UFAs; there are languages for which such operations cause a superpolynomial increase in the number of required states~\cite{Raskin18,GoosK022}.

The fine-grained complexity of problems in automata theory has been investigated in various contexts. Our paper builds on the lower bounds for intersection of DFAs~\cite{WeharThesis16}. Other works regarding fine-grained complexity investigate one-counter systems~\cite{ChistikovCHPW19}, vector addition systems~\cite{KunnemannMSSW23}, Dyck reachability~\cite{ChatterjeeCP18,KoutrisD23,BringmannGKL24}, polynomial growth of transducers~\cite{abs-2501-10270} and NFA universality~\cite{FernauK17}.

Regarding unary automata, lower bounds on the size of DFAs equivalent to unary UFAs were analysed in~\cite{Okhotin12}. The already mentioned superpolynomial lower bound on the state complexity of complementing UFAs was proven for unary automata~\cite{Raskin18}. 
The containment problem for unary NFAs and the union operation for unary UFAs were also analysed in~\cite{CzerwinskiDGH0S23}. Unary weighted automata over fields can be naturally interpreted as linear recurrence sequences (LRSs). Bounded ambiguity subclasses have natural characterisations in terms of restricted LRSs~\cite{BarloyFLM22}.

Questions about two different walks of the same length have been studied for simple paths (i.e.\ walks that do not revisit vertices).
Recently, a linear time algorithm was given for detecting two simple paths assuming that they are disjoint and shortest~\cite{AkmalWW24}. We refer to the introduction of~\cite{AkmalWW24} for a broader discussion on the problem when the simple paths are disjoint.

\subsection{Organisation}
In \cref{sec:technical_overview} we give a high-level overview of our results before introducing all notations and terminology in \cref{sec:preliminaries}. 
In particular, in \cref{sec:technical_overview_lowerbounds} we explain the chain of reductions used to prove the lower bounds of \cref{thm:NFA_lowerbounds,thm:WFA_lowerbound}. The actual reductions are given in \cref{sec:lower}. In \cref{sec:technical_overview_unary} we discuss the unary setting of \cref{thm:NFA_unary,thm:WFA_unary}, with a main focus on checking unambiguity of a unary NFA, i.e. \cref{thm:NFA_unary}\;(a). The full proof of \cref{thm:NFA_unary}\;(a) is given in \cref{sec:unambiguity_upper} while \cref{thm:NFA_unary}\;(b) follows from a simple observation that we prove in \cref{sec:technical_overview_unary}. 
Finally, the full proof of \cref{thm:WFA_unary} is provided in \cref{sec:twins_unary}.

%% file: technical.tex
We now discuss our results in more details. Precise definitions and terminology are given in \cref{sec:preliminaries}, and we assume for now that the reader is familiar with notions of (weighted) finite automata and directed graphs. 
%
For simplicity, in this section, we assume that automata are in normal form, i.e.\ there is a single initial state $s$ and a single accepting state $t$, and trim, i.e.\ every state is reachable from $s$ and co-reachable from $t$.

\subsection{Lower bounds}\label{sec:technical_overview_lowerbounds}

We discuss the conditional lower bounds proven in \cref{thm:NFA_lowerbounds,thm:WFA_lowerbound} for the above four problems. In fact, the last three problems were shown in~\cite{WeberS91,AllauzenMR11} to be equivalent to deciding whether the given automaton satisfies some structural property called respectively \eda, \ida, and \twins.  
We define those properties below and illustrate them in \cref{fig:idaeda}. 
By abuse of notation, we also denote by \eda, \ida and \twins the problems of deciding whether a given NFA or unambiguous weighted automata has the respective properties. 
We note that the \ida property can be extended to characterise finer degrees of ambiguity (e.g. quadratic or cubic, see~\cite{WeberS91}). Remark that the \twins property was first introduced in the context of determinisation of transducers, a model closely related to weighted automata~\cite{choffrut1977caracterisation}.

\begin{definition}[{See~\cite{WeberS91} and~\cite[Def.~2]{AllauzenMR11}}]
    Let $\mathcal{A}$ be an NFA over the alphabet $\Sigma$. We say that $\mc A$ has the \emph{\eda property}, respectively \emph{\ida property}\footnote{\eda and \ida stand for \emph{Exponential} and \emph{Infinite Degree of Ambiguity}.}, if the following holds. 
    \begin{description}
    \item[\eda] There is a state $q$ and two distinct cycles at $q$ on the same non-empty word $v\in\Sigma^\ast$.
    \item[\ida] There are states $p\neq q$, a non-empty word $v\in\Sigma^\ast$, and three runs on $v$: one from $p$ to $p$, one from $p$ to $q$, and one from $q$ to $q$.
    \end{description}
\end{definition}
\begin{definition}[\cite{choffrut1977caracterisation}]\label{def:twins}
    Let $\mc W$ be a weighted automaton over the alphabet $\Sigma$. Two states $p$ and $q$ are siblings if there are initial states $i_p$ and $i_q$ and words $u, v \in \Sigma^*$ such that there are runs from $i_p$ to $p$ and from $i_q$ to $q$ labelled by $u$, and runs $\pi_p$ from $p$ to $p$ and $\pi_q$ from $q$ to $q$ labelled by $v$. 
    Two siblings are \emph{twins} if every pair of cycles $\pi_q$ and $\pi_q$ with equal labels also have equal weight.
    Then the \twins property is defined as follows.
    \begin{description}
        \item[\twins] Every pair of siblings is a pair of twins.
    \end{description}
\end{definition}
  
\begin{figure}
    \centering
    \begin{subfigure}[t]{.24\textwidth}
        \centering
        \begin{tikzpicture}
            [state/.append style={minimum size=2.3em},thick, >={Stealth[length=2.4mm, width=1.2mm]}]
            \node[state] (e) {$q$};
            \draw
            (e) edge[loop left] node[above=5pt,xshift=5pt] {$v$} (e)
            (e) edge[loop right] node[above=5pt,xshift=-5pt] {$v$} (e)
            ;		
        \end{tikzpicture}
        \caption{A witness for \eda.}\label{fig:idaeda-a}
    \end{subfigure}
    \hfill
    \begin{subfigure}[t]{.2\textwidth}
        \centering
        \begin{tikzpicture}
            [state/.append style={minimum size=2.3em},thick, >={Stealth[length=2.4mm, width=1.2mm]}]
            \node[state] (p) {$p$};
            \node[state, right = .75cm of p] (q) {$q$};
            
            \draw
            (p) edge[loop left] node[above=5pt,xshift=5pt] {$v$} (p)
            (q) edge[loop right] node[above=5pt,xshift=-5pt] {$v$} (q)
            (p) edge[->] node[above] {$v$} (q)
            ;
        \end{tikzpicture}
        \caption{A witness for \ida.}\label{fig:idaeda-b}
    \end{subfigure}
    \hfill
    \begin{subfigure}[t]{.45\textwidth}
        \centering
        \resizebox{\textwidth}{!}{
        \begin{tikzpicture}
            [state/.append style={minimum size=2.3em},thick, >={Stealth[length=2.4mm, width=1.2mm]}]
            \node[state] (rp) {$i_p$};
            \node[state, right = 1.25cm of rp] (p) {$p$};
            \node[state, right = 1cm of p] (rq) {$i_q$};
            \node[state, right = 1.25cm of rq] (q) {$q$};
            
            \draw
            (rp) edge[->] node[above] {$u \mid x_{1}$} (p)
            (rq) edge[->] node[above] {$u \mid x_{2}$} (q)
            (p) edge[loop above] node[above] {$v \mid y_{1}$} (p)
            (q) edge[loop above] node[above] {$v \mid y_{2}$} (q)
            ;
        \end{tikzpicture}
        }
        \caption{Two siblings. They violate the \twins property if $y_1 \neq y_2$.}\label{fig:idaeda-c}
    \end{subfigure}
    %
    %
    %
    %
    %
    \caption{
    Illustration of three properties. An arrow labelled by $u$, respectively $u \mid x$, represents a run over the word $u$, respectively of weight $x$.
    }\label{fig:idaeda}
\end{figure}

The crucial observations are the following equivalences.
\begin{lemma}[{See~\cite{WeberS91} and~\cite[Lemma~3]{AllauzenMR11}}]\label{lem:NFA_equivalence}
Let $\mathcal{A}$ be a trim NFA. Then we have that
\begin{itemize}
 \item $\mc A$ is not finitely ambiguous if and only if $\mc A$ has the \ida property;
 \item $\mc A$ is not polynomially ambiguous if and only if $\mc A$ has the \eda property.
\end{itemize}
\end{lemma}
\begin{lemma}[{\cite{Mohri97}}]\label{lem:WFA_equivalence}
    Let $\mc W$ be an unambiguous weighted automaton. Then $\mc W$ is determinisable if and only if $\mc W$ has the \twins property. 
\end{lemma}
As an example, consider the exponentially and polynomially ambiguous automata ${\mc A}_{2}$ and ${\mc A}_4$ in \cref{fig:ambiguity}. \eda holds in ${\mc A}_{2}$ as witnessed by the state $q_1$ and the word $0\#$. Indeed, there are two different runs from $q_1$ back to $q_1$ on $0\#$. \ida holds in ${\mc A}_4$ as witnessed by the states $r_1$, $r_4$ and the word $1^6$. Indeed, there is a $r_1$-cycle on $1^6$, a $r_1r_4$-walk on $1^6$, and $r_4$-cycle on $1^6$.

To establish the lower bounds stated in \cref{thm:NFA_lowerbounds,thm:WFA_lowerbound}, we consider the following intermediate problem on deterministic finite automata. By abuse of notation, we say that two DFAs intersect if the intersection of their languages is non-empty.

\begin{restatable}[\kie]{definition}{defkie}
    \label{def:kie}
    Given $k$ DFAs $\mc D_1, \mc D_2, \dots, \mc D_k$, the problem \kie asks whether the intersection of their languages $\bigcap_{i=1}^k\mc L(\mc D_i)$ is empty.
\end{restatable}

In his doctoral thesis~\cite{WeharThesis16}, Wehar showed that \kie cannot be solved in $\Oh(n^{k-\eps})$ time for any constant $\eps>0$, where $n = |\mc D_1| + \dots + |\mc D_k|$, assuming the Strong Exponential-Time Hypothesis (SETH). To prove the lower bounds of \cref{thm:NFA_lowerbounds,thm:WFA_lowerbound}, we show the following chain of reductions illustrated in \cref{fig:graph_reduction}. 
\begin{figure*}
    \centering
        \resizebox{\textwidth}{!}{%
                \begin{tikzpicture}[yscale=0.825,
                    thick,
                    >={Stealth[length=2.4mm, width=1.2mm]},
                    every text node part/.style={align=center}]
                    \node[draw, ultra thick] at (-3,-2.75) (CYCLE) {$k$-Cycle};
                    \node[draw, ultra thick] at (1, 0) (OV) {$k$-OV};
                    \node[draw] at (4,-2.75) (2DFA) {\twoie};
                    \node[draw] at (8,-2.75) (UA) {\unambiguity};
                    \node[rectangle, draw,anchor=west] at (12, -1.75) (EDA) {\textsc{Polynomial} \\ \textsc{Ambiguity}};
                    \node[rectangle, draw,anchor=west] at (12, -3.75) (TP) {\textsc{Unambiguous} \\ \textsc{Determinisability}};
    
                    \draw[dotted,very thick] (2.8,0)  -- (15.8, 0) node[below left] {$n^2$} node[above left] {$n^3$};
                    \node[draw] at (4, 1.5) (3DFA) {\threeie};
                    \node[draw,anchor=west] at (12, 1.5) (IDA) {\textsc{Finite} \\ \textsc{Ambiguity}};
    
                    \draw (CYCLE)   edge[->,sloped] node[below] {\cref{thm:kcycle_2dfa}} (2DFA);
                    \draw (2DFA)  	edge[->] node[above] {\cref{thm:twoie_unambiguous}} (UA);
                    \draw (3DFA)  	edge[->] node[above] {\cref{thm:3dfa_ida}} (IDA);
                    \draw (UA)  	edge[->,sloped] node[above] {\cref{thm:unambiguity_eda}} (EDA);
                    \draw (OV)  	edge[->,sloped] node[above] {\cref{thm:ov_twoie}} (2DFA) ;
                    \draw (OV)  	edge[->,sloped] node[above] {\cref{thm:ov_twoie}} (3DFA) ;
                    \draw (UA)  	edge[->,sloped] node[below] {\cref{thm:unambiguity_twins}} (TP);
    
                    \node[draw, ultra thick] at (-3,0) (SAT) {SAT};
    
                    \draw (SAT)   	edge[->,dashed] node[above] {\cite[Lem.~A.1]{WilliamsY14}} (OV);
                    \draw (SAT)   	edge[->,dashed,bend left,sloped,out=40] node[above,rotate=2,pos=.72] {\cite[Thm.~7.22]{WeharThesis16}} (3DFA);
                    \draw (SAT)   	edge[->,dashed,sloped] node[below] {\cite[Thm.~7.22]{WeharThesis16}} (2DFA);
                \end{tikzpicture}
    }
            \caption{\small{Summary of reductions.
            Dashed arrows are previously known reductions, the remaining arrows are
            our contribution. The problems in thick boxes are conjectured to be hard. The dotted line separates problems by their conditionally optimal complexity: cubic above and quadratic below. 
            The picture is simplified: we reduce with \cref{thm:ov_twoie} from $k$-OV to $k$-IE for every $k$. The $k$ of $k$-Cycle is a constant independent of the $2$ in $2$-IE. \cref{thm:3dfa_ida,thm:twoie_unambiguous,thm:unambiguity_eda} are reductions to \eda, \ida and \twins respectively, which are equivalent to \PFA, \FFA and \determinisability by \cref{lem:NFA_equivalence,lem:WFA_equivalence}.
    }}\label{fig:graph_reduction}
\end{figure*}

First, to make the hardness assumption of \kie more believable, we adapt Wehar's
lower bound~\cite{WeharThesis16} for \kie to also be based on the $k$-OV
hypothesis. Then, we show a quadratic lower bound for the 2-IE problem assuming the $k$-Cycle hypothesis. As it is unknown how $k$-Cycle compares to $k$-OV or SETH, this reinforced the hardness assumption of 2-IE.

Next, observe that $2$-IE naturally reduces to \unambiguity: two DFAs accept
the same word if and only if the NFA constructed as their union is ambiguous.
The reduction from \unambiguity to \eda is also straightforward: it suffices to
extend the alphabet with a special letter $\#$ and add a transition from the
final state to the initial state that reads $\#$. Then, the automaton accepts
two different words $w_1$ and $w_2$ if and only if the word $w_1\#w_2\#$ is a witness of the \eda
property.
The reduction from $3$-IE to \ida generalizes this observation to three
automata. Suppose $w$ is accepted by all automata. We add two states, $p$
and $q$, and connect the DFAs $\Dd_1$, $\Dd_2$, and $\Dd_3$ such that $\Dd_1$
accepts $w$ on a loop from $p$ to $p$, $\Dd_2$ accepts $w$ on a walk from $p$ to
$q$, and $\Dd_3$ accepts $w$ on a loop from $q$ to $q$.

The reduction from \unambiguity to \twins is different, as we need to
distinguish between two accepting runs of the word $w$ with different
weights. The idea is to create three copies of the automaton and place two layers
of edges between their states. This allows us to assign weights such that two
distinct accepting runs result in different total weights.

Finally, to reduce the alphabet $\Sigma$ to $\{0, 1\}$, it suffices to encode each
letter using $\Oh(\log{|\Sigma|})$ transitions. This introduces only a logarithmic
blow-up, which is good enough for our purposes. This binary encoding is also necessary to render some of the reductions efficient.


\subsection{Unary Automata}\label{sec:technical_overview_unary}

In the unary case, it is natural to think of a finite automaton as a directed graph $G$ with two distinguished vertices: a source $s \in V(G)$ and a target $t \in V(G)$. 
If the automaton is additionally weighted, then we can associate a corresponding edge weight function to the graph. 
We stress that a \emph{walk} is a sequence of adjacent vertices that can repeat, as opposed to a \emph{simple path} where each vertex is unique. 

We observe that, if $G$ is trim, i.e.\ every vertex lies on a source-to-target walk, then one can assume that $G$ is sparse (i.e.\ has only $\Oh(n)$ edges). The proof is deferred to \cref{sec:proofs_Sec2}. 
This allows us to only consider the number of states (vertices) as a parameter and to ignore the number of transitions (edges). In particular, we will show that if the given NFA in \cref{thm:NFA_unary} is trim, then \unambiguity can be decided in $n^{1 + o(1)}$, where $n$ is the number of states.
In the following, we therefore we assume for simplicity that $G$ is trim.

\begin{claim}\label{lem:sparsity}
    Let $\mc A$ be a trim unary NFA with $n$ states and $m$ transitions. If $\mc A$ is unambiguous, then $m \leq 12n$.
\end{claim}




\paragraph{Towards proving~\cref{thm:NFA_unary}}
We first focus on the almost linear time algorithm of \cref{thm:NFA_unary}\;(a) that checks unambiguity of a unary NFA. The algorithms for unary \FFA and \PFA are quite simple and are given at the end of this section.
Observe that, in unary NFAs, a word with an accepting run is characterized by the length of the corresponding source-to-target walk in the underlying directed graph. Thus, the problem of \unambiguity for unary NFAs is equivalent to the following natural question.
\begin{description}
    \item \centering{\emph{Given a directed graph, are there two different source-to-target walks of the same length?}\footnote{Note that requiring disjointness of edges or vertices of these walks would render the problem NP-hard~\cite{Vygen95}.}}
\end{description}
An important observation is that if the given graph is acyclic, i.e.\ a DAG, then we can answer the above question in linear time. This is not an obvious fact as
the natural dynamic program that inductively computes the set of lengths of all $vt$-paths for each vertex $v$ may require quadratic time; see \cref{fig:counterexample}.
\begin{figure*}
    \centering
    \begin{tikzpicture}[
            thick,
        >=Stealth,
        every node/.style={circle, draw}
      ]
      
      \node (v1) at (0,0) {$v_1$};
      \node (v2) [right=of v1] {$v_2$};
      \node (v3) [right=of v2] {$v_3$};
      \node (v4) [right=of v3] {$v_4$};
      \node (v5) [right=of v4] {$v_5$}; 
      
      \draw[->] (v1) -- (v2);
      \draw[->] (v2) -- (v3);
      \draw[->] (v3) -- (v4);
      \draw[->] (v4) -- (v5);
      
      \draw[->] (v1.north) to[bend left=45] (v5);
      \draw[->] (v2.north) to[bend left=35] (v5);
      \draw[->] (v3.north) to[bend left=25] (v5);
      
      \end{tikzpicture}
    \caption{\small{
            Graph $G$ with $s = v_1$, $t = v_n$, and edges $(v_i, v_{i+1})$ and $(v_i, v_n)$ for $i < n$ (shown for $n = 5$). A naive dynamic program computing all $st$-walk lengths also computes all $vt$-walk lengths, leading to a quadratic blow-up. To avoid this, we observe that $v_{i+1}t$-walks can be succinctly represented via $v_it$-walks extended by 1, plus a new walk of length 1.
}}\label{fig:counterexample}
\end{figure*}
We show the linear time algorithm for DAGs in~\cref{sec:unambiguity_upper}. In the general case, we exploit this algorithm on acyclic graphs obtained from the input graph (essentially subgraphs, but slightly modified). This requires understanding the influence of cycles on the length of source-to-target walks. In fact, we relate this to the following number theoretic problem.
We call the sequence $(a+x\cdot b)_{x \in\mathbb{N}}$ the \emph{arithmetic
progression} with \emph{base} $a\in\mathbb{N}$ and \emph{step} $b\in\mathbb{N}$. 
\begin{description}
    \item[\disjointprogressions] Given pairwise distinct steps $1\le
        b_1<\dots<b_k$, and sets of bases $A_i\subseteq\{0,\dots,b_i-1\}$ for each $i \in [k]$, decide if the given arithmetic progressions $(a + x \cdot b_i)_{x \in \mathbb N}$ for $a \in A_i$ and $i \in [k]$ are disjoint. 
\end{description}
Note that the complement is asking if there are $i\neq j$ and $a_i\in A_i, a_j\in A_j, x_i,x_j\in\mathbb{N}$ such that $a_i+x_ib_i=a_j+x_jb_j$?
The time complexity of \disjointprogressions is measured in terms of $n = b_1 + \dots + b_k$.
We show that deciding whether a unary automaton is unambiguous is essentially equivalent to deciding \disjointprogressions. A quadratic reduction can be obtained using Chrobak's normal form \cite[Theorem~1]{To09}, but in our application it can be done more efficiently.

\begin{restatable}{proposition}{unaryproblemequivalence}
    \label{thm:equivalence}
    The \unambiguity problem for unary NFAs and the \disjointprogressions problem are linear time equivalent.
\end{restatable}
\input{disjoint_prog}

\cref{sec:unambiguity_upper} is dedicated to prove the above equivalence.  See \cref{fig:disjoint_unambiguity} for an intuition.
We now discuss how to solve \disjointprogressions. 
Consider the following characterisation of disjointness of arithmetic progressions (proven in \cref{sec:proofs_Sec2}).

\begin{claim}\label{claim:gcd2}
    For any integers $a, b, c, d >0$, the arithmetic progressions $(a + x \cdot b)_{x \in \N}$ and $(c+ x\cdot d)_{x \in \N}$ are disjoint if and only if 
    $(a \bmod p) \neq (c \bmod p)$ where $p = \gcd(b, d)$.
\end{claim}

\cref{claim:gcd2} yields the following equivalent formulation of the \disjointprogressions question: Do we have $(a_i \bmod p) \neq (a_j \bmod p)$, where $p = \gcd(b_i, b_j)$, for all $i \neq j$, $a_i \in A_i$, $a_j \in A_j$?
To find collisions the idea is to compute for every distinct pair $i, j \in \rng{1}{k}$ the set $A_i^{(p)}\coloneqq\{a_i \bmod p\mid a_i \in A_i, p = \gcd(b_i, b_j)\}$. Then it suffices to check disjointness between the sets $A_i^{(p)}$ and $A_j^{(p)}$.
In fact, for every $i \in \rng{1}{k}$ we compute the set $A_i^{(d)} = \set{a_i \bmod d \mid a_i \in A_i}$ for \emph{every divisor $d$ of $b_i$}.
Then we have, in particular, computed the sets $A_i^{(\gcd(b_i, b_j))}$ for
all $j \in \rng{1}{k}$. The following lemma shows that this can be done efficiently.

\begin{lemma}\label{lem:precompute}
    Let $N >0$ be an integer and consider a set $A \subseteq \rng{0}{N-1}$. 
    In time $\Oh(N \log^2{N})$, we can construct all sets
	$
        A^{(d)} \coloneqq \left\{ a \bmod d \mid a \in A \right\} \subseteq
        \rng{0}{d-1},
    $
    for every divisor $d \in \N$ of $N$. 
\end{lemma}


\begin{proof}
    Compute the prime factorisation $N = \prod_{i=1}^k p_i^{\alpha_i}$, where $p_i$ are distinct primes and $\alpha_i \in \N\setminus\{0\}$, e.g.\ using the
    sieve of Eratosthenes, which works in $\Ot(\sqrt{N})$ time.
    All the divisors of $N$ are of the form $\prod_{i=1}^k p_i^{\beta_i}$, where $0 \le \beta_i \le \alpha_i$. 
    Thus, we can identify them with vectors $\vec{\beta} = (\beta_1,\ldots,\beta_{k}) \in \N^k$ such that $\beta_i\le \alpha_i$ for all $i\in\rng 1k$. In particular, let $\vec{\alpha} = (\alpha_1,\ldots,\alpha_k)$.
    We denote by $d(\vec{\beta})$ the divisor of $N$ corresponding to $\vec{\beta}$. 
    Let $\vec{\beta}, \vec{\gamma} \in \N^k$ be two such vectors. We say that $\vec{\gamma}$ is a child of $\vec{\beta}$ if $\vec{\beta} - \vec{\gamma} = \vec{e}_i$ 
	for some $i \in \rng{1}{k}$, where $\vec{e}_i$ is the unit vector of dimension $k$ with value $1$ in the $i$-th entry and zero otherwise.
    Note in particular that $d(\vec{\gamma}) \mid d(\vec{\beta})$.

    Further observe that, except for $\vec{\alpha}$, every divisor is a child of at least one other divisor. Hence, we can compute $A^{(d)}$ for all divisors $d$ of $N$ recursively as we now describe. 
    For $d = N$, we have $A^{(N)} = A$.
    Suppose that we have computed the set $A^{(d(\vec{\beta}))}$ for some $\vec{\beta} \leq \vec{\alpha}$. Let $\vec{\gamma} \in \mathbb N^k$ be a child of $\vec{\beta}$. For any $x, y, z \in \N$ such that $x \mid y$, we have $z \bmod x = (z \bmod y) \bmod x$. Since $d(\vec{\gamma}) \mid d(\vec{\beta})$, we can thus compute $A^{(d(\vec{\gamma}))}$ by considering every $z \in A^{(d(\vec{\beta}))}$ and computing the remainder modulo $d(\vec{\gamma})$. This takes $\Oh(|A^{(d(\vec{\beta}))}|) = \Oh(d(\vec{\beta}))$ time. 
    Since every divisor has at most $k = \Oh(\log N)$ children, all sets $A^{(d)}$ for $d$ a divisor of $N$ are computed in time bounded by 
    \begin{displaymath}
        \Ot(\sqrt{N}) + \Oh \bigg(k \cdot \sum_{\vec{\beta} \le \vec{\alpha}} d(\vec{\beta})
        \bigg) = \Oh(\sigma_1(N) \log{N})
    \end{displaymath}
    where $\sigma_1$ is the standard \emph{sum of divisor function} $\sigma_1(n) \coloneq \sum_{d=1}^n d \cdot \iverson{d \mid n}$. 
    In \cref{sec:proofs_Sec2}, we provide the following bound\footnote{The asymptotics of $\sigma_1$ have been studied in detail, e.g.\ in the context of the Riemann hypothesis, with tighter bounds (essentially $\Oh(n \log\log n)$, see~\cite{hardy1979introduction}). For our purposes, the easier bound in \cref{lemma:sumofdiv} is sufficient.
    }.
    \begin{claim}\label{lemma:sumofdiv}
            We have
            $\sigma_1(n) = \Oh(n \log{n})$.
    \end{claim}
    Thus, we can bound the above running time by $\Oh(N \log^2{N})$.
\end{proof}

To solve the \disjointprogressions problem, it suffices to precompute the set of
bases $A_i$ using \cref{lem:precompute} and then compute the intersection of the obtained sets.

\begin{lemma}\label{thm:ap-gcd-alg}
    There is an algorithm deciding \disjointprogressions
    in time 
    \begin{displaymath}
        \Oh\bigg( n \log^2{n} + \sum_{1 \leq i,j \leq k} \gcd(b_i,b_j) \bigg)
    \end{displaymath}
    where $n = b_1 + \dots + b_k$ is the sum of the given steps.
\end{lemma}

\begin{proof}
    Consider an instance of \disjointprogressions $(b_1, A_1), \dots, (b_k, A_k)$ where $1 \leq b_1 < \dots < b_k$, $A_i \subset \{0, \dots, b_i-1\}$ for each $i \in [k]$ and let $n \coloneq b_1 + \dots + b_k$. 
    For each $i\in\rng{1}{k}$, we apply \cref{lem:precompute} to each set $A_i \subset \rng{0}{b_i-1}$ to obtain the set $A_i^{(d)}$ for all divisors $d$ of $b_i$ in $\Oh(\sum_{i=1}^k b_i \log^2{b_i}) = \Oh(n \log^2{n})$ time. 

    Observe that we constructed in particular the sets $A_i^{(p)}$ with $p = \gcd(b_i, b_j)$ for all $i, j \in \rng{1}{k}$.
    By \cref{claim:gcd2}, in order to decide \disjointprogressions, it suffices to verify whether for each pair of distinct $i, j \in \rng{1}{k}$ we have $A_i^{(p)} \cap A_j^{(p)} = \emptyset$ for $p = \gcd(b_i, b_j)$. For any pair, this can be done in time $\Oh(|A_i^{(p)}| + |A_j^{(p)}|) = \Oh(p) = \Oh(\gcd(b_i, b_j))$. 
    This bounds the total running time by
    \begin{displaymath}
        \Oh\Bigg(n\log^2{n} + \sum_{1 \leq i,j \leq k} \gcd(b_i,b_j)\Bigg).\qedhere
    \end{displaymath}
\end{proof}

We show that the algorithm in \cref{thm:ap-gcd-alg} runs in time $n^{1 + o(1)}$ by bounding the sum of greatest common divisors in \cref{lem:sum-gcd}. 

\begin{restatable}{lemma}{sumgcd}
    \label{lem:sum-gcd}
    Let $A$ be a set of positive integers whose sum is equal to
    $N$. Then we have
    \begin{displaymath}
        \sum_{a,b \in A} \gcd(a,b) \le N^{1+o(1)}.
    \end{displaymath}
\end{restatable}
We are not aware of any explicit bounds on the sum of~\cref{lem:sum-gcd} in the literature. We have found some related results that can be used to obtain slightly better bounds (see~\cref{rem:galsums}). However, these are technically involved; we provide an elementary proof of \cref{lem:sum-gcd} in \cref{sec:proofs_Sec2}.
To prove our conjecture that our algorithm works in time $\Ot(n)$ it suffices to improve this bound to $\sum_{i,j} \gcd(x_i,x_j) = \Ot(n)$.
It is known that the answer is positive in the special case of $x_i = i$~\cite{haukkanen2004lp}, but in the general case we leave this as an open problem.

\begin{remark}\label{rem:galsums}
	The factor $N^{o(1)}$ in~\cref{lem:sum-gcd} is 
	$N^{\Oh(1/\log{\log{N}})}$. 
    The bound
	in~\cref{lem:sum-gcd} can be improved to 
	\begin{equation*}
		\sum_{a,b \in A} \gcd(a,b) \le N \exp\left( \Oh\left( \sqrt{\frac{\log N
		\log\log\log N}{\log\log N}}\right)\right)
	\end{equation*}
	by using bounds on the weighted G\'al sums (e.g. \cite[Theorem
	1]{walker2020gcd} and \cite{bondarenko2015gcd}). 
\end{remark}

Finally, we can prove \cref{thm:NFA_unary} using all above stated results. Note that the core of the algorithm lies on the equivalence with the \disjointprogressions problem~\cref{thm:equivalence}, which we show in \cref{sec:unambiguity_upper}. All other statements are proven in \cref{sec:proofs_Sec2}.

\begin{proof}[Proof of \cref{thm:NFA_unary}]
    Let $\mc A$ be a unary NFA with $n$ states and $m$ transitions, and denote by $G$ the underlying directed graph. 
    We can preprocess $G$ in $\Oh(n+m)$ time such that it is trim (see~\cref{claim:trim}) and has a single source $s \in V(G)$ and single target $t \in V(G)$ (see~\cref{claim:normal}). 
    By \cref{thm:equivalence} and \cref{lem:sparsity} we can, in $\Oh(n)$ time, either decide whether
    $G$ is unambiguous, or construct an instance $(b_1, A_1), \dots, (b_k, A_k)$ of \disjointprogressions with $b_1 + \dots + b_k \leq n$ such that $G$ is not unambiguous if and only if the arithmetic progressions are disjoint. 
    By \cref{lem:sum-gcd,thm:ap-gcd-alg}, the latter is decidable in $n^{1 + o(1)}$ time.
    Note that if we assume trimness, then we do not need to apply \cref{claim:trim}. Thus, the total running time to decide \unambiguity on unary NFAs is $\Oh(n^{1+o(1)} + m)$ in general, and $n^{1 + o(1)}$ if the given NFA is trim.

    To prove \cref{thm:NFA_unary}\;(b), we use the equivalences stated in \cref{lem:NFA_equivalence} between \FFA and \ida, and between \PFA and \eda. 
    Consider the decomposition of $G$ into strongly connected components.
    Note that \eda is equivalent to deciding whether there is a component with multiple simple cycles, which is clearly linear time decidable.
    Furthermore, \eda implies \ida, so assume that \eda does not hold. We show below that \ida holds if and only if there is a walk between two nontrivial components. This  concludes the proof as verifying this property is easy using the decomposition.

    First assume \ida. Since \eda does not hold,  we can see that no witness for \ida can occur in just one component, yielding the walk between components.
    For the converse, let $v$ and $w$ be vertices in two different nontrivial components. Let $a_v, a_w > 0$ be the lengths of arbitrary simple cycles starting in $v$ and in $w$, respectively, and let $b > 0$ be the length of a walk 
    from $v$ to $w$. Let $w = w_0, w_1 \dots, w_{a_w -1}, w_{a_w} = w$ be vertices of the simple cycle starting in $w$.
    We define $c = a_v \cdot a_w \cdot x$, where $x \in \N$ is such that $c > b$.
    Observe that there are cycles of length $c$ starting in each vertex $v, w_0, w_1,\ldots, w_{a_w-1}$ (by walking through the simple cycles of length $a_v$ or $a_w$). 
    Since $c > b$, there is an $i \in \rng{0}{a_w-1}$ such that the vertex $w_i$ is reachable from $v$ by a walk of length $c$. The vertices $v$ and $w_i$, the length $c$ of the walk from $v$ to $w_i$, and the considered cycles starting at $v$ and $w_i$ witness \ida.
\end{proof}

\paragraph{Proof sketch of \cref{thm:WFA_unary}}
Recall that deciding whether an unambiguous weighted automaton is determinisable is equivalent to checking if it has the \twins property (see~\cref{lem:WFA_equivalence}). 
In \cref{sec:twins_unary}, we will first reduce the \twins problem in the unary case to the following graph question:
\begin{description}
    \item \centering{\emph{Given a weighted directed graph, are there two cycles of different average weight?}}
\end{description}
By scaling the weights, we show that this task can in turn be reduced to detecting a cycle of nonzero weight. Finally, since any cycle is contained in a single strongly connected component, we can use the decomposition of the graph into strongly connected components to solve the problem in linear time~\cite{tarjan72}. The complete proof is given in \cref{sec:twins_unary}.

%% file: disjoint_prog.tex
\begin{SCfigure}
\centering
\begin{tikzpicture}[place/.style={circle,draw=black,inner sep=0pt,minimum
    size=7mm,very thick}]
\node[place] (v3) {};
\node[place,left = 1cm of v3] (v1) {};
\node[place,right = 1cm of v3] (v4) {};
\node[place,above left = 1cm and 1.5cm of v1] (s) {$s$};
\node[place,above right = 1cm and 1.5cm of v4] (t) {$t$};
\node[place,above = 1.75cm of v3] (v2) {};
\node[place,above right = 0.7cm and 0.4cm of v2] (d1) {};
\node[place,above left = 0.7cm and 0.4cm of v2] (d2) {};
\node[place,right = 1cm of v4] (d) {};

\path[->,thick,>=stealth]
(s) edge (v1)
(s) edge (v2)
(v1) edge (v2)
(v1) edge (v3)
(v3) edge (v4)
(v4) edge (t)
(v2) edge (t)
(v4) edge[bend left] (d)
(d) edge[bend left] (v4)
(v2) edge (d1)
(d1) edge (d2)
(d2) edge (v2)
(s) edge[dashed] (t)
;
 \end{tikzpicture}
\caption{\small{
        Intuition behind the equivalence of the \unambiguity and
        \disjointprogressions problems. Graphs with nested or multiple cycles on
        a walk violate unambiguity, so we assume each $st$-walk passes through
        at most one cycle. In the example, the \unambiguity instance reduces to
        a \disjointprogressions instance with $k = 2$, $b_1 = 2$, $b_2 = 3$,
        $A_1 = \{0\}$, $A_2 = \{0,2\}$. Here, $b_i$ are cycle lengths, and $A_i$ are
        $st$-walk lengths modulo $b_i$. Walks outside cycles (dashed arrow) are
        handled separately via a DAG algorithm, which also computes the $b_i$
and $A_i$.}}\label{fig:disjoint_unambiguity}

\end{SCfigure}

%% file: preliminaries.tex
Throughout this  paper, $\N$ denotes the set of nonnegative integers 
and for any $a,b \in \N$ with $a\le b$ we use the abbreviation $\rng{a}{b} =
\{a,a+1,\ldots,b\}$.
We use the Iverson bracket notation for the characteristic function; i.e.\ for any logical expression $B$, the value of $\iverson{B}$ is $1$ if $B$
is true and $0$ otherwise. 
$\uplus$ denotes the disjoint union of sets.
For two integers $a,d$, we use $d \divides a$ to denote that $d$ divides $a$, and we denote the
greatest common divisor of $a$ and $b$ by $\gcd(a,b)$. 
The remainder of the division of $a$ by $d$ (which lies in
$\rng{0}{d-1}$) is denoted $a \bmod d$, and we use $a
\equiv_d b$ or $a \equiv b \, (\bmod\, d)$ to say $d \divides (a-b)$. 

\subsection{Finite automata}
A \emph{nondeterministic finite automaton} (NFA) is a tuple $\mathcal{A} = (Q, \Sigma, \delta, I, F)$, 
where: $Q$ is a finite set of states, $\Sigma$ is a finite alphabet, $I, F \subseteq Q$ are the sets of initial and final states, respectively, and $\delta \subseteq Q \times \Sigma \times Q$ is the set of transitions.  
We say that an NFA is deterministic (a DFA) if $|I| = 1$ and, for every $p \in Q$ and $a \in \Sigma$, there is at most one $q$ such that $(p,a,q) \in \delta$.
The \emph{size} of $\mc A$, denoted $|\mc A|$, is defined as $|Q| + |\Sigma| + |\delta|$.

A \emph{run} of an automaton $\mathcal{A}=(Q, \Sigma, \delta, I, F)$ on a word $w = w_1 \ldots w_n \in \Sigma^*$ from state $p$ to state $q$ is a sequence of 
transitions $\rho = t_1, \dots, t_n$ such that there exist states $p_0, p_1, \dots, p_n \in Q$ with $p_0 = p$, $p_n = q$ and $t_i = (p_{i-1}, w_i, p_i) \in \delta$ for all $i \in \rng{1}{n}$. 
We say that the states $p_0, \dots, p_n$ are \emph{visited} by $\rho$.
If a run is non-empty and $q_0=q_n$, then we call it a \emph{cycle}.
The \emph{length} of $\rho$ is $|\rho| = n$. 
We say that $\rho' = t_i \ldots t_j$, where $1 \leq i \leq j \leq n$, is a \textit{subrun} of $\rho$.
If $q \in I$ and $q \in F$, then we say that $\rho$ is \emph{accepting} and that $\mc A$ \emph{accepts} (or \emph{recognizes}) $w$. 
%

Given an automaton $\mc A$ and a word $w \in \Sigma^*$ we write $\runs_{\mc A}(w)$ for the number of \emph{accepting} runs of $\mc A$ on $w$. If $\mc A$ is clear from the context we will simply write $\runs(w)$. We define the following classes of automata, each providing different guarantees on how $\runs(w)$ can be bounded.

\begin{definition}\label{definition:ambiguity}
Let $\mathcal{A} = (Q, \Sigma, \delta, I, F)$ be an NFA. We say that:
\begin{itemize}
 \item $\mc A$ is \emph{unambiguous} if $\runs(w) \le 1$ for all $w \in \Sigma^*$;
 \item $\mc A$ is \emph{finitely ambiguous} if there is $c \in \N$ such that $\runs(w) \le c$ for all $w \in \Sigma^*$;
 \item $\mc A$ is \emph{polynomially ambiguous} if there is a polynomial $p$ such that $\runs(w) \le p(|w|)$ for all $w \in \Sigma^*$. 
\end{itemize}
\end{definition}

We sometimes write an exponentially ambiguous NFA to mean any NFA (e.g.\ to emphasise that there is an exponential bound on the number of accepting runs). See \cref{fig:ambiguity} for example automata with different ambiguity.

An NFA $\mathcal{A} = (Q, \Sigma, \delta, I, F)$ is called \emph{trim} if for every state $q \in Q$ there is an initial state $q_0 \in I$ and a final state $f \in F$ with a run from $q_0$ to $q$ and a run from $q$ to $f$. 
The language $\mc L(\mc A)$ of an automaton $\mc A$ is the set of words accepted by $\mc A$. We say that two automata $\mc A_1$ and $\mc A_2$ are \emph{equivalent} if $\mc L({\mc A_1}) = \mc L({\mc A_2})$.
More refined, two automata $\mc A_1$ and $\mc A_2$ are \emph{parsimoniously
equivalent} if they have the same number of accepting runs for any given word, that is, if $\runs_{\mc A_1}(w) = \runs_{\mc A_2}(w)$ for every $w\in\Sigma^\ast$.
Note that parsimonious equivalence implies equivalence. 
The following standard claim is proven by two applications of a depth-first search algorithm pruning all unreachable and co-unreachable states.
\begin{observation}\label{claim:trim}
Given an NFA $\mathcal{A} = (Q, \Sigma, \delta, I, F)$, one can compute a
parsimoniously equivalent trim NFA $\mathcal{B}$ in time $\Oh(|{\mc A}|)$.
Moreover, $\mc B$ is obtained from $\mc A$ by removing all unreachable and co-unreachable states and their incident transitions. 
\end{observation}

We say that an NFA $\mathcal{A} = (Q, \Sigma, \delta, I, F)$ has \emph{normal form} if $|I| = |F| = 1$. 
The following claim is straightforward to prove by introducing two auxiliary states -- one state becoming the only initial and one becoming the only final state -- and the appropriate transitions from and to them. 

\begin{observation}\label{claim:normal}
Given an NFA $\mathcal{A} = (Q, \Sigma, \delta, I, F)$, one can compute in $\Oh(|{\mc A}|)$ time a parsimoniously equivalent NFA $\mathcal{B}$ with $|\mc B| = \Oh(|\mc A|)$ in normal form. 
\end{observation}

Some problems are more natural to phrase assuming the normal form. When appropriate, we will discuss the complexity both assuming that the input is in normal form and with the overheads arising from \cref{claim:trim} and \cref{claim:normal}. 

A \emph{weighted automaton} is a tuple $\mathcal{W} = (Q, \Sigma, \delta, I, F, \weight)$, where $(Q, \Sigma, \delta, I, F)$ is an NFA and $\weight\colon \delta \to \Z$ is a weight function for the transitions.
Readers familiar with weighted automata can think of weighted automata over the tropical semiring, i.e.\ $\Z(\min,+)$. We avoid the technical definitions as they will not be needed.\footnote{Weighted automata also have input and output weights for states. For simplicity, we omit this part of the definition.}
All the notions and observations made above transfer naturally to weighted automata from the underlying NFA. The weight of a run $\rho = t_1 \ldots t_n$ is $\weight(\rho) = \sum_{i = 1}^{n} \weight(t_i)$.

\subsection{Directed graphs}\label{sec:prelim-graphs}

A \emph{directed graph} is a pair $G = (V,E)$, where $V$ is a finite set of vertices and $E \subseteq V \times V$ is a finite set of directed edges. Sometimes we write $V(G)$ and $E(G)$ to emphasise the graph. Note that we allow for self-loops from a vertex to itself. Such a self-loop is a cycle of length $1$.
If $G$ is acyclic, we call it a \emph{DAG}. Throughout this paper, all graphs are directed. The size $|G|$ of a graph $G$ is defined as $|G|=|V(G)|+|E(G)|$. A weighted graph has an additional weight function $\weight\colon E\to \Z$ and is denoted $(G,\weight)$ or $(V,E,\weight)$.

Let $\mathcal{A} = (Q, \Sigma, \delta, I, F)$ be an NFA in normal form that is \emph{unary} (i.e.\ $|\Sigma| = 1$). Consider the graph with vertex set $V=Q$ and the edge set $E=\{(q,q')\mid \exists a\in \Sigma\colon (q,a,q')\in\delta\}$ obtained by projecting $\delta$ onto $Q \times Q$. 
We further distinguish the \emph{start vertex} $s \in I$ and the \emph{target vertex} $t \in F$ (the unique initial and final state of $\mathcal{A}$). 
Denote such a graph $(G, s, t)$ or call it an \emph{$st$-graph} to emphasise the distinguished vertices. We say that $\mc A$ is the underlying automaton of the $st$-graph $G$ as constructed above. 
Note that walks in an $st$-graph $G$ (formalized as sequences of edges, but also denoted as the corresponding sequences of vertices if convenient) correspond to runs in the underlying $\mc A$, and $st$-walks correspond to accepting runs. 

A \emph{walk} is called a \emph{simple path} if every vertex appears at most once, with the exception that the first and the last vertices can be the same.\footnote{The literature sometimes calls a \emph{path} what we call a \emph{walk}. To avoid misunderstandings we will only use the terms \emph{simple path} and \emph{walk}.}
A \emph{cycle} is a walk that starts and ends at the same vertex.
The cycle is said to be \emph{simple} if the walk is a simple path.

A weighted graph $(G,\omega)$ with distinguished vertices $s$ and $t$ is denoted $(G,s,t,\weight)$ or $(V,E,s,t,\weight)$ for $G=(V,E)$. 
Note that adding $\weight$ to the underlying automaton of $(G,s,t)$ yields a weighted automaton in normal form. 
The weight of a walk is the sum of weights on its edges (with multiplicity). 
We call an $st$-graph and a weighted 
$st$-graph \emph{trim} if their underlying automata are trim.

The graph induced by the strongly connected components of a graph $G$ has the strongly connected components of $G$ as vertices and an edge $(X,Y)$ between two components if and only if $Y$ is reachable from $X$ in $G$. 
Tarjan's algorithm~\cite{tarjan72} computes this graph in $\Oh(n+m)$ time for $n=|V(G)|$ and $m=|E(G)|$. 


\subsection{Complexity analysis}
To describe asymptotic upper bounds on the running time of algorithms, we use the usual Landau notation. 
We also use the notation $f(n) = \Ot(g(n))$, which stands for: 
$\exists k\in\N\colon f(n) =  \Oh(g(n) \log^k(g(n)))$ and hides negligible polylogarithmic factors. 
The running time measured in $n$ is called \emph{almost linear} if it is in $n^{1+o(n)}$. 

For ease of presentation, we assume that a single arithmetic operation takes only unit time. 
To make our upper bounds on the running times precise in any standard computation model, it suffices to multiply them by the maximal cost of the involved arithmetic operations. For unweighted NFAs this is important only in \cref{sec:unambiguity_upper}, where the impact is polylogarithmic and negligible for the complexity. For weighted NFAs, the cost of an operation depends on the assumed model of computation. In \cref{sec:twins_unary} we compare results with those by Allauzen and Mohri~\cite{Allauzen2003EfficientAF}, who implicitly assume that arithmetic operations take unit time.


\subsection{Fine-grained hypotheses}

Now, we briefly state fine-grained complexity hypotheses relative to which we prove our hardness results. 


\begin{hypothesis}[Strong Exponential-Time Hypothesis (SETH) {{\cite{ImpagliazzoPZ01}}}]
	For every $\delta<1$, there is $k\in \N$ such that $k$-\textsc{SAT} with
    $n$ variables (restricted to clauses of width at most $k$) cannot be solved in $2^{\delta n+o(n)}$ time.
\end{hypothesis}

\begin{hypothesis}[Orthogonal-Vector (OV) hypothesis {{\cite{Williams05}}}]
	For every $\eps > 0$, there is $c >0$ such that there is no $\mathcal \Oh(n^{2 - \eps})$-time algorithm that, given two sets $A, B \subset \{0, 1\}^d$ of $|A|=|B|=n$ vectors of dimension $d = c \log n$, decides if $A$ contains a vector orthogonal to a vector in $B$.	
\end{hypothesis}
\begin{hypothesis}[$k$-OV hypothesis]
	For every $\eps > 0$, there is $c>0$ such that there is no
    $\Oh(n^{k - \eps})$-time algorithm that, given sets
$A_1,\ldots,A_k \subset \{0, 1\}^d$  vectors of dimension $d = c
    \log n$ and size $|A_1| = \ldots = |A_k| =n$, decides if exists a $k$-tuple
    $(a_1,\ldots,a_k) \in A_1 \times \ldots \times A_k$ such that $\prod_{i=1}^k a_i[j] = 0$ for every $j \in \rng{1}{d}$.
\end{hypothesis}

It is known that SETH implies the $k$-OV hypothesis for every $k\ge 2$ (see for example~\cite[Lem.~A.1]{WilliamsY14} and the discussion in~\cite{fine-grained-lecture}).

\begin{hypothesis}[$k$-Cycle hypothesis]
	For every $\eps > 0$ there is an integer $k$ such that there is no $\mathcal{O}(m^{2 - \varepsilon})$-time algorithm
	for finding a cycle of length $k$ in directed graphs with $m$ edges.
\end{hypothesis}

%% file: UU.tex
We complete the proof of \cref{thm:NFA_unary} discussed in \cref{sec:technical_overview_unary} by proving \cref{thm:equivalence}.

\unaryproblemequivalence*

As already observed, in unary NFAs, a word with an accepting run is characterized by the length of the corresponding source-to-target walk in the underlying directed graph. 
Since we can assume without loss of generality that the given NFA is trim and in normal form (see~\cref{claim:trim,claim:normal}), we can reformulate the problem of \unambiguity for unary NFAs as follows. 
\begin{description}
    \item[\unaryunambiguity] Given a trim $st$-graph $(G,s,t)$, decide whether all $st$-walks in $G$ have pairwise different lengths.
\end{description}

One direction of the above equivalence is quite simple to argue. In fact, to prove \cref{thm:NFA_unary} we only need the other direction, i.e.\ the reduction from \unaryunambiguity to \disjointprogressions, but we include the reduction from \disjointprogressions to \unaryunambiguity for completeness and to show that relating those two problems is optimal.

\begin{lemma}\label{lem:progressiontoambiguity}
	Let $I = (b_1, A_1), \dots, (b_k, A_k)$ be an instance of \disjointprogressions with $n = b_1 + \dots + b_k$.
	There is an $\Oh(n)$-time algorithm that produces an instance $(G,s,t)$ of \unaryunambiguity with $|V(G)| = \Oh(n)$ and $|E(G)| = \Oh(n)$ such that $I \in \disjointprogressions$ if and only if $(G, s, t) \in \unaryunambiguity$. 
\end{lemma}
\begin{proof}
    First, create the source vertex $s$ and target vertex $t$. Next,
    for each $i \in \rng{1}{k}$, create the following: (1) a simple
    oriented cycle $C_i$ on vertices $v_{i,0},\dots,v_{i,b_i-1}$ (in this order), (2) the
    edge $e_i=(s, v_{i,0})$, and, (3) for each $a \in A_i$, the edge $e_{i,a}=(v_{i,a}, t)$. 

	Since $b_1 + \dots + b_k = n$, the union of all cycles in $G$ contains $n$ vertices and $n$ edges. With the additional vertices $s$ and $t$, we get $|V(G)| \leq n + 2$. There is at most one edge from $s$ to each cycle and at most one edge from every vertex to $t$. Hence, $|E(G)| \leq 3n$ and the construction takes $\Oh(n)$ time.

	Now, observe that any $st$-walk of $G$ starts with the edge $e_i$ for some
    $i \in \rng{1}{k}$, then loops $x$ times around the cycle $C_i$ for some $x
    \in \mathbb N$, next walks from $v_{i, 0}$ to $v_{i, a}$ along $C_i$ for
    some $a \in A_i$, and finally leaves the cycle via the edge $e_{i, a}$ to the
    vertex $t$. The function mapping each $st$-walk of $G$ to its length $2 + x
    \cdot b_i + a$ is injective if and only if the arithmetic progressions $\{a
    + x \cdot b_i \mid x \in \mathbb N\}$ for $i \in \rng{1}{k}$ and $a \in A_i$ are disjoint. In other words, $(G, s, t) \in \unaryunambiguity$ if and only if $I \in \disjointprogressions$.
\end{proof}

In the remainder of this section, consider a trim $st$-graph $(G, s, t)$ and denote $n \coloneqq |V(G)|$ and  $m \coloneq |E(G)|$. We assume without loss of generality that $n = \Oh(m)$ (see~\cref{lem:sparsity}). 
We now prove the other direction of \cref{thm:equivalence}. This relies on the fact that \unaryunambiguity can be decided in linear time if the given graph is acyclic, i.e.\ $G$ is a DAG. The algorithm for the DAG case given in \cref{thm:unary_amb_dag_algo} additionally computes lengths of $vt$-walks for some vertices $v \in V(G)$. This will be useful for the above reduction.

%

\begin{lemma} \label{thm:unary_amb_dag_algo}
   There is an $\Oh(n)$-time algorithm deciding whether $(G, s, t) \in \unaryunambiguity$ for any acyclic trim $st$-graph $(G, s, t)$ with $n$ vertices and $m$ edges.
   
   Moreover, if $(G, s, t) \in \unaryunambiguity$, the algorithm computes: (1) the set $P_{st}$ of lengths of all $st$-walks in $G$, and (2) for every vertex $v \in V(G)$ of in-degree at least $2$ the set $P_{vt}$ of lengths of all $vt$-walks in $G$.
\end{lemma}
\begin{proof}
    We start with a brief overview of the proof.
    Let $S \subset V(G)$ be the set of vertices of in-degree at least 2. Since $G$ is trim, two $vt$-walks of the same length yield two $st$-walks of the same length. So we start by computing the sets $P_{vt}$ of lengths of all $vt$-walks for every $v \in S \cup \{s\}$, but abort the computation and output $(G, s, t) \notin \unaryunambiguity$ if we ever detect two $ut$-walks of the same length for any $u\in S \cup \{s\}$. 
    If the computation is never aborted, it succeeds and guarantees that $(G, s, t) \in \unaryunambiguity$. 
    We show how the algorithm computes the sets $P_{vt}$ for all $v \in S \cup \{s\}$ in this case and bound the runtime. 
    Remark that we cannot afford to compute all sets $P_{vt}$ for \emph{all vertices} $v \in V(G)$ as this might have quadratic size, see~\cref{fig:counterexample} for an example.
    
    \input{fig_dag_tree}
    Let $T$ be an arbitrary directed tree rooted at $s$ spanning $G$. Recall that $S \subseteq V(G)$ is the set of vertices with in-degree at least $2$ (bold in \cref{fig:dag_tree}).
    Note that $S$ is exactly the set of vertices with an incoming non-tree edge. (We have $s\notin S$ because $G$ would not be acyclic otherwise.) 
    Observe that, for every $v \in V(G)$, we can partition the set of non-empty $vt$-walks into walks starting with a tree edge and walks starting with a non-tree edge. We define the following for every vertex $v \in V(G)$:
    \begin{itemize}
        \item $\texttt{STAY}_v$ is the set of lengths of all $vt$-walks starting with a tree edge, except for $v = t$, where $\texttt{STAY}_t = \set{0}$ for the empty $tt$-walk.
        \item $\texttt{JUMP}_v$ is defined
        as the set of lengths of all $vt$-walks starting with a non-tree edge.
    \end{itemize}
    Note that $\texttt{JUMP}_t = \emptyset$.
    Clearly, $\texttt{STAY}_v \cup \texttt{JUMP}_v$ is the set of lengths of all $vt$-walks. 
    We cannot afford to compute $\texttt{STAY}_v$ and $\texttt{JUMP}_v$ for every vertex $v \in V(G)$ because the number of $vt$-walks for every $v \in V(G)$ can be prohibitively large. Therefore, we compute the set $\texttt{JUMP}_v$ for every $v \in V(G)$ but the set $\texttt{STAY}_v$ only for vertices $v \in S \cup \{s, t\}$. We argue that the sum of the cardinalities of these sets is in $\Oh(n)$. We trivially have $|\texttt{STAY}_t| = 1$. 
    The function that, for each $vt$-walk $p$ starting with a non-tree edge, maps $p$ to the $st$-walk consisting of the unique $sv$-walk along the tree followed by $p$, is injective. Since $G$ is trim, two $vt$-walks of the same length yield at least two $st$-walks of the same length. As $G$ is acyclic, any walk has length at most $n$, 
    and there are thus at most $n$ $st$-walks in $G$.
    So we have  
    $\sum_{v \in V(G)} |\texttt{JUMP}_v| \leq n.$
    A further trivial consequence of having at most $n$ $st$-walks in $G$ is $|\texttt{STAY}_s| \leq n$. 
    Finally, for every $v \in S$ there is a non-tree edge $(u, v) \in E(G) \setminus E(T)$. Hence, every $vt$-walk $p$ can be mapped to a unique $st$-walk consisting of the $su$-walk along the tree, the non-tree edge $(u, v)$ and the $vt$-walk $p$. Hence, we have $\sum_{v \in S} |\texttt{STAY}_v| \leq n.$

    We now describe the algorithm. Compute a spanning tree $T$ and sort the
    vertices in reverse topological order in $\Oh(n + m)$ time (e.g.\ with a
    breadth-first search algorithm, notice that if we detect $m = \Omega(n)$, then we can
answer that $(G,s,t) \notin \unaryunambiguity$ as $G$ is trim)  
    We process vertices in this reverse topological order and compute by induction the sets $\texttt{STAY}_v$ for every $v \in S \cup \{s, t\}$ and $\texttt{JUMP}_v$ for every $v \in V(G)$.
    Observe that, for any $v \in V(G)$, any $vt$-walk starting with a non-tree edge $(v, u)$ is the concatenation of the edge $(v, u)$ with a $ut$-walk. Since $u \in S$ is processed before its predecessor $v$, the sets $\texttt{STAY}_u $ and $ \texttt{JUMP}_u$ are already computed and we can compute $\texttt{JUMP}_v$ recursively as
    $$
    \texttt{JUMP}_v = \bigcup_{(v, u) \in E(G) \setminus E(T)} \{\ell +1 \mid \ell \in \texttt{STAY}_u \cup \texttt{JUMP}_u\}.
    $$
    We now show how to compute $\texttt{STAY}_v$ for any given $v \in S \cup \{s, t\}$.  Let $T_v$ be the maximal subtree of $T$ rooted at $v$ such that all internal vertices (i.e.\ all vertices except for the root and the leaves) are in $V \setminus (S \cup \{s\})$. Let $L_v \subset S \cup \{t\}$ be those leaves of $T$ that are in $S \cup \{t\}$ and denote by $d_v(u)$ the depth of $u$ in $T_v$. See \cref{fig:dag_tree} for an illustration.
    Consider a $vt$-walk $p$ starting with a tree edge and let $u \in V(T_v)$ be the first vertex at which the walk exits $T_v$ or, if $p$ is contained in $T_v$, let  $u = t$. Then $p$ exits $T_v$ either
    \emph{(i)} 
    with a non-tree edge, in which case $p$ is the concatenation of the $vu$-walk along $T_v$ with a $ut$-walk starting with a non-tree edge from $u$, or 
    \emph{(ii)} 
    with a tree edge or not at all, in which case $u \in L_v \subset S \cup \{t\}$ due to $T_v$'s maximality and $p$ is the concatenation of the $vu$-walk along $T_v$ with a $ut$-walk starting with a tree edge or the empty $tt$-walk. Since $u$ is processed before $v$, we can compute $\texttt{STAY}_v$ as follows.
    \begin{align*}
        \texttt{STAY}_v = \phantom{\cup} 
        \bigcup_{u \in T_v \setminus \{v\}} &\{\ell + d_v(u)\mid\ell
        \in \texttt{JUMP}_v \}
        \\
        \cup 
        \bigcup_{u \in L_v \setminus \{v\}} &\{\ell + d_v(u)\mid\ell \in
        \texttt{STAY}_v\}
    \end{align*}
    It remains to observe that $\texttt{STAY}_s \cap \texttt{JUMP}_s = \emptyset$ and to compute $P_{vt} = \texttt{STAY}_v \cup \texttt{JUMP}_v$ for each $v \in S \cup \{s\}$. 
    Observe that $\{T_v \setminus \{v\} \mid v \in S \cup \{s, t\}\}$ is a partition of $T$. Thus, the total running time is in $\Oh(\sum_{v \in V(G)} |\texttt{JUMP}_v| + \sum_{v \in S \cup \{s, t\}} |\texttt{STAY}_v| + |T|)$ and thus, by the shown bounds and $|T|\le n$, in $\Oh(n)$. 
\end{proof}

\input{appendix_unary_amb}

%% file: fig_dag_tree.tex
\begin{figure}
  \centering
\begin{tikzpicture}
  [
  scale=0.8,
  thick,
  -{Stealth[length=2.4mm, width=1.2mm]},
  level distance=1.5cm,
  every node/.style={circle,draw, inner sep=1pt,minimum size=.5cm},
  level 1/.style={sibling distance=2cm},
  level 2/.style={sibling distance=2cm},
  level 3/.style={sibling distance=1.5cm},
  level 4/.style={sibling distance=1cm},
  ]
  
  \node (s) {$s$}
    child {node (a) {$a$}}
    child {node[line width=2pt] (b) {$b$}
      child {node[line width=2pt] (c) {$c$}
        child {node[line width=2pt] (e) {$e$}}
      }
      child {node (d) {$d$}
        child {node (f) {$f$}}
        child {node[line width=2pt] (g) {$g$}
          child {node[line width=2pt] (t) {$t$}}
        }
      }
  };

  \begin{scope}[dashed]
    \draw (a)--(b);
    \draw (a) -- (c);
    \path (a)  edge [bend right=50]   (e);
    \draw (f)--(g);
    \path (e)  edge   [bend right=30]   (t);
  \end{scope}

  \begin{scope}[on background layer]
    \draw[gray!50,line width=2.2em,line cap=round, rounded corners] (b.center) -- (c.center) -- ++($.21*(c.center) - .21*(b.center)$);    
    \draw[gray!50,line width=2.2em,line cap=round, rounded corners] (b.center) -- (d.center) -- (g.center) -- ++($.21*(g.center) - .21*(d.center)$);    
    \draw[gray!50,line width=2.2em,line cap=round, rounded corners] (b.center) -- (d.center) -- (f.center) -- ++($.21*(f.center) - .21*(d.center)$);    
  \end{scope}

\end{tikzpicture}
\caption{Illustration for the proof of \cref{thm:unary_amb_dag_algo}: A directed
acyclic graph $G$, the edges of a spanning tree $T$ are continuous and the
non-tree edges are dashed, the set $S$ of vertices of in-degree at least $2$ is
in bold. With grey background we illustrate $T_b$ which is the maximal
subtree of $T$ rooted at $b$ such that all internal vertices are in $V \setminus
\{S \cup \{s\}\}$.}\label{fig:dag_tree}
\end{figure}

%% file: appendix_unary_amb.tex
To reduce \unaryunambiguity to \disjointprogressions, we start with the following observation on cycles in unambiguous graphs. This property is then used in \cref{lem:transformation} to transform the graph.

\begin{lemma}\label{lem:disjoint_cycles}
	Let $(G, s, t)$ be a trim $st$-graph. 
	If $(G, s, t) \in $\unaryunambiguity, then we have the two properties that all cycles in $G$ are disjoint and that every $st$-walk in $G$ intersects vertices with at most one cycle. These two properties are $\Oh(n + m)$-time decidable.
\end{lemma}	
\begin{proof}
	We prove the contrapositive.
	Assume that 
	there are two distinct cycles $c_1$ and $c_2$ intersecting at some vertex $v$. 
	Since $G$ is trim, there is an $sv$-walk $p$ and a $vt$-walk $q$. 
	The walks $p c_1 c_2 q$ and $p c_2 c_1 q$ are two distinct $st$-walks of equal length and thus witness that $(G,s,t) \notin $\unaryunambiguity. 
	Assume now that 
	there is an $st$-walk containing a vertex from each of two distinct cycles $c_1$ and $c_2$ of lengths $\ell_1$ and $\ell_2$, respectively. Let $v_1 \in V(c_1)$ and $v_2 \in V(c_2)$ be these vertices, and assume without loss of generality that $v_1$ appears before $v_2$ in the $st$-walk. Decompose the $st$-walk into an $sv_1$-walk $p$, a $v_1 v_2$-walk $q$ and a $v_2 t$-walk $r$. 
	Then $pc_1^{\ell_2}qr$ and $pqc_2^{\ell_1}r$ are two distinct $st$-walks of equal length, and thus witness that $(G, s, t) \notin $\unaryunambiguity. 

	We now show how to decide the two properties in $\Oh(n + m)$ time. Observe that all cycles are disjoint if and only if every strongly connected component of $G$ is either a single vertex or a cycle, which can be decided in $\Oh(n + m)$ time.

	Let $H$ be the graph induced by the strongly connected components.
	There is an $st$-walk intersecting two distinct cycles if and only if there is a walk between the two cycles in $H$. 
	It can be decided whether such a walk exists by running a depth-first search (DFS) algorithm on $H$ starting from the vertices in $H$ corresponding to cycles in $G$, one after the other, until a walk from the starting point to another vertex in $H$ corresponding to a cycle in $G$ is found or all searches finish without finding one, disproving the existence.  We can stop the depth search of any DFS when it reaches a vertex already visited by any of the previously run DFS, which could not find such walk. Thus every edge is used only by one DFS and the total running time is in $\Oh(m)$.
\end{proof}

\begin{lemma}\label{lem:transformation}
	Let $(G, s, t)$ be a trim instance of \unaryunambiguity.
	There is an $\Oh(n + m)$-time algorithm that either transforms the given graph $G$ into a graph $G'$ with the properties below or, if this fails, decides whether $(G,s,t)\in$\unaryunambiguity:
	\begin{itemize}
		\item $G'$ is trim with $|V(G')| \leq 3n$ and $m \leq |E(G')| \leq 4m$.
		\item In each cycle of $G'$, all but exactly one vertex, called the \emph{cycle gate}, have in-degree and out-degree 1. If $s$ or $t$ is in a cycle, then the unique cycle gate is $s$ or $t$.
		\item $(G, s, t) \in $\unaryunambiguity if and only if $(G', s, t) \in $\unaryunambiguity.
	\end{itemize}
\end{lemma}
\begin{proof}
	By \cref{lem:disjoint_cycles} we can check in $\Oh(n + m)$ time if the cycles in $G$ are disjoint and any $st$-walk intersects with at most one cycle or, if this not the case, conclude that $(G, s, t) \notin $\unaryunambiguity. 
	Otherwise, if $s$ and $t$ lie in the same cycle and $G$ is trim, $G$ must be a simple cycle.
	Then all the $st$-walks in $G$ have different lengths, and we conclude that $(G, s, t) \in $\unaryunambiguity.
	
	So we can assume that there is no cycle containing both $s$ and $t$ in $G$.
	We construct a graph $G'$ by transforming each cycle $C$ of $G$ and its surroundings with the following three steps, as illustrated in \cref{fig:cycleentrance}. 
	\input{fig_cycle_transformation.tex}
	\begin{description}
        \item[\textit{Step 1:} \textit{Add incoming and outgoing walk.}]\hfill\\
		Denote the length of the cycle $C$ by $\ell$.
		Choose an arbitrary vertex $v_0$ of $V(C)$, which we call the \emph{cycle gate}, and denote the remaining vertices of the cycle in the orientation order by $v_1,\dots,v_{\ell-1}$. If $C$ contains $s$ (or $t$), then let $s$ (or $t$) be the cycle gate. 
		Append to $v_0$ an incoming simple path of length
		$\ell'=\max\{k \in \rng{1}{\ell -1} \mid \deg_\text{in}(v_{\ell - k})>1\}$
       	$i_{\ell - \ell'}, \dots, i_{\ell} =
        v_0$ and an outgoing simple path of length
        $\ell''=\max\{k \in \rng{1}{\ell-1} \mid \deg_\text{out}(v_k)>1\}$ on vertices $v_0 = o_0, o_1, \dots, o_{\ell''}$. 
        We call the union of the incoming and outgoing walk (including the cycle gate) the \emph{gateway}.

    \item[\textit{Step 2:} \textit{Replace outgoing edges.}]\hfill\\
		Replace each edge $(v_j, u)$ of $G$ leaving $C$ (i.e.\ with $u \notin V(C)$) by the \emph{exit edge} $(o_j, u)$.
		
    \item[\textit{Step 3:} \textit{Replace incoming edges.}]\hfill\\
		Similarly, replace each edge $(u, v_j)$ of $G$ entering $C$ (i.e.\ with $u \notin V(C)$) by the \emph{(default) entry edge} $(u, i_j)$ 
		and, if $j\le\ell''$, add also		
		the \emph{bypass entry edge} $(u, o_j)$.  
	\end{description}
	Since all the cycles are disjoint, this construction takes $\Oh(n + m)$ time. We have $m \leq |E(G')|$. Additionally, at most $2 (|V(C)| -1)$ vertices and at most $2(|V(C)| -1)$ edges are added in Step 1. Throughout all iterations, an edge is replaced by two edges at most once per edge in Step 3. Since the cycles are disjoint, in total the number of vertices increases by at most $2n$ and the number of edges increases by at most $2n +m \leq 3m$. Hence, $|V(G')| \leq 3n$ and $m \leq |E(G')| \leq 4m$. 

	Clearly, by construction, $G'$ is trim and every vertex in a cycle except for the cycle gate has in-degree and out-degree $1$. 
	We now verify that the number of $st$-walks for any given length is maintained. For this, consider the following mapping between $st$-walks in $G'$ to $st$-walks in $G$:
	Any $st$-walk that contains no gateway vertex is mapped to itself.
	Every other $st$-walk contains exactly one entry edge $e$ (either default
	or bypass), exactly one exit edge $f$, and a (possibly empty) walk $p$
	from $e$ to $f$.  
	Every such $st$-walk is mapped to the walk in $G$ where $epf$ is replaced by $e'p'f'$, where $e'$ is the edge that $e$ (or its corresponding default) was replaced by, $f'$ is the edge that $f$ was replaced by, and $p'$ 
	is either of the following. 
	If $e$ is a bypass entry edge, then $p'$ is the shortest walk from $e$ to $f$ along $C$. 
	If $e$ is a default entry edge, then $p'$ is the concatenation of first the shortest walk from $e$ to the cycle gate, then $\lfloor |p|/|V(C)|\rfloor$ times looping around $C$ starting and ending at the cycle gate, and finally the shortest walk from the cycle gate to $f$. 
	The described mapping is in fact a length-preserving bijection between $st$-walks in $G'$ and $st$-walks in $G$. This implies that the number of $st$-walks of length $\ell$ for any $\ell \in \N$ is preserved by the transformation of $G$ into $G'$, i.e. $(G, s, t) \in $\unaryunambiguity if and only if $(G', s, t) \in $\unaryunambiguity.
\end{proof}

We can now finally prove \cref{thm:equivalence}. 

\begin{proof}[Proof of \cref{thm:equivalence}]
	By \cref{lem:progressiontoambiguity}, we can reduce \disjointprogressions to \unaryunambiguity in linear time. Hence, let $(G, s, t)$ be an instance of \unaryunambiguity with $n = |V(G)|$ and $m = |E(G)|$. We describe an $\Oh(n+m)$-time algorithm that either decides whether $(G, s, t) \in $\unaryunambiguity, or otherwise produces an instance $I = (b_1, A_1), \dots, (b_k, A_k)$ for \disjointprogressions with $b_1 + \dots + b_k = \Oh(n)$ such that $(G,s,t)\in \unaryunambiguity $ if and only if $ I \in\disjointprogressions$. This establishes the linear time equivalence stated in \cref{thm:equivalence}.

By \cref{claim:trim,lem:disjoint_cycles,lem:transformation}, we can assume without loss of generality that $G$ is trim, all cycles are disjoint, every cycle can be entered and exited only via a single of its vertices called its \emph{cycle gate}, and every $st$-walk intersects at most one cycle.
The goal of the algorithm is to establish a correspondence between the 
$st$-walks in $G$ 
and the elements of some arithmetic progressions. If we detect any two $st$-walks of the same length during the construction, then we have a witness for $(G, s, t) \notin $\unaryunambiguity and end the algorithm.
More precisely, observe that the length of an $st$-walk intersecting a cycle of
size $b \in \mathbb N$ can be mapped to an element of the arithmetic progression
with steps $b$ and its bases equal to the length of this walk with all
loops around the cycle removed. The idea is now to compute, for each cycle size
$b$, all possible bases and then their remainders modulo $b$ while checking for overlaps among them (and for overlaps with lengths of $st$-walks in $G$ that do not intersect any cycle) and then construct a \disjointprogressions instance such that two $st$-walks have the same length if and only if the arithmetic progressions intersect.

Let $1 \le b_1 < \dots < b_k$ be the distinct sizes of cycles in $G$. 
For every $i \in \rng{1}{k}$, we are ultimately
interested in the length modulo $b_i$ of all $st$-walks in $G$ intersecting a cycle of size $b_i$. 
We can therefore focus on $st$-walks that intersect a cycle of size $b_i$ without looping on that cycle.
Each such $st$-walk is a simple path that can be decomposed into an $sv_C$-walk and a $v_C t$-walk, where $v_C$ is the cycle gate of the intersecting cycle.
The idea is that by removing all cycles except for their cycle gates we obtain a
DAG on which we can apply the algorithm of \cref{thm:unary_amb_dag_algo} to
compute the lengths of all $v_Ct$-walk for every cycle gate $v_C$ in linear
time. But for that, we also need to ensure that $v_C$ has in-degree at least $2$
in the acyclic graph. We do this by adding an extra walk that simulates one looping on the cycle $C$ and that does not affect the lengths of $v_C t$-walks. 
Finally, we also want to compute the lengths of all $sv_C$-walks. We do this by
applying \cref{thm:unary_amb_dag_algo} on the acyclic graph with reversed edges,
without the previously added extra walk but with a similar extra walk that ensures that every cycle gate $v_C$ has out-degree at least 2 without changing the lengths of the $sv_C$-walks. We now formally describe the construction of these two acyclic graphs.

\input{fig_G_to_H}

Let $H$ be the graph $G$ where for each cycle $C$ we remove all vertices of $V(C)$ except for the cycle gate $v_C$. 
Then $st$-walks in $H$ are in bijection with $st$-walks in $G$ of the same length that do not loop around a cycle; in particular, $(H, s, t) \notin $\unaryunambiguity implies $(G, s, t) \notin $\unaryunambiguity. This graph will be a convenient intermediate step to prove the correctness of the construction.

Suppose $v_C \neq s,t$ and let $u,w$ be its arbitrary predecessor and successor, i.e.\ with edges $(u,v_C)$ and $(v_C,w)$. Let $\ell$ be the length of the cycle $C$.
We define $H_1$ and $H_2$ to be the graphs obtained from $G$ by substituting the edge $(u,v_1)$ for $(v_C,v_1)$ and $(v_{\ell-1},w)$ for $(v_{\ell-1},v_C)$, respectively. 
\cref{fig:G_to_H} illustrates these two constructions, which are possible in $\Oh(n + m)$ time. 
Note that $H_1$ and $H_2$ are DAGs with $n$ vertices where all cycle gates (except $s$ and $t$ if they are cycle gates) have respectively in-degree and out-degree at least 2.

Apply the algorithm of \cref{thm:unary_amb_dag_algo} to $H_1$.
If the algorithm outputs that $(H_1, s, t) \notin$ \unaryunambiguity, then we directly conclude that $(G, s, t) \notin $\unaryunambiguity. Indeed, each simple path added to a cycle gate simulates one looping on the cycle with this gate in $G$, so every $st$-walk in $H_1$ can be mapped to an $st$-walk in $G$ of the same length.
Otherwise, the algorithm guarantees that $(H_1, s, t) \in $\unaryunambiguity and additionally outputs the set $P_{st}$ of lengths of all $st$-walks in $H_1$, and for every cycle $C$ the set $P_{v_Ct}$ of lengths of all $v_C t$-walks in $H_1$.
Assume this latter case. Then we also have $(H_2, s, t) \in $\unaryunambiguity, since $st$-walks in $H_1$ are in bijection with $st$-walks in $H_2$ of the same length. Additionally, $st$-walks in $H$ are preserved by the addition of simple paths, so we also have $(H, s, t) \in $\unaryunambiguity. 
Now apply the algorithm of \cref{thm:unary_amb_dag_algo} to the graph obtained by reversing all edges of $H_2$. This computes the set $P_{s v_C}$ of lengths of $s v_C$-walks in $H_2$ for each cycle gate $v_C$. 

Observe that by construction, the lengths of $v_C t$-walks in $H$ are preserved by the transformation to $H_1$, and the lengths of $s v_C$-walks in $H$ are preserved by the transformation to $H_2$. Hence, $P_{s v_C}$ and $P_{v_C t}$ contain the lengths of all $s v_C$-walks and $v_C t$-walks in $H$, respectively. 
We can therefore compute, for every cycle size $b_i > 0$, the set 
\begin{align*}
P_i \coloneqq \bigcup_{\substack{\text{ cycle } C \text{ in } G \\ |V(C)|=b_i}} \{\ell_1 + \ell_2 \mid& \ell_1 \in P_{s v_C}, \ell_2 \in P_{v_C t}\}
\end{align*}
of lengths of all $st$-walks in $H$ that intersect a cycle of size $b_i$.

It remains to deal with the set of walks $P_0$ that to do not go through any cycle gates.
Observe that $st$-walks in $G$ that loop at most once around a cycle are preserved in $H_1$, so we can compute the set of their lengths as $P_0 \coloneqq P_{st} \setminus \bigcup_{i =1}^k P_i$. 
Notice that since $H$ is a trim DAG, and we assume that $(H, s, t) \in $\unaryunambiguity, $H$ contains at most $|V(H)| = n$ distinct
$st$-walks by the pigeonhole-principle. Thus, the sets $P_0, P_1, \dots, P_k$ can be
computed in $\Oh(n)$-time. 
It remains to verify that there are no collisions among $P_0$ and arithmetic
progressions defined by the steps $b_1, \dots, b_k$ and sets of bases $P_1, \dots, P_k$.
Since $|V(H_1)| = n$, the lengths in $P_0$ are at most $n$. We check for
collisions with the following procedure that takes $\Oh(n)$ time. Initialize an
empty array of size $n$ and change the $a$-th entry to $1$ for every $a \in
P_0$. Then, for every $i \in \rng{1}{k}$ and all $a \in P_i$, change the $(a + x b_i)$-th entry to 1 for all $x \in \mathbb N$ such that $a + x b_i \leq n$. If any such entry is already set to 1, then we found a collision between two lengths of $st$-walks in $G$, implying $(G, s, t) \notin $\unaryunambiguity. 
Otherwise, we are guaranteed in particular that every $st$-walk in $G$ that 
does not intersect any cycle has a length different from that of any other $st$-walk in $G$. 
It is also simple to, in $\Oh(n)$ time, initialize for every $i \in \rng{1}{k}$ an array of length $b_i$ and compute the set $A_i \coloneqq \{a \mod b_i \mid a \in P_i\}$ while verifying that there are no distinct $a,a'\in A_i$ with $a \equiv a' \mod b_i$. 
If the latter case occurs, we can trivially conclude that $(G, s, t) \notin $\unaryunambiguity. 
Otherwise, we indeed have $(G, s, t) \in $\unaryunambiguity if and only if $I
\in \disjointprogressions$ for the instance $I$ with $1 \le b_1 < \dots < b_k$,
the set of bases $A_1,\dots,A_k$. Note that since the cycles are disjoint, $b_1 + \dots + b_k \leq |V(G)|$.
\end{proof}


%% file: fig_cycle_transformation.tex
\begin{figure*}
    \begin{subfigure}[t]{.47\textwidth}\centering
        \begin{tikzpicture}[scale=8/10,very thick,-{Stealth[length=2.4mm, width=1.2mm]},
            vertex/.style={
                draw,
                circle, 
                inner sep=.0pt, 
                minimum size=0.62cm},
                    rotate=-90]
                    \newcommand{\sidelength}{1.8}
                    \node[vertex] (v0) at (0*72:\sidelength) {$v_0$};
                    \node[vertex] (v1) at (1*72:\sidelength) {$v_1$};
                    \node[vertex] (v2) at (2*72:\sidelength) {$v_2$};
                    \node[vertex] (v3) at (3*72:\sidelength) {$v_3$};
                    \node[vertex] (v4) at (4*72:\sidelength) {$v_4$};

                    \draw (v0)--(v1);
                    \draw (v1)--(v2);
                    \draw (v2)--(v3);
                    \draw (v3)--(v4);
                    \draw (v4)--(v0);					
                    
                    \begin{scope}[lightgray]
                        \node[vertex] (u1) [left =of v4] {$e$};
                        \draw (u1)--(v4);
                        \node[vertex] (u2) [left =of v3] {$d$};
                        \draw (v3)--(u2);

                        \node[vertex] (u4) [above right = of v1] {$c$};
                        \draw (u4)--(v1);
                        \node[vertex] (u7) at (-32:2.2*\sidelength) {$a$};
                        \draw (u7)--(v0);
                        \node[vertex] (u9) at (32:2.2*\sidelength) {$b$};
                        \draw (v0)--(u9);
                    \end{scope}

                    \begin{scope}[opacity=0] 
                        \node[vertex] (i4) [left of = v0] {$i_4$};
                        \node[vertex] (i3) [left of = i4] {$i_3$};
                        \node[vertex] (i2) [left of = i3] {$i_2$};
                        \node[vertex] (i1) [left of = i2] {$i_1$};
                        \node[vertex] (o1) [right of = v0] {$o_1$};

                        \node[vertex] (c) [below of = i1] {$c$};
                        \draw (c) -- (i1);
                        \draw [->] (c) to [out=30,in=0] (o1);

                    \end{scope}
                \end{tikzpicture}	
                \caption{
                Example of a cycle of length $5$ in $G$. Incoming and outgoing edges are shown in grey.\medskip}
                \label{fig:cycleentranceA}
            \end{subfigure}
            \hfill
            \begin{subfigure}[t]{.49\textwidth}\centering
                \begin{tikzpicture}[scale=8/10,very thick,-{Stealth[length=2.4mm, width=1.2mm]},
                    vertex/.style={
                        draw,
                        circle, 
                        inner sep=.0pt, 
                        minimum size=0.62cm},
                            rotate=-90]
                            \newcommand{\sidelength}{1.8}
                
                            \node[vertex] (v0) at (0*72:\sidelength) {$v_0$};
                            \node[vertex] (v1) at (1*72:\sidelength) {$v_1$};
                            \node[vertex] (v2) at (2*72:\sidelength) {$v_2$};
                            \node[vertex] (v3) at (3*72:\sidelength) {$v_3$};
                            \node[vertex] (v4) at (4*72:\sidelength) {$v_4$};
                            

                            \node[vertex] (i4) [left of = v0] {$i_4$};
                            \node[vertex] (i3) [left of = i4] {$i_3$};
                            \node[vertex] (i2) [left of = i3] {$i_2$};
                            \node[vertex] (i1) [left of = i2] {$i_1$};

                            \node[vertex] (o1) [right of = v0] {$o_1$};
                            \node[vertex] (o2) [right of = o1] {$o_2$};
                            \node[vertex] (o3) [right of = o2] {$o_3$};


                            \draw (v0)--(v1);
                            \draw (v1)--(v2);
                            \draw (v2)--(v3);
                            \draw (v3)--(v4);
                            \draw (v4)--(v0);
                            
                            \draw (i4)--(v0);
                            \draw (i3)--(i4);
                            \draw (i2)--(i3);
                            \draw (i1)--(i2);
                            \draw (v0)--(o1);
                            \draw (o1)--(o2);
                            \draw (o2)--(o3);

                            \begin{scope}[lightgray,>={Stealth[length=2.4mm, width=1.2mm]}]
                                \node[vertex] (a) [below of = i4] {$a$};
                                \draw (a)--(v0);
                                \node[vertex] (b) [below of = v0] {$b$};
                                \draw (v0)--(b);

                                \node[vertex] (c) [below of = i1] {$c$};
                                \draw (c) -- (i1);
                                \draw [->] (c) to [out=30,in=0] (o1);

                                \node[vertex] (d) [below of = o3] {$d$};
                                \draw (o3) -- (d);

                                \node[vertex] (e) [below of = i3] {$e$};
                                \draw (e) -- (i4);

                            \end{scope}
                        \end{tikzpicture}	
                        \caption{The corresponding cycle  in $G'$.
                        The edges $(c, i_1)$, $(e, i_4)$ and $(a, v_0)$ are (default) entry edges, $(v_0, b)$ and $(o_3, d)$ are exit edges and $(c, o_1)$ is a bypass entry edge.
                        \medskip}
                    \label{fig:cycleentranceB}
                \end{subfigure}
                \caption{Transformation of a cycle described in the proof of \cref{lem:transformation}.}
                \label{fig:cycleentrance}
    \end{figure*}

%% file: fig_G_to_H.tex
\begin{figure*}
    \begin{subfigure}[t]{.3\textwidth}\centering
        \begin{tikzpicture}[scale=8/10,very thick,-{Stealth[length=2.4mm, width=1.2mm]},
            vertex/.style={
                draw,
                circle, 
                inner sep=.0pt, 
                minimum size=0.62cm},
                    rotate=-90]
                    \newcommand{\sidelength}{1.8}
                    \node[vertex] (v0) at (0*72:\sidelength) {$v_C$};
                    \node[vertex] (v1) at (4*72:\sidelength) {$v_1$};
                    \node[vertex] (v2) at (3*72:\sidelength) {$v_2$};
                    \node[vertex] (v3) at (2*72:\sidelength) {$v_3$};
                    \node[vertex] (v4) at (1*72:\sidelength) {$v_4$};

                    \draw (v0)--(v1);
                    \draw (v1)--(v2);
                    \draw (v2)--(v3);
                    \draw (v3)--(v4);
                    \draw (v4)--(v0);
                    
                    \node[vertex] (u) [left = of v0] {$u$};
                    \node[vertex] (w) [right = of v0] {$w$};
                    
                    \draw (u)--(v0);
                    \draw (v0)--(w);
        \end{tikzpicture}	
        \caption{A cycle $C$  in the graph $G$ of length $\ell = 5$ and cycle gate $v_C$ with arbitrary predecessor $u$ and successor $w$.
        \medskip}
    \end{subfigure}
    \hfill
    \begin{subfigure}[t]{.3\textwidth}\centering
        \begin{tikzpicture}[scale=8/10,very thick,-{Stealth[length=2.4mm, width=1.2mm]},
            vertex/.style={
                draw,
                circle, 
                inner sep=.0pt, 
                minimum size=0.62cm},
                    rotate=-90]
                    \newcommand{\sidelength}{1.8}
                    \node[vertex] (v0) at (0*72:\sidelength) {$v_C$};
                    \node[vertex] (v1) at (4*72:\sidelength) {$v_1$};
                    \node[vertex] (v2) at (3*72:\sidelength) {$v_2$};
                    \node[vertex] (v3) at (2*72:\sidelength) {$v_3$};
                    \node[vertex] (v4) at (1*72:\sidelength) {$v_4$};

                    \draw (v1)--(v2);
                    \draw (v2)--(v3);
                    \draw (v3)--(v4);
                    \draw (v4)--(v0);
                    
                    \node[vertex] (u) [left = of v0] {$u$};
                    \node[vertex] (w) [right = of v0] {$w$};
                    
                    \draw (u)--(v0);
                    \draw (u)--(v1);
                    \draw (v0)--(w);
        \end{tikzpicture}	
        \caption{In $H_1$, the cycle $C$ is replaced by a simple path from $u$ to $v_C$ of length $\ell + 1 = 6$. 
        \medskip}
    \end{subfigure}
    \hfill
    \begin{subfigure}[t]{.3\textwidth}\centering
        \begin{tikzpicture}[scale=8/10,very thick,-{Stealth[length=2.4mm, width=1.2mm]},
            vertex/.style={
                draw,
                circle, 
                inner sep=.0pt, 
                minimum size=0.62cm},
                    rotate=-90]
                    \newcommand{\sidelength}{1.8}
                    \node[vertex] (v0) at (0*72:\sidelength) {$v_C$};
                    \node[vertex] (v1) at (4*72:\sidelength) {$v_1$};
                    \node[vertex] (v2) at (3*72:\sidelength) {$v_2$};
                    \node[vertex] (v3) at (2*72:\sidelength) {$v_3$};
                    \node[vertex] (v4) at (1*72:\sidelength) {$v_4$};

                    \draw (v0)--(v1);
                    \draw (v1)--(v2);
                    \draw (v2)--(v3);
                    \draw (v3)--(v4);
                    
                    \node[vertex] (u) [left = of v0] {$u$};
                    \node[vertex] (w) [right = of v0] {$w$};
                    
                    \draw (u)--(v0);
                    \draw (v0)--(w);
                    \draw (v4)--(w);
        \end{tikzpicture}	
        \caption{In $H_2$, the cycle $C$ is replaced by a simple path from $v_C$ to $w$ of length $\ell + 1 = 6$. 
        \medskip}
    \end{subfigure}
    
    \caption{Construction of the DAGs $H_1$ and $H_2$ from the graph $G$ in \cref{thm:equivalence}. The construction ensures that every cycle gate has in-degree at least 2 in $H_1$ and out-degree at least 2 in $H_2$, and that every $st$-walk in $H_1$ or $H_2$ corresponds to a unique $st$-walk in $G$ of the same length. Indeed, for $i \in \{1, 2\}$, any $st$-walk in $H_i$ that uses the simple path from $u$ to $v_C$ (or from $v_C$ to $w$) of length $|V(C)| + 1$, is mapped to the $st$-walk in $G$ that loops around the cycle $C$ once, which has the same length.}
    \label{fig:G_to_H}
\end{figure*}

%% file: twins.tex
In this section, we give a linear time algorithm for deciding \determinisability in the unary case; i.e.\ we prove \cref{thm:WFA_unary}. Note that here linear time is in the size of the automaton, which is both the number of states and transitions. It is not hard to see that an equivalent of \cref{lem:sparsity} is not possible and transitions need to be taken into account.

We use the equivalence between \determinisability and \twins (see~\cref{lem:WFA_equivalence}).
Recall from~\cref{def:twins} that an automaton has the twins property if and only if all pairs of siblings
are twins. 
We can assume without loss of generality that the given automata is trim and in normal form (see~\cref{claim:trim,claim:normal}) with source vertex $s$ and target vertex $t$. Hence, a pair of non-twin siblings in $\mc W$ is equivalent to the existence of distinct vertices $p, q$ in the underlying weighted $st$-graph that are both reachable from $s$ by walks of the same length and cycles at $p$ and $q$, respectively, of the same length but different weight. We thus reformulate the \twins problem for unary alphabets as follows.
\begin{description}
    \item[\unarytwins] Given a trim directed weighted $st$-graph ${G = (V, E, s, t, \omega)}$, are all sibling pairs also twins?
\end{description}



For any instance $G = (V, E, s, t, \omega)$, we denote the number of vertices and edges by $n = |V|$ and $m = |E|$. 
%
We first reduce the problem of finding a pair of non-twin siblings to the problem of finding two cycles of different average weight. The average weight of a cycle is the total weight divided by its number of vertices.

\begin{lemma}\label{lem:twin-if-same-average-weight}
    Let $G$ be a trim weighted $st$-graph.
    Then there is a pair of non-twin siblings in $G$ if and only if $G$ has two cycles of different average weight.
\end{lemma}
\begin{proof}
    First, assume that $p$ and $q$ are non-twin siblings, i.e\ there are a $p$-cycle $\pi_p$ and a $q$-cycle $\pi_q$ of the same length $\ell$ but distinct weights $\weight(\pi_p) \neq \weight(\pi_q)$.
    Then they also have distinct average weights   ${\weight(\pi_p)}/{\ell}\neq  {\weight(\pi_q)}/{\ell}$. 

    Now assume that $\pi_1$ and $\pi_2$ are two cycles of different average weight in $G$. We can assume that $|\pi_1| = |\pi_2|$; otherwise, we can obtain two cycles of the same length by repeating $|\pi_1|$ times the cycle $\pi_2$ and $|\pi_2|$ times the cycle $\pi_1$ without changing the average weights.
    As $G$ is trim, both $\pi_1$ and $\pi_2$ are reachable from $s$.
    Let $\rho_1$ and $\rho_2$ be the shortest walk from $s$ to $\pi_1$ and $\pi_2$, respectively, and assume without loss of generality that $|\rho_1| \le |\rho_2|$.
    Let $p_1$ and $p_2$ be the end vertices of $\rho_1$ and $\rho_2$, which are in $\pi_1$ and $\pi_2$, respectively.
    Let $\rho$ be the walk of length $|\rho_2| - |\rho_1|$ starting at $p_1$ and going along the cycle $\pi_1$, possibly multiple times. 
    Let $q$ be the end vertex of $\rho$ and let $\pi_1'$ be the cycle $\pi_1$
    re-attached to start and end at $q$.
    Then $\rho_1' \coloneqq \rho_1 \rho$ and $\rho_2$ are two walks of the same length from $s$ to $q$ and $p_2$, respectively, and $\pi_1'$ and $\pi_2$ are cycles of equal length but different average weights at $q$ and $p_2$.
    Thus, $q$ and $p_2$ are non-twin siblings.
\end{proof}

We now reduce detecting two cycles of different average weight to detecting a cycle of nonzero weight.
\begin{lemma}\label{lem:shifted-average-weight}
    Given a directed weighted graph  $G = (V, E, \weight)$, 
    we can find in $\Oh(|V| + |E|)$ time a weight function $\weight'\colon E \rightarrow \mathbb{Z}$ %
    such that $G$ contains two cycles of different average weight if and only if $G' = (V, E, \weight')$ contains a cycle of nonzero weight.
\end{lemma}


\begin{proof}
    Using a depth-first search algorithm we can, in $\mathcal{O}(|V| + |E|)$ time, find a cycle $\pi$ in $G$ or conclude that $G$ is acyclic. We assume the former since in the latter case any $\weight'$ will do.
    Let $a$ denote the average weight of $\pi$.
    Now set $\weight'(e) = \weight(e) - a$ for each $e \in E$. Observe that the average weight of each cycle is smaller in $G'$ than in $G$ by $a$. In particular, the average weight of $\pi$ in $G'$ is $0$.

    Suppose that $G$ contains two cycles $\pi_1$ and $\pi_2$ of different average weights in $G$. At least one of them, say $\pi_1$, has an average weight in $G$ different from $a$, i.e.\ $\weight(\pi_1) / |\pi_1| \neq a$. Then $\weight'(\pi_1) = \weight(\pi_1) - |\pi_1| \cdot a \neq \weight(\pi_1) - |\pi_1| \cdot (\weight(\pi_1) / |\pi_1|) = 0$, so $\pi_1$ is a nonzero weight cycle in $G'$.

    Now suppose that $G'$ contains a nonzero weight cycle $\pi_1$. Then the average weight in $G$ of $\pi_1$ is ${\weight'(\pi_1)}/{|\pi_1|} + a \neq a$,
    so $\pi$ and $\pi_1$ are cycles of different average weight in $G$.
\end{proof}

We have reduced the problem of deciding the twins property for $\mc W$ to detecting nonzero cycles in $G'$.
Observe that any cycle is contained in a single strongly connected component. We can thus decompose
$G'$ into its strongly connected components in $\mathcal{O}(n + m)$ time and proceed for
each component independently.
The following lemma concludes the proof of \cref{thm:WFA_unary}.
\begin{lemma}\label{lem:potential}
    Let $G = (V, E, \weight)$ be a strongly connected directed weighted graph.
    Every cycle in $G$ has weight zero if and only if 
    there is a \emph{potential function} $\varphi\colon V \rightarrow \mathbb{Z}$ such that
    $\varphi(v) - \varphi(u) = \weight((u,v))$ for each $(u,v) \in E$.
    Moreover, there is an algorithm that can compute such a $\varphi$ or detect a nonzero cycle in time $\Oh(n+m)$.
\end{lemma}
\begin{proof}
    First, assume that such a potential function $\varphi$ exists. Then each cycle $\pi = ((v_0,v_1),
    \ldots, (v_{\ell - 1},v_\ell))$ with $v_\ell=v_0$ has indeed weight zero:
    \begin{displaymath}
        \weight(\pi) = \sum_{i = 0}^{\ell - 1} \weight((v_i,v_{i+1})) = \sum_{i = 0}^{\ell - 1} (\varphi(v_{i+1}) - \varphi(v_i)) =\varphi(v_\ell) - \varphi(v_0) = 0.
    \end{displaymath}

Conversely, suppose that all cycles have weight zero.
We define a procedure that yields a potential function $\varphi : V \rightarrow \mathbb Z$.
Let $s$ be an arbitrary vertex and set $\varphi(s) \coloneqq 0$. Next, run a
depth-first search algorithm (DFS) that starts from $s$ and that, for every vertex
$v$ visited for the first time after traversing the edge $(u, v)$, sets
$\varphi(v) \coloneqq \varphi(u) + \weight((u,v))$. Then, for each visited vertex
$v$, the potential $\varphi(v)$ is the weight of a walk from $s$ to $v$.

Clearly, this algorithm runs in $\Oh(n+m)$ time. 
We argue that $\varphi(v) - \varphi(u) = \weight((u,v))$ for each $(u,v) \in E$. 
For the sake of contradiction, assume that there is an edge $(u, v) \in E$ with
$\varphi(v) \neq \varphi(u) + \weight((u, v))$. 
Then there are two $sv$-walks $\rho_1$ and $\rho_2$ of different weights
$\varphi(v)$ and $\varphi(u) + \weight((u, v))$ respectively (the $sv$-walk visited by the DFS and the $su$-walk visited by the DFS followed by the $(u, v)$ edge). Since $G$ is strongly connected, there is a walk $\rho$ from $v$ to $s$. This means that $\sigma_1 = \rho_1\rho$ and $\sigma_2 = \rho_2\rho$ are cycles with distinct weights. Hence, at least one of these weights must be nonzero, contradicting the assumption.

Finally, observe that in $\Oh(n + m)$ time we can run the above algorithm and
check whether $\varphi(v) = \varphi(u) + \weight((u,v))$ for every edge $(u, v)$. 
\end{proof}


%% file: lower_bound.tex
We now prove the lower bounds stated in \cref{thm:NFA_lowerbounds,thm:WFA_lowerbound}. 
As already sketched in \cref{sec:technical_overview_lowerbounds}, we use the intermediate problem \kie defined in \cref{def:kie} that asks whether the $k$ given DFAs intersect. 

In \cref{sec:LB_kie}, we show hardness of \kie based on both the $k$-OV hypothesis or the $k$-Cycle hypothesis. In \cref{sec:LB_from_kie}, we reduce \twoie to \unambiguity, which in turn reduces to both \PFA and \determinisability. We also reduce \threeie to \PFA. \cref{fig:graph_reduction} summarises our reductions.

\subsection{Lower bounds for \kie}\label{sec:LB_kie}
\input{lowerbounds_kie}

\subsection{Reductions from \kie}\label{sec:LB_from_kie}
Using the equivalences between \FFA, \PFA and \determinisability to respectively \ida, \eda and \twins (see~\cref{lem:NFA_equivalence,lem:WFA_equivalence}), we reduce \kie to the studied problems. 
More precisely, \cref{thm:twoie_unambiguous} reduces \twoie to \unambiguity and implies that for any $\eps > 0$, if we could decide \unambiguity in $n^{2 - \eps}$ time, then we could decide \twoie in $m^{2 - \eps}$ time, where $n$ is the size of the \unambiguity NFA and $m$ is the size of the \twoie DFAs.
Next, \cref{thm:unambiguity_eda} reduces \unambiguity to \eda. This implies that for any $\eps > 0$, if we could decide \eda in $n^{2 - \eps}$ time, then we could decide \unambiguity in $m^{2 - \eps}$ time, where $n$ is the size of the \eda NFA and $m$ is the size of the \unambiguity NFA.
Then, we reduce \threeie to \ida in \cref{thm:3dfa_ida}, which implies that for any $\eps > 0$, if we could decide \twins in $n^{3 - \eps}$ time, then we could decide \threeie in $m^{3 - \eps}$ time, where $n$ is the size of the \twins NFA and $m$ is the size of the \threeie DFAs.
Finally, \cref{thm:unambiguity_twins} shows a linear time reduction from \unambiguity to \twins. This implies that for any $\eps >0$, if we could decide \twins in $n^{2-\eps}$ time, then we could decide \unambiguity in $m^{2-\eps}$ time, where $n$ is the size of the \twins weighted automaton and $m$ is the size of the \unambiguity NFA. 

By the hardness of \kie proven above this implies that, unless both the $k$-Cycle hypothesis and the $k$-Cycle hypothesis fail, \unambiguity, \PFA and \determinisability cannot be decided in $\Oh(|\mc A|^{2 - \eps})$ time, for any $\eps >0$. Similarly, unless the $k$-OV hypothesis fails, \FFA cannot be decided in $\Oh(|\mc A|^{3- \eps})$ time, for any $\eps >0$.

\begin{lemma}[\twoie reduces to \unambiguity]\label{thm:twoie_unambiguous}
	Suppose that, given an NFA $\mc A = (Q, \Sigma, \delta, I, F)$,  one can decide whether $\mc A \in \unambiguity$ in time $T(|Q|, |\delta|, |\Sigma|)$. Then, given two DFAs $\mc D_1$ and $\mc D_2$, one can decide whether $(\mc D_1, \mc D_2) \in \twoie$ in time $\Oh(T(n, m, s))$, where $n$ is the total number of states, $m$ is the total number of transitions, and $s$ is the total alphabet size of $\mc D_1 \cup \mc D_2$. 
\end{lemma}
\begin{proof}
	Let $\mc D_1 = (Q_1, \Sigma_1, \delta_1, \{q^I_1\}, F_1)$ and $\mc D_2 = (Q_2, \Sigma_2, \delta_2, \{q^I_2\}, F_2)$ be two DFAs. Take their union and obtain the NFA $\mc A = (Q, \Sigma, \delta, I, F)$ with $Q = Q_1 \cup Q_2$, $\Sigma = \Sigma_1 \cup \Sigma_2$, $\delta = \delta_1 \cup \delta_2$, $I = \{q^I_1, q^I_2\}$ and $F = F_1 \cup F_2$. Then any accepting run in $\mc A$ is either an accepting run in $\mc D_1$ or an accepting run in $\mc D_2$. Since $\mc D_1$ and $\mc D_2$ are deterministic, a word $w \in \Sigma^*$ has two accepting runs in $\mc A$ if and only if it has an accepting run in $\mc D_1$ and an accepting run in $\mc D_2$. In other words, $\mc A$ is unambiguous if and only if there is no $w \in \mc L(\mc D_1) \cap \mc L(\mc D_2)$. 
\end{proof}
\begin{lemma}[\unambiguity reduces to \textsc{\eda}]\label{thm:unambiguity_eda}
    Suppose that, given an NFA $\mc A' = (Q', \Sigma', \delta', I', F')$, one can decide whether $\mc A' \in \eda$ in time $T(|Q'|, |\delta'|, |\Sigma'|)$. Then, given an NFA $\mc A = (Q, \Sigma, \delta, I, F)$, one can decide whether $\mc A \in \unambiguity$ in time
    $\Oh(T(|Q|, |\delta| + 1, |\Sigma| + 1))$.
\end{lemma}
\begin{proof}
Let $\mathcal{A}$ be an NFA for which we need to decide whether it is unambiguous.
Assume without loss of generality that $\mathcal{A}$ is trim and in normal form (see Claims \ref{claim:trim} and \ref{claim:normal}).
%
%
We construct $\mc A'$ from $\mc A$ as follows.
Extend the alphabet of $\mathcal{A}$ by adding a new symbol $\#$ and add a transition from the final state $q_F$ to the initial state $q_I$ reading $\#$. 
We show that $\mc A'$ has EDA if and only if $\mc A$ is not unambiguous.

$(\impliedby)$
Assume that $\mathcal{A}$ is not unambiguous, that is, there is a word $w\in \Sigma^*$ with two distinct accepting runs from $q_I$ to $q_F$. Adding the transition $(q_F,\#,q_I)$ closes these two runs to distinct $q_I$-cycles reading the word $w\#$. This shows that $\mathcal{A}'$ has \eda.

$(\implies)$
Assume now that $\mathcal{A}'$ has \eda. Then there is a state $q$ with two distinct cycles on the same word $v$.
If the word contains the letter $\#$, then without loss of generality we can assume that $q = q_I$ (otherwise we could shift the cycles on $v = u\#w$ at $q$ to get cycles on $wu\#$ at $q_I$).
Then $v = v_1\#v_2\#\ldots \#v_k\#$, for some words $v_1, \ldots v_k \in \Sigma^*$.
Let $i$ be such that the first transition where the cycles differ is taken after reading $v_1\#v_2\#\ldots \#v_i'$, where $v_i'$ is a prefix of $v_i$.
Since after reading $v_1\#v_2\#\ldots \#v_i$ both cycles reach the transition $(q_F, \#, q_I)$, $v_i$ has two different accepting runs in $\mc A$, i.e. $\mc A$ is not unambiguous.
Assume now that the word does not contain the letter $\#$. Since $q$ is reachable and co-reachable, there are words $u,w\in\Sigma^\ast$ such that there is a run from $q_I$ to $q$ on $u$ and one from $q$ to $q_F$ on $w$.
Then the word $uvw$ has at least two distinct accepting runs that are identical on the prefix $u$ and the suffix $w$ but differ on the middle part $v$. 
This shows that $\mathcal{A}$ is not unambiguous.
\end{proof}


\begin{lemma}[\threeie reduces to \ida]\label{thm:3dfa_ida}
    Suppose that, given an NFA $\mc A = (Q, \Sigma, \delta, I, F)$, one can decide whether $\mc A \in \ida$ in time $\Oh(T(|Q|, |\delta|, |\Sigma|))$. Then, given three DFAs $\mc D_1$, $\mc D_2$, and $\mc D_3$, one can decide whether $(\mc D_1, \mc D_2, \mc D_3) \in \threeie$ in time $\Oh(T(n + 2, m + n + 3, s + 2))$, where $n$ is the total number of states, $m$ is the total number of transitions, and $s$ is the total alphabet size of $\mc D_1 \cup \mc D_2 \cup \mc D_3$.
\end{lemma}
\begin{proof}
    Let $\mc D_1 = (Q_1, \Sigma, \delta_1, \{q^I_1\}, F_1)$, $\mc D_2 = (Q_2, \Sigma, \delta_2, \{q^I_2\}, F_2)$ and $\mc D_3  = (Q_3, \Sigma, \delta_3, \{q^I_3\}, F_3)$ be three DFAs on the same alphabet $\Sigma$. 
    We show how to construct in linear time an NFA $\mc A$ that has the property \ida if and only if $\mc L(\mc D_1) \cap \mc L(\mc D_2) \cap \mc L(\mc D_3)$ is non-empty.  

    Let $s$ and $t$ be new states, and let $\$$ and $\# $ be extra symbols. Define the NFA $\mc A = (Q, \Sigma \cup \{\$, \#\}, \delta, \{s\}, \{t\})$, where the states $Q$ are defined by:
    \begin{align*}
        Q_1' &\coloneqq Q_1 \cup \{s\},\quad 
        Q_2' \coloneqq Q_2,\quad
        Q_3' \coloneqq Q_3 \cup \{t\} \\
        Q &\coloneqq Q_1' \uplus Q_2' \uplus Q_3' \;\;
    \intertext{and transitions $\delta$ are defined by:}
        \delta_1' &\coloneqq \{(s, \$, q^I_1)\} \cup \delta_1 \cup \{(p, \#, s)\mid p \in F_1\} \\
        \delta_2' &\coloneqq  \{(s, \$, q^I_2)\} \cup \delta_2 \cup \{(p, \#, t)\mid p \in F_2 \} \\
        \delta_3' &\coloneqq  \{(t, \$, q^I_3)\} \cup \delta_3 \cup \{(p, \#, t)\mid p \in F_3 \} \\
        \delta &\coloneqq \delta_1' \uplus \delta_2' \uplus \delta_3'
    \end{align*}
	\Cref{fig:3dfa_IDA} illustrates this construction. We show that $\mc L(\mc D_1) \cap \mc L(\mc D_2) \cap \mc L(\mc D_3) \neq \emptyset$ if and only if $\mc A$ has \ida.
	\begin{figure}
		\centering
			\includesvg[width=0.4\columnwidth]{hardness_fig/ida_box}
		\caption{
		Example of the reduction in \cref{thm:3dfa_ida} from \threeie to \ida. 
		The three DFAs $\mc D_1$, $\mc D_2$ and $\mc D_3$ (with grey background, from top to bottom) are aggregated into the NFA $\mc A$ using transitions reading $\$$ and $\#$. 
		The word $aa$ is recognized by all three DFAs. Correspondingly, $\mc A$ has the property \ida with states $s$ and $t$: the cycle 
		$s \reach[\$] q^I_1 \reach[a] y \reach[a] z \reach[\#] s$, 
		the run 
		$s \reach[\$] q^I_2 \reach[a] x' \reach[a] y' \reach[\#] t$, 
		and the cycle 
		$t \reach[\$] q^I_3 \reach[a] x'' \reach[a] y'' \reach[\#] t$ 
		all read the word $\$aa\#$.
        }\label{fig:3dfa_IDA}
	\end{figure}
	
    $(\implies)$
    Suppose that there is a word $w \in \Sigma^*$ recognized by $\mc D_1$, $\mc D_2$, and $\mc D_3$ via the runs $\rho_1 \subset \delta_1$, $\rho_2 \subset \delta_2$ and $\rho_3 \subset \delta_2$
    ending in $p_1 \in F_1, \ p_2 \in F_2, \ p_3 \in F_3$, respectively. 
    %
    Then we can extend them to the runs 
    $\pi_1 \coloneqq (s, \$, q_1^I) \rho_1 (p_1, \#, s)$,
    $\pi_2 \coloneqq (s, \$, q_2^I) \rho_2 (p_2, \#, t)$, and
    $\pi_3 \coloneqq (t, \$, q_3^I) \rho_3 (p_3, \#, t)$
    in $\mc A$ 
    Observe that $\pi_1$ is an $s$-cycle, $\pi_2$ is a run from $s$ to $t$, and $\pi_3$ is a $t$-cycle, and they all read the same word $\$ w \#$. Thus, $\mc A$ has the \ida property.
    
    $(\impliedby)$
    Conversely, assume that $\mc A$ has the \ida property, i.e.\ there are distinct states $p, q \in Q$, a $p$-cycle $\pi_1$, a run $\pi_2$ from $p$ to $q$, and a $q$-cycle $\pi_3$ that read the same input.
    
    Since $\pi_1$ and $\pi_2$ are identically labelled but distinct, they branch at a vertex with two identically labelled outgoing transitions. The only such vertex is $s$ with only two outgoing transitions: one to $q_1^I \in Q_1'$ and one to $q_2^I \notin Q_1'$, both reading $\$$.
    Since $s$ is reachable only from $Q_1'$, the cycle $\pi_1$ is contained in $Q_1'$; thus $\pi_2$ transitions from $s$ to $Q\setminus Q_1'$ and then stays there. 
    The cycle $\pi_3$ starts where $\pi_2$ ends (in $Q\setminus Q_1'$) and also reads $\$$. The only transition labelled $\$$ starting in $Q\setminus Q_1'$ is from $t$ to $q_3^I \in Q_3'$. Since $t$ is not reachable from $Q\setminus Q_3'$, $\pi_3$ is contained in $Q_3'$. 
    We infer that $\pi_2$, which visits $s$, ends inside $Q_3'$ and thus contains an $st$-subrun labelled $\$w\#$ for some $w\in\Sigma^\ast$. Hence, $\pi_1$ and $\pi_3$ must contain identically labelled subruns. Removing from each subrun the first and last transition (labelled $\$$ and $\#$) yields an accepting run on $w$ for each DFA. 
\end{proof}

\begin{lemma}[\unambiguity reduces to \twins]\label{thm:unambiguity_twins}
    Suppose that, given a weighted automaton $\mc W = (Q', \Sigma', \delta', I', F', \weight)$, one can decide whether $\mc W \in \twins$ in time $T(|Q'|, |\delta'|, |\Sigma'|)$. Then given an NFA $\mc A = (Q, \Sigma, \delta, I, F)$, one can decide whether $\mc A \in \unambiguity$ in time $\Oh(T(3 \cdot |Q|, |Q| + 3 \cdot |\delta| + 1, |\Sigma| + 2))$.
\end{lemma}
\begin{proof}
    Let $\mc A = (Q, \Sigma, \delta, I, F)$ be an NFA. 
    Without loss of generality assume that there is a unique initial state $q_I$ and a unique accepting state $q_F$ (see \cref{claim:normal}). 
    To prove the statement, we build a weighted automaton $\mc W \coloneqq (Q', \Sigma', \delta', I', F', \weight)$ with $|Q'| = \Oh(|Q|)$, $|\Sigma'| = \Oh(|\Sigma|)$ and $|\delta'| = \Oh(\delta)$ 
    such that $\mc W$ has the twins property if and only if $\mc A$ is unambiguous. 
    Construct three copies of the states in $\mc A$ denoted $Q_1 \coloneqq \{q^{(1)} \mid q \in Q\}$, $Q_2 \coloneqq \{q^{(2)} \mid q \in Q\}$, and $Q_3 \coloneqq \{q^{(3)} \mid q \in Q\}$,
    and two copies of the transitions of $\mc A$ where we set the output weight to $0$:
    \begin{align*}
        \delta_1 &\coloneqq \{(p^{(1)}, \sigma, q^{(1)}) \mid (p, \sigma, q) \in \delta\}, \\
        \delta_3 &\coloneqq \{(p^{(3)}, \sigma, q^{(3)}) \mid (p, \sigma, q) \in \delta\}, \\
        \weight(t) &\coloneq 0 \text{ for each } t \in \delta_1 \cup \delta_3.
    \intertext{
        Let $\$$ and $\#$ be symbols that do not appear in $\Sigma$.
        Build transitions reading $\$$ of weight 0 that go from states in $Q_1$ to their corresponding copy in $Q_2$:
    }
        \delta_{\$} &\coloneqq \{(p^{(1)}, \$, p^{(2)}) \mid p \in Q\}, \\
        \weight(t) &\coloneq 0 \text{ for } t \in \delta_\$.        
    \intertext{
        Now, for every state $p \in Q$ and symbol $\sigma \in \Sigma$, fix an arbitrary order on the states in $Q$ that are reachable from $p$ by reading $\sigma$. Let $\ell(p, \sigma, q) \geq 1$ be the index of the state $q$ in that ordering. We construct transitions from states in $Q_2$ to states in $Q_3$ as follows: 
    }
        \delta_2 &\coloneqq \{(p^{(2)}, \sigma, q^{(3)}) \mid (p, \sigma, q) \in \delta\}, \\
        \weight(p^{(2)}, \sigma, q^{(3)}) &\coloneqq \ell(p, \sigma, q) \text{ for each } (p, \sigma, q) \in \delta.
    \intertext{Finally, we add a transition $t^\#$ of weight 0 that connects $Q_3$ to $Q_1$:}
    t^{\#} &\coloneqq (q^{(3)}_F, \#, q^{(1)}_I),\\
    \weight(t^\#) &\coloneq 0.
    \end{align*}
    Consider the weighted automaton $\mc W \coloneqq (Q', \Sigma \cup \{\$, \#\}, \delta', \{q_I^{(1)}\}, \{q_I^{(1)}\}, \weight)$ with $Q' \coloneqq Q_1 \cup Q_2 \cup Q_3$ and $\delta' \coloneqq \delta_1 \cup \delta_{\$} \cup \delta_2 \cup \delta_3 \cup \{t^{\#}\}$. \cref{fig:unambiguity_twins} illustrates this construction.
	\begin{figure}[h]
		\centering
		\includesvg[width=0.5\columnwidth]{hardness_fig/TP.svg}
		\caption{Example of the reduction in \cref{thm:unambiguity_twins} from \unambiguity to \twins. The NFA $\mc A$ (top) is reduced to the weighted automaton $\mc W$ (bottom). There are two accepting runs in $\mc A$ reading $w= aba$: $q_I \reach[a] x \reach[b] y \reach[a] q_F$ and $q_I \reach[a] x \reach[b] q_F \reach[a] q_F$. In $\mc W$, the corresponding cycles 
        $y^{(3)} \reach[a] q_F^{(3)} \reach[\#] q_I^{(1)} \reach[a] x^{(1)} \reach[\$] x^{(2)} \reach[b] y^{(3)}$ and 
        $q_F^{(3)} \reach[a] q_F^{(3)} \reach[\#] q_I^{(1)} \reach[a] x^{(1)} \reach[\$] x^{(2)} \reach[b] q_F^{(3)}$
        read the same input $a\#a\$b$ but have weight 1 and 2 respectively.}\label{fig:unambiguity_twins}
	\end{figure}
    We show that $\mc A$ is unambiguous if and only if $\mc W$ has the twins property.

    $(\impliedby)$
    Assume that $\mc A$ is not unambiguous, i.e.\ there are two distinct runs $\rho$ and $\rho'$ from $q_I$ to $q_F$ on the same word $v \in \Sigma^*$.
    Let $(p, \sigma, q)$ and $(p, \sigma, q')$ be the first transitions where $\rho$ and $\rho'$ differ
    and $v = u \sigma u'$ where $u$ is the maximal prefix of $v$ on which $\rho$ and $\rho'$ are the same.
    Then $(q^{(3)}, {q'}^{(3)})$ is a pair of non-twin siblings in $\mc W$.
    Indeed, let $i = \ell(p, \sigma, q)$ and $j = \ell(p, \sigma, q')$.
    Both $q^{(3)}$ and ${q'}^{(3)}$ are reachable from $q_I^{(1)}$ reading the same word $u \$ \sigma$, by following the behaviour of $\rho$ in the copy $\delta_1$ up to $p^{(1)}$, then using the transition $(p^{(1)}, \$, p^{(2)}) \in \delta_\$$, and then using either $(p^{(2)}, \sigma, q^{(3)}) \in \delta_2$ or $(p^{(2)}, \sigma, {q'}^{(3)}) \in \delta_2$.
	Consider the cycle $\pi$ at $q^{(3)}$ that goes to $q_F^{(3)}$ by following the behaviour of $\rho$ in the copy $\delta_3$, then uses the transition $t^\#$, and then goes from $q_I^{(1)}$ to $q^{(3)}$ using the same transitions as described above.
    Let $\pi'$ be the cycle at ${q'}^{(3)}$ constructed the same way using $\rho'$.
    Both cycles $\pi$ and $\pi'$ read the word $u' \# u \$ \sigma$.
    However, $\weight(\pi) = i$ while $\weight(\pi') = j$, and $i \neq j$ since $q \neq q'$. Thus, $q^{(3)}$ and ${q'}^{(3)}$ are indeed siblings but not twins. 
    
   
    $(\implies)$
    Conversely, assume that $\mc W$ admits a pair of siblings $(s, s')$ that are not twins, i.e.\ there exist two cycles $\pi$ and $\pi'$ at $s$ and $s'$, respectively, that read the same input but have different weights.
    Observe that all transitions of $\mc W$ except for the ones in $\delta_2$ have weight 0.
    Since $\pi$ and $\pi'$ differ in weight, at least one of them uses a transition of $\delta_2$, and thus goes through a state in $Q_2$.
    By construction, any cycle involving a state of $Q_2$ uses the transition $t^\#$.
    Since this is the only transition of $\mc W$ that reads the symbol $\#$, and since $\pi$ and $\pi'$ read the same input, we conclude that both $\pi$ and $\pi'$ use the transition $t^\#$. In particular, they both visit the states $q_I^{(1)}$ and $q_F^{(3)}$. Shift $\pi$ and $\pi'$ to obtain distinct cycles $\rho$ and $\rho'$ at $q_I^{(1)}$. Note that $\rho$ and $\rho'$ read the same input (the shifted input of $\pi$ and $\pi'$) but have different weights (the same weights as $\pi$ and $\pi'$, respectively).
    
    We now divide $\rho$ and $\rho'$ into subruns from $q_I^{(1)}$ to $q_F^{(3)}$ by cutting the cycles at transitions $t^\#$.
    Observe that, since $t^\#$ is the only transition reading $\#$, it must hold that the $i$-th part of $\rho$ reads the same input as the $i$-th part of $\rho'$.
    Let $\lambda$ and $\lambda'$ be the first pair of corresponding parts of $\rho$ and $\rho'$ such that $\weight(\lambda) \neq \weight(\lambda')$. Let $v$ be the word read by $\lambda$ and $\lambda'$.
    By the construction of $\mc W$, we have $\lambda = \lambda_1 (p^{(1)}, \$, p^{(2)}) (p^{(2)}, \sigma, q^{(3)}) \lambda_3$ for some $p, q \in Q, \sigma \in \Sigma$, and runs $\lambda_1 \subseteq \delta_1$ and $\lambda_3 \subseteq \delta_3$.
    (We write $\rho_0 \subset \delta_0$ for a run $\rho_0$ and a transition set $\delta_0$ if $t\in \delta_0$ for each transition $t$ in the transition sequence $\rho$.)
    Similarly, we have $\lambda' = \lambda'_1 (p'^{(1)}, \$, p'^{(2)}) (p'^{(2)}, \sigma', q'^{(3)}) \lambda'_3$ for some $p', q' \in Q, \sigma' \in \Sigma$ and runs $\lambda'_1 \subseteq \delta_1$ and $\lambda'_3 \subseteq \delta_3$. 
    Furthermore, since $\lambda$ and $\lambda'$ read the same word $v$ and none of the transitions from $\delta_1$ and $\delta_3$ read $\$$, we have
    $\sigma = \sigma'$, and the runs $\lambda_1$ and $\lambda'_1$, and $\lambda_3$ and $\lambda_3'$, read the same word, respectively.

    We consider two cases.
    If $p \neq p'$, then consider the two runs in $\mc A$ that reproduce the behaviour of $\lambda$ (respectively $\lambda'$) in the corresponding states of $\mc A$ by omitting the transitions $(p^{(1)}, \$, p^{(2)})(p^{(2)}, \sigma, q^{(3)})$ (resp. $({p'}^{(1)}, \$, {p'}^{(2)})$, $({p'}^{(2)}, \sigma', {q'}^{(3)})$) and replacing them by the transition $(p, \sigma, q) \in \delta$ (resp. $(p', \sigma', q') \in \delta$).
    This results in two runs in $\mc A$ from $q_I$ to $q_F$ that read the word $v$ where the symbol $\$$ is skipped. In other words, $\mc A$ is not unambiguous.
    Suppose now that $p = p'$. Note that, by construction, $\weight(\lambda) = \ell(p, \sigma, q)$ and $\weight(\lambda') = \ell(p, \sigma, q')$. Since $\weight(\lambda) \neq \weight(\lambda')$, we have $q \neq q'$. Similarly to the previous case, we can construct two runs in $\mc A$ from $q_I$ to $q_F$ that read the word $w$ were the symbol $\$$ is skipped. Hence, $\mc A$ is not unambiguous in this case as well.
\end{proof}

%% file: lowerbounds_kie.tex
We show hardness of \kie.
In \cref{thm:kcycle_2dfa}, we show that \twoie cannot be decided in $\Oh(|\mc A|^{2 - \eps})$ time for any $\eps >0$ unless the $k$-Cycle hypothesis fails. In \cref{thm:ov_twoie}, we show, for any $k\ge2$, that \kie cannot be decided in $\Oh(|\mc A|^{k - \eps})$ time for any $\eps >0$, unless the $k$-OV hypothesis fails.



We remark that by binary encoding every symbol of the alphabet $\Sigma$, we can assume that the alphabet of the DFAs in the considered \kie problems is binary. This only incurs a $\Oh(\log |\Sigma|)$ multiplicative factor in the instance size. 
We defer the efficient constructions to \cref{sec:appendix_lowerbounds} and focus on showing how such gadgets are exploited in the reductions from $k$-Cycle and $k$-OV.

\begin{claim}[Binary encoding of transitions]\label{def:binary_gadget_tree}
    In almost linear time, we can transform an instance of \kie over any given alphabet $\Sigma$ into an equivalent instance over the binary alphabet. 
    This transformation increases the size of the instance by a multiplicative factor of at most $\Oh(\log{|\Sigma|})$.
\end{claim}

\begin{claim}[Wildcard transitions]\label{def:binary_gadget_identity}
    A \emph{wildcard transition} between states $p$ and $q$ over the alphabet $\Sigma$ is a transition, denoted by $(p, \Sigma, q)$, that allows to transition from $p$ to $q$ by reading any symbol of the alphabet (instead of being restricted to a single letter). 

    An \emph{almost wildcard transition} between states $p$ and $q$ over the alphabet $\Sigma$ is a transition, denoted by $(p, \Sigma \setminus \{\sigma\}, q)$, that allows to transition from $p$ to $q$ by reading any letter of the alphabet except the letter $\sigma$, for some $\sigma \in \Sigma$.

    In almost linear time, we can transform an instance of \kie over the alphabet $\Sigma$ with \emph{(almost) wildcard transitions} over $\Sigma$ into an equivalent instance over the binary alphabet. The transformation increases the size of the instance by a multiplicative factor of at most $\Oh(\log
   |\Sigma|)$.
\end{claim}

Note that one could easily implement a wildcard transition $(p, \Sigma, q)$ with $|\Sigma|$ regular transitions, specifically $(p, \sigma, q)$ for each $\sigma \in \Sigma$, and then apply \cref{def:binary_gadget_tree} to reduce the alphabet size. However, this increases the automaton size by $|\Sigma|$. \cref{def:binary_gadget_identity} shows how to achieve this while increasing the automaton size by only a $\Oh(\log{|\Sigma|})$ factor.

The following is a quadratic lower bound for \twoie assuming the $k$-Cycle hypothesis.

\begin{lemma}[$k$-\textsc{Cycle} reduces to \twoie]\label{thm:kcycle_2dfa}
    For every $\eps > 0$, there is no $\Oh(m^{2- \eps})$ time algorithm that
    decides whether the intersection of two DFAs of size $m$ is non-empty, assuming the $k$-Cycle hypothesis.
\end{lemma}
\begin{proof}
Let $k \ge 3$ be a fixed constant that depends on $\eps$ and consider a graph $G = (V, E)$.
We can assume without loss of generality that $G$ is a $k$-circle-layered graph, i.e. $V  = V_0 \uplus V_1 \dots \uplus V_{k-1}$ and $E \subseteq \{V_i \times V_{i + 1 \text{ mod } k} \mid i \in \rng{0}{k-1}\}$~\cite[Lemma~2.2]{LincolnWW18}.
We can also assume without loss of generality that every layer of $G$ contains the same number $n = 2^{h}$ of vertices for some $h \in \mathbb N$. 
We denote, for each $i \in \rng{0}{k-1}$, the vertices of $V_i$ by $v^i_0, v^i_1, \dots, v^i_{n-1}$. Let $m \coloneqq |E|$. 
To prove the statement, we describe two DFAs $\mc D_1$ and $\mc D_2$ using wildcard transitions whose intersection is non-empty if and only if $G$ has a $k$-Cycle, such that $\mc D_1$ and $\mc D_2$ have size $\Oh(n + m)$ and an alphabet of size $n$. 
By applying \cref{def:binary_gadget_identity,def:binary_gadget_tree}, we can reduce the alphabet size to $2$ (and thus remove wildcard transitions), while preserving the described equivalence. 
\cref{fig:kcycle_hardness} illustrates this construction.

Let $\Sigma \coloneqq \rng{0}{n-1}$.
We design $\mc D_1$ to recognise words that represent simple paths from $V_0$ to $V_0$ of length $k$ in $G$, and $\mc D_2$ to recognise words encoding sequences of vertices where the first and last vertex are the same. Then words accepted by both automata correspond to $k$-Cycles in $G$. 
Therefore, let $\mc D_1$ be a DFA over the alphabet $\Sigma$ that recognizes the language
\begin{align*}
	\mc L(\mc D_1) &= \{{j_0} {j_1} \dots {j_{k-1}} {j_k} \mid  (v^0_{j_0}, v^1_{j_1}), (v^1_{j_1}, v^2_{j_2}), \dots, (v^{k-1}_{j_{k-1}}, v^0_{j_k}) \in E\}
\intertext{
    and let $\mc D_2$ be a DFA over the alphabet $\Sigma$ that recognises the language
}
\mc L(\mc D_2) &= \{{j} \sigma^1 \dots \sigma^{k-1} {j} \mid  j, \sigma^1,
\sigma^2, \dots, \sigma^{k-1} \in \Sigma \}
\intertext{
    Precisely, let $\mc D_1 = (Q_1, \Sigma, \delta_1, \{x_I\}, F_1)$ with
}
    Q_1 & \coloneqq \{x_I \} \cup \{x^{i}_j \mid j \in \rng{0}{n-1}, i \in \rng{0}{k}\},\\
F_1 & \coloneqq \{x^k_j \mid j \in \rng{0}{n-1}\},\\
\delta_1 &\coloneqq \{(x_I, j, x^{0}_j) \mid j \in \rng{0}{n-1}\} \cup \{(x^{k-1}_{j'}, j, x^{k}_j)\mid (v^{k-1}_{j'}, v^{0}_j) \in E \} \\
&\phantom{\coloneqq } \cup \bigcup_{i \in \rng{0}{k-2} } \{(x^{i}_{j'}, j, x^{i+1}_j) \mid (v^{i}_{j'}, v^{i+1}_j) \in E\}
\intertext{
    and let $\mc D_2 \coloneqq (Q_2, \Sigma, \delta_2, \{y_I\}, F_2)$ with
}
    Q_2 & \coloneqq \{y_I \} \cup \{y^{i}_j \mid j \in \rng{0}{n-1}, i \in \rng{0}{k}\},\\
F_2 & \coloneqq \{y^k_j \mid j \in \rng{0}{n-1}\},\\
\delta_2 &\coloneqq \{(y_I, j, y^{0}_j) \mid j \in \rng{0}{n-1}\}  \cup \{(y^{k-1}_j, j, y^{k}_j)\mid j \in \rng{0}{n-1} \} \\
&\phantom{\coloneqq } \cup \bigcup_{i \in \rng{0}{k-2} } \{(y^{i}_j, \Sigma,
y^{i+1}_j) \mid j \in \rng{0}{n-1} \}.
\end{align*}
Note that $\Dd_2$ contains wildcard transitions.
The size of $\Dd_1$ and $\Dd_2$ is $\Oh(n+m)$. 
We now argue the correctness of the reduction, i.e. we show that there is a word $w \in \Sigma^*$ accepted by both $\mc D_1$ and $\mc D_2$ if and only if $G$ admits a $k$-Cycle.
 
\begin{figure*}[t]
    \centering
    \begin{subfigure}[t]{.3\textwidth}
        \resizebox{\textwidth}{!}{%
            \begin{tikzpicture}[thick,font=\large]
         
        \draw[fill=black] (0,0) circle[radius=2.5pt] node(C1)  {};
        \draw[fill=black] ($(C1) + (1,0)$) circle[radius=2.5pt] node(C2) {};
        
        \draw[fill=black] ($(C1) + (-1.5,2)$) circle[radius=2.5pt] node(A1) {};
        \draw[fill=black] ($(A1) + (0.5,0.5)$) circle[radius=2.5pt] node(A2) {};
        
        \draw[fill=black] ($(C2) + (1.5,2)$) circle[radius=2.5pt] node(B1) {};
        \draw[fill=black] ($(B1) + (-0.5,0.5)$) circle[radius=2.5pt] node(B2) {};
        
        \node[font=\Large,below=5pt] at (C1) {$c_1$};
        \node[font=\Large,below=5pt] at (C2) {$c_0$};
        \node[font=\Large,above left=2.5pt] at (A1) {$a_0$};
        \node[font=\Large,above left=2.5pt] at (A2) {$a_1$};
                
        \node[font=\Large,above right=2.5pt] at (B1) {$b_0$};
        \node[font=\Large,above right=2.5pt] at (B2) {$b_1$};

        \path[cb_blue] (A1) edge (B1);
        \path[cb_blue,line width=1.7pt] (A2) edge (B2);
        \path[cb_red] (B2) edge (C1);
        \path[cb_red,line width=1.7pt] (B2) edge (C2);
        \path[cb_green] (C1) edge (A1);
        \path[cb_green,line width=1.7pt] (C2) edge (A2);
         
        \node at (.5, 3.5) {\Large{$G$}};

        \begin{scope}[opacity=0]
            \draw[fill=black] (0,-1.5) circle[radius=2.5pt]  {};
        \end{scope}

        \end{tikzpicture}
        }
    \end{subfigure}
    \hfill
    \begin{subfigure}[t]{0.6\textwidth}
        \resizebox{\textwidth}{!}{%
        \begin{tikzpicture}[
            font=\large,xscale=.75,yscale=.8,
            state/.append style={minimum size=2.3em,inner sep=0pt},thick,>={Stealth[length=2.4mm, width=1.2mm]}, 
            align=center, node distance=1.5cm
            ]
            
            \node[state,initial,initial where = left,initial text=$ $] at (-1.5, 0) (xI) {$x_I$};
            \node[state] at (1, 1) (x00) {$x^0_0$};
            \node[state,right = of x00] (x01) {$x^1_0$};
            \node[state,right = of x01] (x12) {$x^2_0$};
            \node[state,right = of x12, accepting]  (x03) {$x^3_0$};
            \node[state] at (1, -1) (x10) {$x^0_1$};
            \node[state,right = of x10] (x11) {$x^1_1$};
            \node[state,right = of x11] (x02) {$x^2_1$};
            \node[state,right = of x02,accepting] (x13) {$x^3_1$};

            \path
            (xI) edge[->] node[above] {$0$} (x00)
            (xI) edge[->,line width=1.7pt] node[above] {$1$} (x10)
            (x00) edge[->,cb_blue] node[above] {$0$} (x01)
            (x11) edge[->,cb_red,line width=1.7pt] node[above] {$0$} (x12)
            (x11) edge[->,cb_red] node[above] {$1$} (x02)
            (x02) edge[->,cb_green] node[above left=10pt] {$1$} (x03)
            (x12) edge[->,cb_green,line width=1.7pt] node[above right=10pt] {$0$} (x13)
            (x10) edge[->,cb_blue,line width=1.7pt] node[above] {$1$} (x11)
            ;

            \node[state, initial,initial where = left,initial text=$ $] at (-1.5, -4)  (yI) {$y_I$};
            \node[state] at (1, -3) (y00) {$y^0_0$};
            \node[state,right = of y00] (y01) {$y^1_0$};
            \node[state,right = of y01] (y02) {$y^2_0$};
            \node[state,right = of y02,accepting] (y03) {$y^3_0$};
            \node[state] at (1, -5) (y10) {$y^0_1$};
            \node[state,right = of y10] (y11) {$y^1_1$};
            \node[state,right = of y11] (y12) {$y^2_1$};
            \node[state,right = of y12,accepting] (y13) {$y^3_1$};

            \path
            (yI) edge[->] node[above] {$0$} (y00)
            (yI) edge[->,line width=1.7pt] node[above] {$1$} (y10)
            (y00) edge[->] node[above] {$\Sigma$} (y01)
            (y10) edge[->,line width=1.7pt] node[above] {$\Sigma$} (y11)
            (y11) edge[->,line width=1.7pt] node[above] {$\Sigma$} (y12)
            (y01) edge[->] node[above] {$\Sigma$} (y02)
            (y02) edge[->] node[above] {$0$} (y03)
            (y12) edge[->,line width=1.7pt] node[above] {$1$} (y13)
            ;

            \node at (-3, 0) {$\mc D_1$};
            \node at (-3, -4) {$\mc D_2$};
        \end{tikzpicture}
        }
    \end{subfigure}
    \caption{
    Example of the reduction of $k$-\textsc{Cycle} to \twoie in \cref{thm:kcycle_2dfa}. The $3$-circle-layered directed graph $G$ is reduced to the DFA $\mc D_1$  and the DFA $\mc D_2$ over the alphabet $\Sigma = \{0, 1\}$. The cycle $(a_1, b_1, c_0)$ corresponds to the word $w = 1101$, which is recognised by both automata $\mc D_1$ and $\mc D_2$ via the runs 
    highlighted in bold. 
    }\label{fig:kcycle_hardness}
\end{figure*}

$(\impliedby)$
Assume that there is a $k$-Cycle on vertices $v^0_{j_0}, v^1_{j_1}, \dots, v^{k-1}_{j_{k-1}}$ in $G$. Then the word $w = {j_0} {j_1} \dots {j_{k-1}} j_0$ is recognised by $\mc D_1$ via the run visiting states $x_I, x^{0}_{j_0}, x^{1}_{j_1}, \dots, x^{k-1}_{j_{k-1}}, x^{k}_{j_0}$ and recognised by $\mc D_2$ via the run visiting states $y_I, y^0_{j_0}, y^1_{j_0}, \dots, y^{k-1}_{j_0}, y^k_{j_0}$. 

$(\implies)$
Conversely, assume that there is a word $w \in \mc L(\mc D_1) \cap \mc L(\mc D_2)$. 
Then, $w \in \mc L(\mc D_1) \cap \mc L(\mc D_2)$ if and only if $w = {j_0} {j_1} \dots {j_{k-1}} {j_0}$ such that $(v^0_{j_0}, v^1_{j_1}), (v^1_{j_1}, v^2_{j_2}), \dots, (v^{k-1}_{j_{k-1}}, v^0_{j_0})$ are edges of $G$. This means that there is a $k$-Cycle in $G$ on the vertices $v^0_{j_0}, v^1_{j_1}, \dots, v^{k-1}_{j_{k-1}}$.
\end{proof}

Now we prove quadratic hardness of $k$-IE under the $k$-OV hypothesis. As previously mentioned, this is strongly inspired by~Wehar's proof~\cite[Theorem 7.22]{WeharThesis16}.
Note that~\cite[Theorem 7.22]{WeharThesis16} is a corollary of~\cref{thm:ov_twoie} as SETH implies the $k$-OV hypothesis~\cite{Williams05}.

\begin{lemma}\label{thm:ov_twoie}
    For every $k \ge 2$ and $\eps > 0$, deciding non-emptiness of the intersection of $k$ DFAs of
    size $n$ does not admit an $\Oh(n^{k-\eps})$ time algorithm, assuming the $k$-\ov hypothesis.
\end{lemma}

\begin{proof}
Fix $k \ge 2$ and consider a $k$-OV instance $A_1, \dots, A_k \subset \{0, 1\}^d$ with $n = |A_1| = \dots = |A_k|$ and $d = \Theta(\log n)$. Denote the vectors in $A_i$ by $A_i = \{a^i_1, \dots, a^i_n\}$, for every $i \in \rng{1}{k}$.
To prove the statement, we construct $k$ DFAs $\mc D_1, \dots, \mc D_k$, each of size $\Oh(n \cdot k + n \cdot d)$, whose intersection is non-empty if and only if the answer to $k$-OV is positive. The constructed DFAs use an alphabet of size $\Oh(n +k)$ and wildcard transitions. By \cref{def:binary_gadget_tree,def:binary_gadget_identity}, we can reduce the alphabet to a binary one (and thus remove wildcard transitions). \cref{fig:ov_twoie} illustrates the construction.

Recall that the answer to $k$-OV is positive if and only if there are vectors $(a^1_{j_1}, \dots, a^k_{j_k}) \in A_1 \times \dots \times A_k$ such that for every coordinate $\ell \in \rng{1}{d}$ there is an index $r_\ell \in \rng{1}{k}$ with $a^{r_\ell}_{j_{r_\ell}}[\ell] = 0$. In that case, we say that $A_{r_\ell}$ is \emph{responsible} for the coordinate $\ell$, since, intuitively, it should provide a vector whose $\ell$-th coordinate is zero. 
The choice of vectors can be encoded by the indices $j_1, \dots, j_k \in \rng{1}{n}$, and the assignment of responsibilities can be encoded by the indices $r_1, \dots, r_d \in \rng{1}{k}$.
The idea is therefore to design $\mc D_i$ such that it recognises words of the shape $w = j_1 \dots j_k r_1 \dots r_d$ with $j_1, \dots, j_k \in \rng{1}{n}$, $r_1, \dots, r_d \in \rng{1}{k}$, and such that for every $\ell \in \rng{1}{d}$ if $r_\ell = i$ then $a^i_{j_i}[\ell] = 0$. Then a word $w = j_1 \dots j_k r_1 \dots r_d$ recognised by all DFAs $\mc D_1, \dots, \mc D_k$ witnesses the orthogonality of the vectors $a^{1}_{j_1}, a^{2}_{j_2}, \dots, a^{k}_{j_k}$. 

Formally, let $\Sigma_1 \coloneqq \rng{1}{n}$, $\Sigma_2 \coloneqq \rng{1}{k}$, and $\Sigma \coloneqq \Sigma_1 \cup \Sigma_2$. 
%
We define the following alphabet for every $i \in \rng{1}{k}$, $j \in \rng{1}{n}$ and $\ell \in \rng{1}{d}$:
\begin{align*}
    \Sigma(i, j, \ell) &\coloneqq \begin{cases}
        \Sigma_2 &\text{ if } a^i_j[\ell] = 0\\
        \Sigma_2 \setminus \{i\} &\text{ otherwise.}
    \end{cases}
\intertext{For any $i \in \rng{1}{k}$ and $j \in \rng{1}{n}$, consider the regular language }
    \mc L_{i, j} &\coloneq \Sigma(i, j, 1) \Sigma(i, j, 2) \dots \Sigma(i, j, d)
\intertext{that contains all words $w$ on the alphabet $\Sigma_2$ of length $d$ such that $w[\ell]$, the $\ell$-th letter of $w$, can be the symbol $i \in \Sigma_2$ if and only if $a^i_j[\ell] = 0$. 
Next, define for every $i \in \rng{1}{k}$ the DFA $\mc D_i$ over the alphabet $\Sigma = \Sigma_1 \cup \Sigma_2$ to recognise the regular language}
 \mc L(\mc D_i) &\coloneqq \{j_1 \dots j_k w \ \mid \ j_1, \dots, j_k \in \rng{1}{n} \text{ and } w \in \mc L_{i, j_i}\}.
\end{align*}
We claim that there are orthogonal vectors in $A_1, \dots, A_k$ if and only if $\mc D_1, \dots, \mc D_k$ intersect.
Suppose that there are $j_1, \dots, j_k \in \rng{1}{n}$ such that $a^1_{j_1}, \dots, a^k_{j_k} \in A_1 \times \dots \times A_k$ are orthogonal. In particular, this means that there are indices $r_1, \dots, r_d \in \rng{1}{k}$ such that for every $\ell \in \rng{1}{d}$ we have $a^{r_\ell}_{j_{r_\ell}}[\ell] = 0$.
Consider the word $w = r_1 \dots r_d$. By the above definitions, for all $\ell \in \rng{1}{d}$ we have $r_\ell \in \Sigma(i, j_{i}, \ell)$ where $i = r_\ell$. So for every $i \in \rng{1}{k}$, we have $w \in \mc L_{i, j_i}$ and so the word $w' = j_1 \dots j_k w $ belongs to the language $\mc L(\mc D_i)$.
This means that the intersection of $\mc D_1, \dots, \mc D_k$ contains the word $w'$. 
Conversely, assume that $\mc D_1, \dots, \mc D_k$ have a non-empty intersection and let $w \in \bigcap_{i=1}^k \mc L(\mc D_i)$. Then $w$ is of the form $w = j_1 \dots j_k w$ where for each $i \in \rng{1}{k}$ we have $w \in \mc L_{i, j_i}$. In particular, for every $\ell \in \rng{1}{d}$ we have $a^i_{j_i}[\ell] = 0$ for $i = w[\ell]$. This means that $a^1_{j_1}, \dots, a^k_{j_k}$ are orthogonal vectors.

It remains to show an efficient construction of $\mc D_i$.
See \cref{fig:ov_twoie} for an example.
\input{fig_ov_kie}
Fix $i \in \rng{1}{k}$ and define the following set of states and transitions that encode the choice of vectors.
\begin{align*}
    Q_i^{(1)} & \coloneq \{x_I^{i, i'} \ \mid \ i' \in \rng{0}{i-1}\} \cup \bigcup_{j=1}^{n}\{x_j^{i, i'} \ \mid \ i' \in \rng{i}{k}\} \\
    \delta_i^{(1)} &\coloneq \{(x_I^{i, i' -1}, \Sigma_1, x_I^{i, i'}) \ \mid \ i'\in \rng{1}{i-1}\} \\
    & \cup \{(x_I^{i, i-1}, j, x^{i, i}_j) \ \mid \ j \in \rng{1}{n}\} \\
    & \cup \{(x_j^{i, i' -1}, \Sigma_1, x_j^{i, i'}) \ \mid \ i'\in \rng{1}{i-1}, j \in \rng{1}{n}\} 
\end{align*}
Intuitively, the automaton $\mc D_i$ stays put for the first $i-1$ letters, branches on the vector chosen in $\mc A_i$ when reading the $i$-th letter, and then stays put until $k$ letters are read.
Now fix $j \in \rng{1}{n}$. For the branch corresponding to choosing the vector $a^i_j \in A_i$, we construct the following sets of states and transitions that simulate the choice of responsibilities.
For ease of description, let $y_j^{i, 0} := x_j^{i, k}$.
\begin{align*}
    Q_{i, j}^{(2)} \coloneqq {}& \bigcup_{j=1}^n\{y_j^{i, \ell} \mid \ell \in \rng{0}{d}\} \\
    \delta_{i, j}^{(2)} \coloneqq {}& \bigcup_{j=1}^n\{(y_j^{i, \ell - 1}, \Sigma_2, y_j^{i, \ell}) \mid \ell \in \rng{1}{d}, a^i_j[\ell] = 0\} \\
    &\!\!\cup \{(y_j^{i, \ell -1}, \Sigma_2 \setminus \{i\}, y_j^{i, \ell}) \mid \ell \in \rng{1}{d}, a^i_j[\ell] = 1\} 
\end{align*}
Intuitively, if the automaton reads the letter $i$ in position $\ell \in \rng{1}{d}$, then it ensures that the run is not accepted if $a^i_j[\ell] = 1$. Otherwise, the automaton accepts any letter at position $\ell$.
Finally, let $\mc D_i = (Q_i, \Sigma, \delta_i, {x_I^{i, 0}}, F_i)$ where $Q_i \coloneqq Q_i^{(1)} \cup \bigcup_{j=1}^n Q_{i, j}^{(2)}$, $\delta_i \coloneqq \delta_i^{(1)} \cup \bigcup_{j=1}^n \delta_{i, j}^{(2)}$ and $F_i \coloneqq \{y_j^{i, d} \mid j \in \rng{1}{n}\}$. This uses $\Oh(n \cdot k + d \cdot k)$ states and transitions, including wildcard transitions over the alphabets $\Sigma_1$ and $\Sigma_2$, and almost wildcard transitions over the alphabet $\Sigma_2$. 


Note that $\mc D_i$ can be constructed with $\Oh(n \cdot k + n \cdot d)$
states and transitions, including wildcard transitions over the alphabets $\Sigma_1$ and $\Sigma_2$, and almost wildcard transitions over the alphabet $\Sigma_2$. 
\end{proof}

%% file: fig_ov_kie.tex
\begin{figure*}[t]
    \centering
    \resizebox{\textwidth}{!}{%
    \begin{tikzpicture}[xscale=.75,yscale=.8,
        state/.append style={minimum size=2.3em,inner sep=0pt},thick,>={Stealth[length=2.4mm, width=1.2mm]}]

        \node[state,initial,initial where = left,initial text= $ $,initial distance=2em] (xI10)  {$x_I^{1, 0}$};
            \node[state] (x211) [right = of xI10] {$x_2^{1, 1}$};
            \node[state] (x111) [above=of x211] {$x_1^{1, 1}$};
            \node[state] (x311) [below=of x211] {$x_3^{1, 1}$};
            \node[state] (x112) [right = of x111] {$x_1^{1, 2}$};
            \node[state] (x212) [right = of x211] {$x_2^{1, 2}$};
            \node[state] (x312) [right = of x311] {$x_3^{1, 2}$};

            \node[state,initial,initial where = left,initial text= $ $,initial distance=2em] (xI20) [below = 5cm of xI10] {$x_I^{2, 0}$};
            \node[state] (xI21) [right = of xI20] {$x_I^{2, 1}$};
            \node[state] (x222) [right = of xI21] {$x_2^{2, 2}$};
            \node[state] (x122) [above = of x222] {$x_1^{2, 2}$};
            \node[state] (x322) [below = of x222] {$x_3^{2, 2}$};

            \draw (xI10) edge[->,cb_red] node[above] {$1$} (x111);
            \draw (xI10) edge[->,cb_red, line width=1.5pt] node[above] {$2$} (x211);
            \draw (xI10) edge[->,cb_red] node[above] {$3$} (x311);

            \draw (x111) edge[->,cb_blue] node[above] {$\Sigma_1$} (x112);
            \draw (x211) edge[->,cb_blue, line width=1.5pt] node[above] {$\Sigma_1$} (x212);
            \draw (x311) edge[->,cb_blue] node[above] {$\Sigma_1$} (x312);

            \draw (xI20) edge[->,cb_red, line width=1.5pt] node[above] {$\Sigma_1$} (xI21);

            \draw (xI21) edge[->,cb_blue, line width=1.5pt] node[above] {$1$} (x122);
            \draw (xI21) edge[->,cb_blue] node[above] {$2$} (x222);
            \draw (xI21) edge[->,cb_blue] node[above] {$3$} (x322);

            \node[state] (y111) [right = 1.5cm of x112] {$y_1^{1, 1}$};
            \node[state] (y211) [right = 1.5cm of x212] {$y_2^{1, 1}$};
            \node[state] (y311) [right = 1.5cm of x312] {$y_3^{1, 1}$};
            \node[state] (y112) [right = 1.5cm of y111] {$y_1^{1, 2}$};
            \node[state] (y212) [right = 1.5cm of y211] {$y_2^{1, 2}$};
            \node[state] (y312) [right = 1.5cm of y311] {$y_3^{1, 2}$};

            \node[state] (y113) [right = 1.5cm of y112] {$y_1^{1, 3}$};
            \node[state] (y213) [right = 1.5cm of y212] {$y_2^{1, 3}$};
            \node[state] (y313) [right = 1.5cm of y312] {$y_3^{1, 3}$};
            \node[state, accepting] (y114) [right = 1.5cm of y113] {$y_1^{1, 4}$};
            \node[state, accepting] (y214) [right = 1.5cm of y213] {$y_2^{1, 4}$};
            \node[state, accepting] (y314) [right = 1.5cm of y313] {$y_3^{1, 4}$};

            \node[state] (y121) [right = 1.5cm of x122] {$y_1^{2, 1}$};
            \node[state] (y221) [right = 1.5cm of x222] {$y_2^{2, 1}$};
            \node[state] (y321) [right = 1.5cm of x322] {$y_3^{2, 1}$};

            \node[state] (y122) [right = 1.5cm of y121] {$y_1^{2, 2}$};
            \node[state] (y222) [right = 1.5cm of y221] {$y_2^{2, 2}$};
            \node[state] (y322) [right = 1.5cm of y321] {$y_3^{2, 2}$};

            \node[state] (y123) [right = 1.5cm of y122] {$y_1^{2, 3}$};
            \node[state] (y223) [right = 1.5cm of y222] {$y_2^{2, 3}$};
            \node[state] (y323) [right = 1.5cm of y322] {$y_3^{2, 3}$};
            \node[state, accepting] (y124) [right = 1.5cm of y123] {$y_1^{2, 4}$};
            \node[state, accepting] (y224) [right = 1.5cm of y223] {$y_2^{2, 4}$};
            \node[state, accepting] (y324) [right = 1.5cm of y323] {$y_3^{2, 4}$};

            \draw (x112) edge[->] node[above] {$\Sigma_2$} (y111);
            \draw (x212) edge[->, line width=1.5pt] node[above] {\textbf{\textcolor{cb_green}{$\Sigma_2$}}} (y211);
            \draw (x312) edge[->] node[above] {$\Sigma_2 \setminus \{1\}$} (y311);
            
            \draw (y111) edge[->] node[above] {$\Sigma_2$} (y112);
            \draw (y211) edge[->, line width=1.5pt] node[above] {\textbf{\textcolor{cb_green}{$\Sigma_2 \setminus \{1\}$}}} (y212);
            \draw (y311) edge[->] node[above] {$\Sigma_2 \setminus \{1\}$} (y312);
        
            \draw (y112) edge[->] node[above] {$\Sigma_2 \setminus \{1\}$} (y113);
            \draw (y212) edge[->, line width=1.5pt] node[above] {\textbf{\textcolor{cb_green}{$\Sigma_2 \setminus \{1\}$}}} (y213);
            \draw (y312) edge[->] node[above] {$\Sigma_2 \setminus \{1\}$} (y313);
        
            \draw (y113) edge[->] node[above] {$\Sigma_2$} (y114);
            \draw (y213) edge[->, line width=1.5pt] node[above] {\textbf{\textcolor{cb_green}{$\Sigma_2$}}} (y214);
            \draw (y313) edge[->] node[above] {$\Sigma_2 \setminus \{1\}$} (y314);

            \draw (x122) edge[->, line width=1.5pt] node[above] {\textbf{\textcolor{cb_green}{$\Sigma_2$}}} (y121);
            \draw (x222) edge[->] node[above] {$\Sigma_2$} (y221);
            \draw (x322) edge[->] node[above] {$\Sigma_2 \setminus \{2\}$} (y321);

            \draw (y121) edge[->, line width=1.5pt] node[above] {\textbf{\textcolor{cb_green}{$\Sigma_2$}}} (y122);
            \draw (y221) edge[->] node[above] {$\Sigma_2 \setminus \{2\}$} (y222);
            \draw (y321) edge[->] node[above] {$\Sigma_2 \setminus \{2\}$} (y322);

            \draw (y122) edge[->, line width=1.5pt] node[above] {\textbf{\textcolor{cb_green}{$\Sigma_2$}}} (y123);
            \draw (y222) edge[->] node[above] {$\Sigma_2$} (y223);
            \draw (y322) edge[->] node[above] {$\Sigma_2 \setminus \{2\}$} (y323);
        
            \draw (y123) edge[->, line width=1.5pt] node[above] {\textbf{\textcolor{cb_green}{$\Sigma_2 \setminus \{2\}$}}} (y124);
            \draw (y223) edge[->] node[above] {$\Sigma_2 \setminus \{2\}$} (y224);
            \draw (y323) edge[->] node[above] {$\Sigma_2$} (y324);

            \node [right = .5cm of y114] {\large $a_1 = 0010$};
            \node [right = .5cm of y214] {\large $a_2 = \textcolor{cb_green}{0110}$};
            \node [right = .5cm of y314] {\large $a_3 = 1111$};
            \node [right = .5cm of y124] {\large $b_1 = \textcolor{cb_green}{0001}$};
            \node [right = .5cm of y224] {\large $b_2 = 0101$};
            \node [right = .5cm of y324] {\large $b_3 = 1110$};

            \node [above = 1.5cm of xI10] {\large $\mc D_1$};
            \node [above = 1.5cm of xI20] {\large $\mc D_2$};
    \end{tikzpicture}
    }
    \caption{
        Example of the construction in \cref{thm:ov_twoie} of the DFAs $\mc D_1$ and $\mc D_2$. 
        In the 2-OV problem we want to check whether $\textcolor{cb_red}{\exists a} \in A, \textcolor{cb_blue}{\exists b} \in B$ such that $\textcolor{cb_green}{\forall \ell} \in [d], a[\ell] = 0 \text{ or } b[\ell] = 0$. In our construction, each ``branch'' of the DFAs corresponds to a vector of the 2-OV instance $A \coloneqq \{a_1, a_2, a_3\}$ and $B = \{b_1, b_2, b_3\}$. The parameters are $k=2$, $n = 3$ and $d = 4$, and thus $\Sigma_1 = \{1, 2, 3\}$ and $\Sigma_2 = \{1, 2\}$.
        In the above example, the vectors $a_2 = 0110$ and $b_1 = 0001$ are orthogonal, which is witnessed by the two accepting runs on $w = \textcolor{cb_red}{2} \textcolor{cb_blue}{1} \textcolor{cb_green}{1221}$ highlighted in bold.
    }
    \label{fig:ov_twoie}
\end{figure*}

%% file: appendix_UU.tex
\section{Missing proofs of Section~\ref{sec:technical_overview}}\label{sec:proofs_Sec2}

\begin{proof}[Proof of \cref{lem:sparsity}]
	Recall that $n = |V(G)|$ and $m = |E(G)|$. By \cref{lem:transformation}, if $(G, s, t) \in \unaryunambiguity$ then there is a graph $G'$ with the properties listed in \cref{lem:transformation}. In particular, $|V(G')| \leq 3n$, $|E(G')| \geq m$ and all cycles in $G'$ are disjoint.
	Let $H$ be the graph induced by the strongly connected components of $G'$. Notice that $H$ is the DAG corresponding to removing all cycles in $G'$ except for the cycle gates. Since the cycles of $G'$ are disjoint, we removed at most $|V(G')| \leq 3n$ edges. So we have $|V(H)| \leq |V(G')| \leq 3n$ and $|E(H)| \geq |E(G')| - |V(G')| \geq m - 3n$.
	Consider the subgraph $H'$ of $H$ constructed as follows:
	\begin{description}
        \item[\textbf{Step 1}.] \;\,Initially, let $H'$ be an arbitrary $st$-walk in $H$. 
        \item[\textbf{Step 2}.] \;\,While there is a vertex $v \in V(H) \setminus V(H')$, let $P$ be a simple path from a vertex in $V(H')$ to $v$ that uses only vertices of $V(H) \setminus V(H')$ and let $Q$ be a simple path from $v$ to a vertex in $V(H')$ that uses only vertices of $V(H) \setminus V(H')$ and that is disjoint from $P$ (except for the endpoints).
        Add all vertices and edges of $P$ and $Q$ to $H'$. 
	\end{description}
    Note that we can always find simple paths $P$ and $Q$ in Step 2. If they intersect at some vertex $u \in V(H) \setminus (V(H') \cup \{v\})$, then repeat Step 2 with $u$ instead of $v$. 

	Observe that we have $V(H') = V(H)$ at the end of the construction.  
    We bound the number of edges in $H'$. 
	After Step 1, we have $|E(H')| - |V(H')| = -1$. In every iteration of Step 2, the number of vertices of $H'$ increases by $|P| + |Q| - 1$, while the number of edges of $H'$ increases by $|P| + |Q|$. 
    Thus, $|E(H')| - |V(H')|$ increases by 1 in each iteration of Step 2. Since Step 2 is iterated at most $|V(H)|- 2 \leq 3n -2$ times, we get that in the end of the construction $|E(H')| \leq 3n - 3 + |V(H')| \leq 6n - 3$.
	
	Furthermore, by construction, every vertex is both reachable from $s$ and co-reachable from $t$ in the graph $H'$. 
    Hence, introducing any edge in $E(H) \setminus E(H')$ into $H'$ increases the number of $st$-walks in $H'$ by at least 1. 
	Observe that since $H$ is a trim DAG, there are at most $|V(H)| \leq 3n$ $st$-walks in $H$, or otherwise it would not be unambiguous. In particular, $|E(H) \setminus E(H')| \leq 3n$ and thus $|E(H)| = |E(H')| + |E(H) \setminus E(H')| \leq 9n - 3$. Since $|E(H)| \geq m - 3n$, this implies that 
	$m \leq 12n$.
\end{proof}

\begin{proof}[Proof of \cref{claim:gcd2}]
Let $p \coloneq \gcd(d, b)$.
Assume that $(a + x \cdot b)_{x \in \mathbb N} \cap (c + x \cdot d)_{x \in \mathbb N} \neq \emptyset$. Then we have $a - c = d \cdot n_2 -  b\cdot n_1$ for some $n_1,n_2 \in \N$. The claim follows from $b \equiv 0 \pmod{p}$ and $d \equiv 0 \pmod{p}$.

Assume now that $a \equiv c \pmod{p}$ and, without loss of generality, $a > c$. 
We then have $a - c = z \cdot p$ for some $z \in \N$. 
By B\'{e}zout's identity there are $x,y \in \Z$ such that $x \cdot b + y \cdot d = p$. 
We obtain
$
 a - c = z \cdot p = (x \cdot z)\cdot b + (y \cdot z) \cdot d,
 $
from which follows $a + (-x \cdot z)\cdot b = c + (y \cdot z)\cdot d$. Adding $k=|xz|+|yz|$ times $b\cdot d$ to both sides yields 
$a + (-x \cdot z+k\cdot d)\cdot b = c + (y \cdot z+k\cdot b)\cdot d$ with $(-x \cdot z+k\cdot d), (y \cdot z+k\cdot b)\in \N$, as required.
\end{proof}

\begin{proof}[Proof of \cref{lemma:sumofdiv}]
	For any natural $n$, let $n = p_1^{\alpha_1} \cdot \ldots \cdot p_s^{\alpha_s}$ be the prime factorisation of $n$, where $2 \le p_1 < p_2 \ldots < p_s$ are pairwise different prime numbers, and $\alpha_1, \ldots, \alpha_s \in \N$. Then
	\begin{align*}
		\sigma_1(n) & = \sum_{d \mid n} d   = \prod_{i=1}^s \left(p_i^0 + p_i^1 + \ldots
        +p_i^{\alpha_i}\right) \\& = \prod_{i=1}^s \frac{p_i^{\alpha_i + 1} - 1}{p_i -
    1} < \prod_{i=1}^s \frac{p_i^{\alpha_i + 1}}{p_i - 1} \\
& = \prod_{i=1}^s p_i^{\alpha_i} \cdot \left(1 + \frac{1}{p_i - 1}\right) \\
     & = n \cdot \prod_{i=1}^s \left(1 + \frac{1}{p_i - 1}\right).
	\end{align*}
	Observe that $1 + \frac{1}{p_i - 1} \le 1 + \frac{1}{p_{i-1}}$ for $i > 1$. For
	$i = 1$ we have $1 + \frac{1}{p_i - 1} \le 2$ (as the maximum is obtained for
$p_1=2)$. Thus we have \[\prod_{i=1}^s \left(1 + \frac{1}{p_i - 1}\right) \le 2 \cdot \prod_{i=1}^{s-1} \left(1 + \frac{1}{p_i}\right).\] This product expands to a sum of
elements of the form $1/\prod_{i \in I} p_i$ for sets $I \subseteq \rng{1}{s - 1}$. The
	indices in the denominator range over different sets, thus they are pairwise
	different natural numbers bounded by $n$. Hence, we get the desired bound as follows.
	\begin{align*}
    \sigma_1(n) & \le 2n \cdot \prod_{i=1}^{s-1} \left(1 +
    \frac{1}{p_i}\right) \le 2n\cdot \sum_{i=1}^n \frac{1}{i} = \Oh(n \log n).
	\end{align*}
\end{proof}

\paragraph{Bounding the sum of GCDs}
We prove \cref{lem:sum-gcd}. See \cref{sec:technical_overview_unary,rem:galsums} for a discussion about previously known bounds. 

\sumgcd*

Recall the standard \emph{divisor function} $\sigma_0$. This function maps any natural $n$ to the number of positive integer divisors of $n$. In other words, for every positive $n \in \N$:
\begin{displaymath}
    \sigma_0(n) \coloneqq \sum_{d=1}^n \iverson{d \mid n}.
\end{displaymath}
The divisor function can be bounded as follows.\footnote{For a brief and easily accessible proof we refer to Tao's blog~\cite{TTaoBlog}.}
\begin{lemma}[Chapter 18.1 in~\cite{hardy1979introduction}]\label{fact:divisor-function}
    The divisor function can be bounded from above by $\sigma_0(n) \le n^{o(1)}$.
\end{lemma}
We use this to prove another intermediate inequality.
\begin{lemma}\label{claim:gcd}
    For any positive $n \in \N$ it holds that:
    \begin{displaymath}
        \sum_{i =1}^n \gcd(n,i) \le n^{1+o(1)}.
    \end{displaymath}
\end{lemma}
\begin{proof}
  Recall the Iverson bracket notation from \cref{sec:preliminaries} for the characteristic function. Clearly, we have:
    \begin{align*}
        \sum\nolimits_{i=1}^n \gcd(n,i) & = \sum\nolimits_{i=1}^n \; \sum\nolimits_{d=1}^n d
        \cdot \iverson{d = \gcd(n,i)} \\& \le
        \sum\nolimits_{i=1}^n \sum\nolimits_{d=1}^n d \cdot \iverson{d \divides n} \cdot \iverson{d \mid i}.
        \intertext{Next, we swap the sums and obtain}
        \sum\nolimits_{i=1}^n \gcd(n,i) & \le \sum\nolimits_{d=1}^n d \cdot \iverson{d \divides n} \cdot \left( \sum\nolimits_{i=1}^n \iverson{d \divides i} \right).
        \intertext{Notice that $\sum\nolimits_{i=1}^n \iverson{d \divides i} = \floor{n/d}$, hence we have }
        \sum\nolimits_{i=1}^n \gcd(n,i) & \le \sum_{d=1}^n d \cdot \iverson{d \divides n} \cdot (n/d) =
        n \sum\nolimits_{d=1}^n \iverson{d \divides n}  = n \cdot \sigma_0(n).
    \end{align*}
    By~\cref{fact:divisor-function}, we conclude that $\sigma_0(n) \le n^{o(1)}$.
\end{proof}

\begin{proof}[Proof of~\cref{lem:sum-gcd}]
    Adding additional terms yields:
    \begin{align*}
        \sum_{a,b \in A} \gcd(a,b) & \le 2 \sum_{\substack{a,b \in A\\b\le a}} \gcd(a,b)
        \le 2 \sum_{a \in A} \sum_{i=1}^a \gcd(a,i).
        \intertext{By~\cref{claim:gcd}, we can bound this further by}
        \sum_{a,b \in A} \gcd(a,b) & \le 2 \sum_{a \in A} a^{1+o(1)} \le 2 \, N^{1+o(1)},
    \end{align*}
    where the last step holds by Jensen's inequality and the fact that the sum
    of the integers in $A$ is $N$. Note that the constant $2$ can be absorbed into
    $N^{o(1)}$.
\end{proof}



%% file: wildcard.tex
\section{Missing proofs of Section~\ref{sec:lower}}\label{sec:appendix_lowerbounds}


\begin{proof}[Proof of \cref{def:binary_gadget_tree}]  
    \begin{figure*}
        \centering
        \begin{subfigure}[b]{0.7\textwidth}
        \includesvg{hardness_fig/binary_decompo_tree.svg}
        \end{subfigure}
        \caption{
        To reduce the alphabet size, in \cref{def:binary_gadget_tree}, we replace every transition outgoing from $x$ (left) by a tree-like structure (right) with $\Oh(\log |\Sigma|)$ states and transitions encoding every label read. 
        In the above example, $\Sigma = \{0, 1, \dots, 7\}$.
        }\label{fig:gadget_identity}
    \end{figure*}
        Consider a DFA $\mc A = (Q, \Sigma, \delta, \{q_I\}, F)$ with alphabet $\Sigma =
        \rng{0}{2^h-1}$ for some $h \in \mathbb N$.  For any state $x\in Q$ of
        $\mc A$ by $\delta^+(x)$ denote the subset of transitions in $\delta$ outgoing from $x$.  For
    any $\sigma \in \Sigma$ and $k \in \rng{1}{h}$, let $s(\sigma, k)
        \coloneqq \left\lfloor {\sigma}/{2^{h - k}} \right\rfloor$ and $b(\sigma, k)
        \coloneqq
        s(\sigma, k) \bmod 2$; that is, $b(\sigma, k)$ is the $k$-th bit of $\sigma$'s binary representation. 
        The idea is to encode every letter in $\Sigma$ with $h$ bits. 
    
        For every $(x, \sigma, y) \in \delta^+(x)$, define the following set of states and transitions that decompose $(x, \sigma, y)$ into $h$ transitions reading the bits of $\sigma$.
        \begin{align*}
            Q({x, \sigma, y}) &\coloneqq  \{x^i_{s(\sigma, i)} \mid i \in \rng{1}{h-1}\}, \\
            \delta({x, \sigma, y}) &\coloneqq 
             \{(x, b(\sigma, 1), x^1_{s(\sigma, 1)}), (x^{h-1}_{s(\sigma, h-1)}, b(\sigma, h), y)\}  \\
             & \bigcup_{i \in \rng{1}{h-2}} \{(x^i_{s(\sigma, i)}, b(\sigma, i+1), x^{i+1}_{s(\sigma, i+1)})\}. \\
        \end{align*}
        \Cref{fig:gadget_identity} illustrates this construction.
        By taking the union of the states and transitions above over all outgoing transitions $\delta^+(x)$, we obtain a binary-tree-like structure that simulates the behaviour of $\delta^+(x)$ with a binary encoding of the read symbols: 
        \begin{align*}
            Q_\text{BT}(x) \coloneqq & \bigcup_{(x, \sigma, y) \in \delta^+(x)} Q({x, \sigma, y}) \quad \text{ and } \\
            \delta_\text{BT}(x) \coloneqq & \bigcup_{(x, \sigma, y) \in \delta^+(x)} \delta({x, \sigma, y}).\\
        \end{align*}
        Now, consider the DFA $\mc A' \coloneqq (Q', \{0, 1\}, \delta', \{q_I\}, F)$, where
        $Q' \coloneqq Q \cup \bigcup_{x \in Q} Q_\text{BT}(x)$ and $\delta' \coloneqq \bigcup_{x \in Q} \delta_\text{BT}(x)$. Then the language recognised by $\mc A'$ is the same as $\mc A$, up to the binary encoding of every letter of $\Sigma$. 
        By applying the above transformation to each of the DFAs $\mc D_1, \dots, \mc D_k$ of a \kie instance, we obtain an equivalent instance $\mc D_1', \dots, \mc D_k'$ over the binary alphabet. 
        Finally, note that $|Q_\text{BT}(x)| = \Oh(h \cdot |\delta^+(x)|)$ and $|\delta_\text{BT}(x)| = \Oh(h \cdot |\delta^+(x)|)$ for every $x \in Q$, so the size of each DFA has increased by at most $\Oh(h) = \Oh(\log |\Sigma|)$ multiplicative factor.
\end{proof}
    
\begin{proof}[Proof of \cref{def:binary_gadget_identity}]
	Consider a DFA $\mc A = (Q, \Sigma, \delta, \{q_I\}, F)$ with alphabet $\Sigma = \rng{0}{2^h-1}$ for some $h \in \mathbb N$. Partition the transitions $\delta$ into the wildcard transitions $\delta_\text{WC}$, the almost wildcard transitions $\delta_\text{AWC}$, and the remaining regular transitions $\delta_\text{RG}$. 
    
    Let $x$ and $y$ be states of $\mc A$ with a \emph{wildcard} transition $(x, \Sigma, y)$.
    This means that regardless of what symbol in $\Sigma$ is read, after visiting state $x$, any accepting run visits state $y$ next. 
	Define the following set of states and transitions that simulate the wildcard transition: they ensure that any accepting run visiting state $x$ visits state $y$ after reading any sequence of $h$ bits. 
	\begin{align*}
        Q_\text{WC}(x, y) \coloneqq & \{q^{x, y}_i\mid i \in \rng{0}{h-1}\} \\
		\delta_\text{WC}(x, y) \coloneqq
		& \{(x, b, q^{x, y}_0) \mid b \in \{0, 1\}\}\\
        & \cup \!\!\bigcup_{\substack{b \in \{0, 1\} }}\!\! \{(q^{x, y}_i, b, q^{x, y}_{i+1}) \mid i \in \rng{0}{h-2}\} \\
        & \cup \{(q^{x, y}_{h-1}, b, y) \mid b \in \{0, 1\}\}
	\end{align*}
    \cref{fig:gadget_wildcard} illustrates this construction.

    \begin{figure*}
        \centering
        \begin{subfigure}[t]{.48\textwidth}
            \centering
            \resizebox{\textwidth}{!}{%
            \begin{tikzpicture}[
                state/.append style={minimum size=2.3em,inner sep=0pt},thick,>={Stealth[length=2.4mm, width=1.2mm]}]
                \node[state] (x) {$x$};
                \node[state,right = 1.5cm of x] (y) {$y$};
                
                \path
                (x) edge[->] node[above] {$\Sigma$} (y)
                ;

                \node[state,below = of x] (x) {$x$};
                \node[state,right = of x] (q1) {$q^{x, y}_1$};
                \node[state,right = of q1] (q2) {$q^{x, y}_2$};
                \node[state,right = of q2] (q3) {$q^{x, y}_3$};
                \node[state,right = of q3] (y) {$y$};

                \node[state,below = of q1,opacity=0] (p1) {$p^{x, y}_1$};

                \path
                (x) edge[->, bend left] node[above] {$0$} (q1)
                (q1) edge[->, bend left] node[above] {$0$} (q2)
                (q2) edge[->, bend left] node[above] {$0$} (q3)
                (q3) edge[->, bend left] node[above] {$0$} (y)
                (x) edge[->, bend right] node[below] {$1$} (q1)
                (q1) edge[->, bend right] node[below] {$1$} (q2)
                (q2) edge[->, bend right] node[below] {$1$} (q3)
                (q3) edge[->, bend right] node[below] {$1$} (y)
                ;
            \end{tikzpicture}}
            \caption{The wildcard transition $(x, \Sigma, y)$ (top) for the alphabet $\Sigma = \{0, 1, \dots, 15\}$ is simulated by 8 transitions over the binary alphabet (bottom). All runs starting at $x$ and reading 4 bits end in state $y$.}
            \label{fig:wildcard}
        \end{subfigure}
        \hfill
        \begin{subfigure}[t]{.48\textwidth}
            \centering
            \resizebox{\textwidth}{!}{%
            \begin{tikzpicture}[
                state/.append style={minimum size=2.3em,inner sep=0pt},thick,>={Stealth[length=2.4mm, width=1.2mm]}]
                \node[state] (x) {$x$};
                \node[state,right = 1.5cm of x] (y) {$y$};

                \path
                (x) edge[->] node[above] {$\Sigma \setminus \{9\}$} (y)
                ;
                
                \node[state,below = of x] (x) {$x$};
                \node[state,right = of x] (q1) {$q^{x, y}_1$};
                \node[state,right = of q1] (q2) {$q^{x, y}_2$};
                \node[state,right = of q2] (q3) {$q^{x, y}_3$};
                \node[state,right = of q3] (y) {$y$};
                \node[state,below = of q1] (p1) {$p^{x, y}_1$};
                \node[state,right = of p1] (p2) {$p^{x, y}_2$};
                \node[state,right = of p2] (p3) {$p^{x, y}_3$};
                
                \path
                (x) edge[->, bend left] node[above] {$0$} (q1)
                (q1) edge[->, bend left] node[above] {$0$} (q2)
                (q2) edge[->, bend left] node[above] {$0$} (q3)
                (q3) edge[->, bend left] node[above] {$0$} (y)

                (x) edge[->, bend left] node[above] {$0$} (q1)
                (q1) edge[->, bend left] node[above] {$0$} (q2)
                (q2) edge[->, bend left] node[above] {$0$} (q3)
                (q3) edge[->, bend left] node[above] {$0$} (y)
                (q1) edge[->, bend right] node[above] {$1$} (q2)
                (q2) edge[->, bend right] node[above] {$1$} (q3)
                (q3) edge[->, bend right] node[above] {$1$} (y)

                (x) edge[->] node[below left] {$1$} (p1)
                (p1) edge[->] node[below] {$0$} (p2)
                (p1) edge[->] node[below right] {$1$} (q2)
                (p2) edge[->] node[below] {$0$} (p3)
                (p2) edge[->] node[below right] {$1$} (q3)
                (p3) edge[->] node[below  right] {$0$} (y)
                ;
            \end{tikzpicture}}
        \caption{The almost wildcard transition $(x, \Sigma \setminus \{9\}, y)$ (top) for the alphabet $\Sigma = \{0, 1, \dots, 15\}$ is simulated by 13 transitions over the binary alphabet (bottom). The only run starting at $x$ and reading 4 bits that does not end in $y$ reads the binary encoding 1001 of 9.}
        \label{fig:almost_wildcard}
        \end{subfigure}
        \caption{Reducing the alphabet size for (almost) wildcard transitions in \cref{def:binary_gadget_identity}.
            }
            \label{fig:gadget_wildcard}
    \end{figure*}

    Let $x$ and $y$ be states of $\mc A$ with an \emph{almost wildcard} transition $(x, \Sigma \setminus \{\sigma\}, y)$, for some letter $\sigma \in \Sigma$. Let $b(\sigma, i)$ be the $i$-th bit of $\sigma$ when encoding over $h$ bits, for $i \in \rng{1}{h}$.  
	Define the following set of states and transitions that simulate the almost wildcard transition: they ensure that any accepting run visiting state $x$ visits state $y$ after reading any combination of $h$ bits, except for the combination of bits encoding $\sigma$.
	\begin{align*}
        Q_\text{AWC}(x, \sigma, y) \coloneqq & \{q^{x, y}_i\mid i \in \rng{0}{h-1}\} \\
        &\cup \{p^{x, y}_i\mid i \in \rng{1}{h-1}\} \\
		\delta_\text{AWC}(x, \sigma, y) \coloneqq
		& \{(x, b(\sigma, 1), p^{x, y}_0), (x, 1-b(\sigma, 1), q^{x, y}_0)\} \\
        & \cup \bigcup_{\substack{b \in \{0, 1\} }} \{(q^{x, y}_i, b, q^{x, y}_{i+1}) \mid i \in \rng{0}{h-2}\} \\
        & \cup \{(p^{x, y}_{i-1}, b(\sigma, i), p^{x, y}_{i}) \mid i \in \rng{2}{h-1}\} \\
        & \cup \{(p^{x, y}_{h-1}, 1 - b(\sigma, h-1), y)\} \\
        &\cup \{(q^{x, y}_{h-1}, b(\sigma, h-1), y) \}
	\end{align*}
    \cref{fig:gadget_wildcard} illustrates this construction.
    
	Now, consider the DFA $\mc A' \coloneqq (Q', \{0, 1\}, \delta', q_I, F)$, where
    \begin{align*}
        Q' \coloneq & Q \cup \bigcup_{x \in Q} Q_\text{BT}(x)  \cup \!\!\bigcup_{(x, \Sigma, y) \in \delta_\text{WC}}\!\! Q_\text{WC}(x, y) \\
        & \cup \!\!\bigcup_{(x, \Sigma \setminus \{\sigma\}, y) \in \delta_\text{AWC}}\!\! Q_\text{AWC}(x, \sigma, y) 
    \end{align*}
    and
    \begin{align*}
    \delta' \coloneq & \bigcup_{x \in Q} \delta_\text{BT}(x)  \cup \!\!\bigcup_{(x, \Sigma, y) \in \delta_\text{WC}}\!\! \delta_\text{WC}(x, y) \\
    & \cup \!\!\bigcup_{(x, \Sigma \setminus \{\sigma\}, y) \in \delta_\text{AWC}}\!\! \delta_\text{AWC}(x, \sigma, y) 
    \end{align*}
    where $Q_\text{BT}$ and $\delta_\text{BT}$ are constructed in \cref{def:binary_gadget_tree} and are defined over the regular transitions $\delta_\text{RG}$.
    Then the language recognised by $\mc A'$ is the same as $\mc A$, up to binary encoding of every letter of $\Sigma$. 
    By applying the above transformation to each of the DFAs $\mc D_1, \dots, \mc D_k$ of a \kie instance, we obtain an equivalent instance $\mc D_1', \dots, \mc D_k'$ over the binary alphabet. 
	Finally, note that $|Q_\text{WC}(x, y)| = \Oh(h)$ and $|\delta_\text{WC}(x, y)| = \Oh(h)$ for every wildcard transition $(x, \Sigma, y)$,  thus the transformation has increased the size of each DFA at most by a factor of $\Oh(h) = \Oh(\log |\Sigma|)$.
\end{proof}